\documentclass[letterpaper,12pt]{article}
\pdfoutput=1

\usepackage{amssymb,amsmath,amsthm,amsfonts}
\usepackage{graphicx, xcolor, varwidth}
\usepackage[percent]{overpic}
\usepackage{mathtools}
\usepackage{setspace}
\usepackage{cite}
\usepackage{tikz}
\usepackage{enumitem}

\usetikzlibrary{math}
\usetikzlibrary{shapes.geometric}
\usetikzlibrary{shapes.misc}
\tikzstyle{none}=[inner sep=0pt]
\tikzstyle{simplethree}=[circle,fill=black,draw=black,line width=1.500,minimum size = 2pt,inner sep = 2pt]
\tikzstyle{simple}=[circle,fill=black,draw=black,line width=1.000,minimum size = 2pt,inner sep = 2pt]
\tikzstyle{simpletoo}=[circle,fill=black,draw=black,line width=1.000,minimum size = 1pt,inner sep = 1pt]
\tikzstyle{simpled}=[circle,fill=black,draw=black,dashed,line width=1.000,minimum size = 1pt,inner sep = 1pt]
\tikzstyle{simplesq}=[rectangle,fill=black,draw=black,line width=1.000,minimum size = 1pt,inner sep = 1.5pt]
\tikzstyle{simplestar}=[star,fill=black,draw=black,line width=1.000,minimum size = 2pt,inner sep = 1.5pt]
\tikzstyle{simplecirc}=[circle,fill=black,draw=black,line width=1.000,minimum size = 2pt,inner sep = 1.5pt]
\tikzstyle{simpletrapez}=[trapezium,fill=black,draw=black,line width=1.000,minimum size = 2pt,inner sep = 1.5pt]
\tikzstyle{simplesemicircle}=[semicircle,fill=black,draw=black,line width=1.000,minimum size = 2pt,inner sep = 1.5pt]
\tikzstyle{simplecrossout}=[cross out,fill=black,draw=black,line width=1.000,minimum size = 2pt,inner sep = 1.5pt]
\tikzstyle{simplekite}=[kite,fill=black,draw=black,line width=1.000,minimum size = 2pt,inner sep = 1.5pt]

\usepackage{xcolor}
\usepackage[dotinlabels]{titletoc}
\usepackage[colorlinks=true, urlcolor=dblue, linktocpage=true, linkcolor=dblue,citecolor=purple]{hyperref}

\newtheorem{defn}{Definition}

\newtheorem{thm}{Theorem}
\newtheorem{lemma}{Lemma}
\newtheorem{remark}{Remark}

\newcommand{\mrm}{\mathrm}

\newcommand{\nn}{\nonumber}
\newcommand{\pd}{\partial}
\newcommand{\lb}{\left(}
\newcommand{\rb}{\right)}
\newcommand{\lcb}{\left\{}

\newcommand{\DD}{\mathcal{D}}

\newcommand{\GG}{\mathcal{G}}

\newcommand{\LL}{\mathcal{L}}
\newcommand{\MM}{\mathcal{M}}

\newcommand{\OO}{\mathcal{O}}
\newcommand{\PP}{\mathcal{P}}

\newcommand{\UU}{\mathcal{U}}

\newcommand{\rcb}{\right\}}
\newcommand{\lsb}{\left[}
\newcommand{\rsb}{\right]}

\newcommand\be{\begin{equation}}
\newcommand\ba{\begin{eqnarray}}
\newcommand\ee{\end{equation}}
\newcommand\ea{\end{eqnarray}}

\newcommand\otoz{\omega_{\mu\to0}}

\numberwithin{equation}{section}

\colorlet{dblue}{blue!70!black}

\newcommand{\arxivold}[1]
  {\href{http://arxiv.org/abs/#1}{#1}}
\newcommand{\arxiv}[1]
  {\href{http://arxiv.org/abs/#1}{arXiv:#1}}

\newcommand{\ipj}{{\langle i j \rangle}}

\usepackage{graphicx}
\usepackage{amsmath}

\begin{document}
\begin{spacing}{1.3}
\begin{titlepage}

\begin{center}
\begin{spacing}{2}
{\LARGE{Massless $p2$-brane modes and the critical line}}
\end{spacing}
% Massless $p2$-brane spectrum and the critical Riemann zeta zeros

\vspace{0.5cm}
An Huang,$^{1}$ Bogdan Stoica,$^{2}$ and Xiao Zhong$^{3}$

\vspace{5mm}

{\small

\textit{
$^1$Department of Mathematics, Brandeis University, Waltham, MA 02453, USA}\\

\vspace{3mm}

\textit{$^2$Department of Physics \& Astronomy, Northwestern University, Evanston, IL 60208, USA}\\

\vspace{3mm}

\textit{$^3$Department of Pure Mathematics, University of Waterloo, Waterloo, Ontario, N2L 3G1, Canada}

\vspace{5mm}

{\tt anhuang@brandeis.edu, bstoica@northwestern.edu, xiao.zhong@uwaterloo.ca}

\vspace{1cm}
}

\end{center}

\begin{abstract}
\noindent We consider a $p2$-brane model as a theory of maps from the vertices of the Bruhat-Tits tree times $\mathbb{Z}$ into $\mathbb{R}^d$. We show that in order for the worldsheet time evolution to be unitary, a certain spectral parameter of the model must localize at special loci in the complex plane, which include the (0,1) interval and the critical line of real part $1/2$. The excitations of the model are supported at these loci, with commutation relations closely resembling those of the usual bosonic string. We show that the usual Hamiltonian, momentum, and angular momentum are conserved quantities, and the Poincar\'e algebra is obeyed. Assuming an Euler product relation for the spectrum, the Riemann zeta zeros on the critical line have spectral interpretation as the massless photon and graviton in the Archimedean~theory.
\end{abstract}

\vfill

\begin{flushleft}{\small nuhep-th/21-09}\end{flushleft}

\end{titlepage}

$ $
\vspace{-2cm}
\setcounter{tocdepth}{2}
\tableofcontents

\section{Introduction}

\noindent Analogues to the bosonic string have been investigated since the inception of $p$-adic physics. While various features of the strings have been considered in a $p$-adic setting (such as scattering amplitudes \cite{FreundOlson,FreundWitten,BFOW,adelicNpoint}, or a worldsheet description \cite{zabrodin,zabrodin2}), so far an analogue of the usual bosonic quantization treatment has been missing. In this paper we will address this, by introducing a $p$-adic system with oscillation modes closely resembling the usual bosonic oscillators of string theory.

The motivation for considering such $p$-adic models is to explore the landscape of target space theories which can be obtained starting from $p$-adic worldsheets. Because the worldsheet is non-Archimedean, these theories can be in regimes which are not easily accessible from conventional descriptions of physics. However, as we will explain in Section \ref{secconc}, these theories could still be  part of the lanscape of vacua in string theory and M-theory. Thus, the $p$-adic approach can provide a handle on these points in theory space. We should also emphasize that despite the $p$-adic nature of the worldsheet, the field content of the target space theory is Archimedean.

There is another motivation for our work, which may not be immediately apparent. It turns out that the types of models we will consider here are intimately related to the Riemann zeta function, and to the irreducible unitary representations of $\mrm{PGL}(2,\mathbb{Q}_p)$. In fact, the scalar field dynamics, zeta function and irreducible representations combine in a natural way into the \emph{physics}: the physical content of the models contains certain subspaces of the representations, which evolve according to scalar field equations of motion, with the zeta function encoding information on the spectrum of the theory.

In the case of the usual Archimedean critical bosonic string theory, demanding that the target theory contain a massless $U(1)$ gauge field and graviton fixes the dimension of the spacetime to $D=26$. However, in the case of the $p$-adic models considered in this paper the situation is different. Assuming a certain Euler product relation for the masses of the $p$-adic and Archimedean excitations (first anticipated in \cite{Huang:2019nog}), demanding that the spectrum contain a massless $U(1)$ field and graviton does not immediately fix the spacetime dimension, rather it enforces that the physical excitations live at the critical line zeros of the zeta function, in a way which will be made precise in the paper. Because our approach only applies to the critical line zeros of the zeta function, this result can be thought of as a converse to a modern version of the Hilbert-P\'olya conjecture, in that the critical line zeros obtain a spectral interpretation as the massless modes of our models, when considered across all places. With the benefit of hindsight, given the central nature of string theory in physics and of the Riemann zeta function in mathematics, perhaps such a deep connection should have been expected.\footnote{Recent works in non-Archimedean and $p$-adic physics include \cite{Gubser:2016guj,Bocardo-Gaspar:2016zwx,Gubser:2017vgc,Bhattacharyya:2017aly,Gubser:2017tsi,Dutta:2017bja,Gubser:2017qed,Bocardo-Gaspar:2017atv,2018arXiv180101189H,Marcolli:2018ohd,Gubser:2018bpe,Mendoza-Martinez:2018ktr,Gubser:2018yec,Jepsen:2018dqp,Gubser:2018cha,Gubser:2018ath,Heydeman:2018qty,Hung:2018mcn,Jepsen:2018ajn,Hung:2019zsk,Bocardo-Gaspar:2019pzk,Gubser:2019uyf,Garcia-Compean:2019jvk,Zuniga-Galindo:2020ypw,Chen:2021ipv,Chen:2021rsy,Chen:2021qah,Bocardo-Gaspar:2021ukf}.}

\subsection{The $p2$-brane}

Let's now briefly describe our setup. In \cite{Huang:2019nog}, building on the earlier work of Zabrodin \cite{zabrodin,zabrodin2}, a model of the $p$-adic string as maps $X:V\lb T_p\rb\to \mathbb{R}^D$ (that is from the vertices of the Bruhat-Tits tree into $D$-dimensional flat space) was considered. In the present paper, we will extend this model by including an extra time direction. Although the notion of dimensionality is somewhat fuzzy in the context of $p$-adic physics, we can interpret the number of dimensions of the worldsheet as follows. The (asymptotically hyperbolic) Bruhat-Tits worldsheet considered by Zabrodin \cite{zabrodin,zabrodin2} can be thought of as two-dimensional, for e.g. because it plays the same role in deriving the $n$-point Veneziano tree level amplitudes via tachyon vertex operator insertion as the usual bosonic string worldsheet. Alternatively, equipping the tree with edge degrees of freedom \cite{Gubser:2016htz} produces a close analogue of Euclidean two-dimensional gravity (\hspace{-0.002cm}\cite{Stoica:2018zmi}, based on \cite{unpub}). If the Bruhat-Tits tree can be roughly thought-of as two-dimensional, then adding an extra time direction will produce a three-dimensional worldsheet, i.e. a $p$-adic 2-brane, or \emph{$p2$-brane}. Unlike the usual bosonic string worldsheet, which is flat, the $p2$-brane preserves the hyperbolic asymptotics of the Bruhat-Tits tree.  The reason for the appearance of extended objects is simply that our theory will likely not be in a regime easily accessible from a string description.\footnote{Nonetheles, it is an interesting  open question whether a $p$-adic string description based on the Bruhat-Tits tree with no time direction exists for our models.}

We now describe the symmetries of the scalar fields $X$ on the worldsheet. Although in principle $X$ could have any dependence on the worldsheet coordinates, we will restrict the fields on a constant time slice to be invariant under the action of $B\cap \mrm{PGL}\lb 2,\mathbb{Z}_p\rb$, with $B$ the Borel subgroup, as in \cite{Huang:2019nog}. With this restriction, the time slice Laplacian eigenfunctions will be plane waves propagating from a boundary point $\infty$, with the eigenvalues obeying the zeta function symmetry $\mu\to 1-\mu$. Crucially, the eigenfunctions will satisfy orthogonality relations \eqref{eq11} below, first derived in \cite{Huang:2019nog}. As we will see, these orthogonality relations will be fundamental for our constructions in the rest of the paper. It is possible to relax the symmetry requirement of invariance under $B\cap \mrm{PGL}\lb 2,\mathbb{Z}_p\rb$, by considering at most $p+1$ points on the boundary from which the plane waves propagate, so that orthogonality relations \eqref{eq11} are still preserved (for details see Appendix~A in \cite{Huang:2019nog}), however we will not consider this in the present paper.

\subsection{Plan of the paper and outline: $p$-adic results}

We now present a brief summary of the paper. 

In Section \ref{sec2} we start by writing down the Lagrangian, which will be the Bruhat-Tits tree version of the Polyakov action introduced by Zabrodin \cite{zabrodin,zabrodin2}, but written for the tree times an extra time direction. We will normalize separately the time and spatial parts of the Lagrangian, and fix the relative sign to be negative so that the time evolution will be a $p$-adic analogue of Lorentzian evolution; our results will be sensitive to all of these choices. We will furthermore show that our system can be treated in either Lagrangian or Hamiltonian formalism.

In Section \ref{sec3} we will introduce the Fourier mode expansion, in the plane waves $\varphi_\mu(i)\coloneqq p^{\mu\langle i,\infty \rangle}$ of Zabrodin \cite{zabrodin,zabrodin2}. Throughout the discussion we root the tree at an $\infty$ point on the boundary and fix a center $O$. The parameter $\mu$ is complex, however in order for the quantum mechanical time evolution of the system to be unitary, $\mu$ will have to be restricted to certain loci in the complex plane, splitting the complex plane into sectors that will give different contributions to the $p2$-brane spectrum. Up to repetition of period $2\pi i /\ln p$, the sectors we will consider are i) $\mu=0,1$ (which will correspond to a zero mode), ii) the open interval $\mu\in (0,1)$, and iii) the critical line $\mu\in \{1/2+ it| t\in \mathbb{R}\}$. Depending on the value of a parameter in the action there can be a third sector where evolution is unitary, contained in the lines of imaginary part $\pm i\pi / \ln p$, however we will not consider it in this paper.

The modes $\varphi_\mu(i)$ obey the crucial orthogonality relation \cite{Huang:2019nog}
\be
\label{eq11}
\langle \varphi_\mu | \varphi_\nu \rangle \coloneqq \sum_{i\in V\lb T_p \rb} \varphi_\mu^*(i) \varphi_\nu(i) = \frac{1}{1-p}\lb \delta_{\mu^*+\nu} + \delta_{\mu^*+\nu-1} \rb,
\ee
with the star denoting complex conjugation and the $L^2$ norm sum regularized by analytic continuation. Furthermore, a highly nontrivial coincidence is that the orthogonality relation is compatible with the $\mu$-plane sectors, in that Eq. \eqref{eq11} pairs points in the same sector (with the $\delta_{\mu*+\nu-1}$ pairing in all three sectors, and $\delta_{\mu^*+\nu}$ giving an additional contribution for the zero mode). This coincidence will be crucial for the construction of oscillation modes. Note that Eq. \eqref{eq11} and the $\mu$-plane sectors are derived in logically independent ways, from the regularized $L^2$ tree norm and equation of motion respectively.

In Section \ref{sec4} we introduce position-momentum Poisson brackets, which will be promoted to canonical commutators. Eq. \eqref{eq11} can then be used to Fourier decompose the position-momentum commutators into commutators for the mode raising and lowering operators. This decomposition is closely analogous to the usual bosonic string mode decomposition, except that in the case of the $p2$-brane the modes will belong to sectors i)--iii) above. Note however that the $p$-adic raising and lowering operators will not obey a Virasoro algebra. This is because for the usual bosonic string, the Virasoro algebra arises from a remnant worldsheet diffeomorphism invariance (local conformal transformations), after the worldsheet metric is locally fixed. However, the $p2$-brane is discrete, and there are no infinitesimal conformal transformations to begin with.

There is another subtlety that occurs in the construction of Poisson brackets and canonical commutators. In the Archimedean world, the position-momentum commutator $[x,p]=i$ is equivalent to the momentum operator having a presentation as derivative. On a graph, due to its discrete nature, the momentum operator does not have a representation as a local (infinitesimal) derivative, which in general will affect the existence  of the commutation relations. Nonetheless, we will see that, in the case of the Bruhat-Tits tree, the orthogonality relations \eqref{eq11} will provide a way around this obstruction (see Section \ref{sec36intep}). Note that the orthogonality relations depend on the graph geometry, and will not readily generalize to other types of graphs. Thus, the existence of the commutation relations is a special feature of Bruhat-Tits trees.

Section \ref{sec5} contains the expansion of the Hamiltonian, momentum and mass operators in the Fourier modes, taking into account the three sectors. In Section \ref{sec8} we will check that the target space Poincar\'e algebra is obeyed at the semi-classical level. This is necessary because our discrete worldsheet model is not guaranteed to obey target space Poincar\'e invariance, and in fact we will see that a subtlety arises with the zero~mode, in that the $ \delta_{\mu^*+\nu}$ term in the inner product \eqref{eq11} is crucial in order for the zero sector contributions to the momentum and angular momentum to obey the Poincar\'e algebra. 

\subsection{Plan of the paper and outline: Adelic results}

Section \ref{sec7} contains results on the Archimedean spectrum, obtained via the Euler product construction, reliant on three additional conditions. These three additional conditions that we impose (and which were not needed for the $p$-adic results of the previous sections) are:
\begin{enumerate}
\item $H-\mathfrak{a}=0$, with $H$ the Hamiltonian and $\mathfrak{a}$ a normal ordering constant. This condition is analogous to the bosonic string mass-shell condition $L_0-\mathfrak{a}=0$ when acting on the physical states. We will give a worldsheet geometric symmetry interpretation for it.
\item We will impose the lightcone gauge, and gauge fix one null creation operator as $b_\mu^+=0$.
\item That the Euler product for the masses of the first excited level in the spectrum holds, that is
\be
\label{eq12}
\prod_p m^2_{(p)} = m^{-2}_{(\infty)}.
\ee
\end{enumerate}  

With these three conditions, the spectral interpretation of the critical line zeta zeros as massless excitations follows.

Let's now comment on the Euler product \eqref{eq12}. Although we will not prove this equation here, there are two pieces of evidence to support it. Firstly, in the case of field theory constructed from Vladimirov derivatives, a similar product holds for the free field Green's functions, which can be established rigorously and in fact follows from Tate's thesis \cite{Huang:2020aao}. Secondly, it can be explicitly checked in the case of the quantum mechanical particle-in-a-box, defined $p$-adically and at the usual Archimedean place, that the spectrum obeys precisely product formula \eqref{eq12} \cite{Huang:2020vwx}. Note however that, even in the simple case of the particle-in-a-box, a deeper understanding of why product formula \eqref{eq12} holds is currently missing (apart from explicitly checking that both sides are equal).

In Section \ref{secconc} we will propose a conjectural way to connect our construction to the usual landscape of string theory.

\begin{remark}[Zeta function from the Bruhat-Tits tree Laplacian eigenvalues]
Although in this paper we will consider \emph{one} particular model, it is important to emphasize that the zeta function arising from the eigenvalues of the Bruhat-Tits tree Laplacian is a \emph{general mechanism}, which can be adapted to many other models. This is because the minimal requirement for the zeta function to emerge is to have the Bruhat-Tits tree with one marked point on the boundary. As we will explain below, moving towards this vertex changes the Laplacian eigenfunction by $p^\mu$, and moving away from it will give a net change of $p^{1-\mu}$, because there are $p$ vertices pointing away from the marked point, for a complex parameter $\mu$. This simple fact is how the fundamental $\mu\leftrightarrow 1-\mu$ symmetry of the zeta function appears in our context. The factors of $p^{\mu}$ and $p^{1-\mu}$ get encoded in the Laplacian eigenvalues as local zeta factors, which will then Euler product into the zeta function at the Archimedean~place.
\end{remark}

\section{Lagrangians, Hamiltonians, and Fourier mode expansion on the Bruhat-Tits tree}
\label{sec2}

\subsection{Lagrangian and Hamiltonian formalisms}

In this section we set up the Lagrangian and Hamiltonian formalisms in the context of the $p2$-brane. The end goal is to define Poisson brackets and canonical commutators, however in order to do so we will first need to introduce worldsheet notions of \emph{time} and \emph{space}.

There is a well-known philosophical question regarding the nature of time. While in the Archimedean setting such a question may be more or less well-defined, depending on context, in the $p$-adic setting the question on the nature of time is automatically sharp, because there is a priori no natural way to introduce time. In the present paper we will simply take time to be $\mathbb{Z}$ (or equivalently the vertices of the line graph $T_1$), times an unit of length $\ell_\tau$. While it may be possible to understand from deeper principles such as double-coset constructions why time should be $\mathbb{Z}$ (at least in certain classes of models), our reason for doing so here is simply that it will allow for interesting $p$-adic models.

\begin{remark}
We introduce a time direction, such that our maps are defined as
\be
X: V\lb T_p \rb \times V\lb T_1 \rb \to \mathbb{R}^{1,d}.
\ee
Here $T_1$ is the tree with $p\to 1$, i.e. the line graph.
\end{remark}

\textbf{Notation and conventions:} We denote the target space directions by indices $a,b,\dots$ (mostly, but not always, written as a superscript), and the vertices on a constant time slice $T_p$ of the worldsheet by indices $i,j,\dots$ written as subscripts. We write $i\sim j$ when vertices $i$ and $j$ are neighbors on a time slice. We denote the worldsheet time by $\tau$. The spacetime indices $a,b,\dots$ are raised and lowered with the Minkowski metric $\eta^{ab}$.

\textbf{Conventions on units of length:} In principle all edges in $E\lb T_1 \rb\times E\lb T_p \rb$ can have different edge lengths, however in this paper we will consider only one length $\ell$ for all the edges in $E\lb T_p \rb$, and another length $\ell_\tau$ for all the edges in $E\lb T_1 \rb$ (this could be relaxed in future work). Furthermore, we will adopt the convention that $2\ell_\tau=\ell$ (this convention does not change the physics but is convenient). 

\begin{defn}
For functions $\psi: V\lb T_p \rb \times V\lb T_1 \rb \to \mathbb{R}$, we define the derivative in the time direction as
\be
\pd_\tau \psi_i \coloneqq \frac{1}{2\ell_\tau} \lsb \psi_i\lb \tau + \ell_\tau\rb - \psi_i\lb \tau - \ell_\tau\rb \rsb,
\ee
so that the second order derivative in the time direction is the same as the time direction Laplacian $\Delta_\tau$,
\be
\pd_\tau^2 \psi_i = \Delta_\tau \psi_i = \frac{1}{4\ell_\tau^2} \lsb \psi_i\lb \tau+2\ell_\tau \rb - 2\psi_i\lb \tau \rb + \psi_i\lb \tau-2\ell_\tau \rb \rsb .
\ee 
\end{defn}

\begin{defn}
We postulate the p2-brane action to be
\be
\label{Seq}
S \coloneqq \sum_{T_1} \lb \frac{1}{2} \sum_{i\in V\lb T_p \rb} \lb \pd_\tau X_i^a \rb^2  - \sum_{\ipj \in E\lb T_p \rb} \frac{\lb X^a_i-X^a_j\rb^2}{v_p \ell^2}\rb.
\ee
\end{defn}
We normalize the potential term separately, by introducing dimensionless coefficient $v_p$. In order to obtain the restriction to the $(0,1)$ interval and critical line in the $\mu$ plane, the value of $v_p$ will have to be constrainted, and the relative sign between the kinetic and potential contributions must be negative; we will discuss these shortly below. The relative minus sign between the kinetic and potential contributions in Eq. \eqref{Seq} can be thought of as indicating the ``Lorentzian'' nature of the worldsheet.

Note that coefficient $v_p$ can alteratively be thought of as a uniform rescaling of the spatial edge lengths relative to the time direction edges.

With the two definitions above, introducing the Lagrangian and Hamiltonian formalisms proceeds as expected.

\begin{remark}[Conjugate momenta and the Hamiltonian] 
The momentum density conjugate to $X_i$ is
\be
P_{ia} \coloneqq \frac{\delta S}{\delta \lb \pd_\tau X^a_i \rb} = \pd_\tau X_{ia}.
\ee
The Hamiltonian is
\ba
\label{eqhami}
H &\coloneqq& \sum_{i\in V\lb T_p \rb} \lsb \lb \pd_\tau X_{ia} \rb P^a_i - \LL\rsb \\
&=& \sum_{i\in V\lb T_p \rb} \frac{\lb P_i^a \rb^2 }{2} + \sum_{\ipj \in E\lb T_p \rb} \frac{\lb X^a_i-X^a_j\rb^2}{v_p \ell^2}.\nn
\ea
\end{remark}

In Section \ref{secHamicons} we will show that the Hamiltonian is worldsheet time $\tau$-independent.%, and in Appendix \ref{appB} we will show that the Hamiltonian is $\tau$-independent under a more general class of models than the one considered in the main part of the paper.

\begin{remark}[Hamilton's equations]
In classical mechanics, Hamilton's equations translate to
\be
\frac{d p_i}{dt} = - \frac{\pd H}{\pd q_i}, \quad \frac{d q_i}{dt} = \frac{\pd H}{\pd p_i}.
\ee
In our setup these equations are
\ba
\label{eq2p9}
\pd_\tau P_i^a &=& - \frac{2}{v_p} \sum_{j\sim i} \frac{X_i^a-X_j^a}{\ell^2}, \\
\pd_\tau X_i^a &=& P_i^a.\nn
\ea
Thus the equation of motion is
\be
\label{EOM}
\pd_\tau^2 X_i^a = - \frac{2}{v_p} \sum_{j\sim i} \frac{X_i^a-X_j^a}{\ell^2},
\ee
i.e.
\be
\label{EOM2}
 X_i^a(\tau+\ell) + X_i^a(\tau-\ell) - 2 X_i^a(\tau) = - \frac{2}{v_p} \sum_{j\sim i} \lsb X_i^a(\tau)-X_j^a(\tau) \rsb.
\ee
\end{remark}

In Eq. \eqref{EOM2} above we have used that $2\ell_\tau=\ell$. Note that although we have not independently justified Hamilton's equations in our setup, the end result (equation of motion \eqref{EOM2}) is valid, because it follows directly in the Lagrangian formalism from action \eqref{Seq}, by taking variation $\delta X_i^a(\tau)$.

\begin{defn}[Poisson bracket]
The classical mechanics Poisson bracket is
\be
\label{eq212}
\{f,g\}\coloneqq \sum_\alpha\lb \frac{\pd f}{\pd q_\alpha}\frac{\pd g}{\pd p_\alpha} - \frac{\pd g}{\pd q_\alpha}\frac{\pd f}{\pd p_\alpha} \rb.
\ee
\end{defn}

Index $\alpha$ in Eq. \eqref{eq212} above runs over the generalized coordinates. Applying this definition to our case (and noting that $X^a_i$ and $P^b_j$ are independent even when $a=b$, $i=j$, because $P_j^b$ only has information about the advanced and previous time steps), we can compute the position-momentum Poisson brackets:

\begin{remark} The position-momentum Poisson brackets in our setup are
\ba
\label{eeqq2p14}
\{X^a_i, P^b_j\} &=& \sum_{k\in V\lb T_p \rb} \eta^{cd}  \lb \frac{\pd X_i^a}{\pd X^c_k}\frac{\pd P^b_j}{\pd P^d_k} - \frac{\pd P^b_j}{\pd X^c_k}\frac{\pd X_i^a}{\pd P^d_k} \rb\\
&=& \eta^{ab} \delta_{ij}, \nn
\ea
and
\be
\label{eeqq2p15}
\{X^a_i, X^b_j\} = \{P^a_i, P^b_j\} = 0.
\ee
\end{remark}

\subsection{Plane waves on the Bruhat-Tits tree}

The Fourier modes $\phi_{\mu,\alpha}(i)$ on the Bruhat-Tits tree were first introduced by Zabrodin \cite{zabrodin,zabrodin2}, and are defined w.r.t. a boundary vertex $\alpha\in V\lb\pd T_p\rb$, a complex parameter $\mu\in\mathbb{C}$, and a choice of the center $O$ of the tree. In the present paper we will only use one choice of boundary vertex $\alpha$ for all the Fourier modes, which we denote by $\infty$; consequently we will write $\phi_{\mu}$ instead of $\phi_{\mu,\alpha}$.

\begin{defn}
\label{deffm}
For parameter $\mu\in\mathbb{C}$ and $i\in V\lb T_p \rb$, the Fourier mode $\phi_\mu(i)$ is
\be
\label{eq216}
\phi_\mu(i) \coloneqq \sqrt{p-1} p^{\mu\langle i,\infty \rangle},
\ee
with the bracket $\langle \cdot ,\infty \rangle: V\lb T_p \rb\to \mathbb{Z}$ defined as
\be
\langle i,\infty \rangle \coloneqq \begin{cases} 
\ell\lb i, O \rb, \qquad \qquad \,\;\; \mrm{if\ } \infty,O,i \mrm{\ are\ on\ the\ same\ geodesic,\ in\ this\ order} \\
\ell\lb O,i' \rb - \ell\lb i',i \rb, \ \, \mrm{if\ not}
\end{cases}.
\ee
$\ell\lb\cdot,\cdot\rb: V\lb T_p \rb \times V\lb T_p \rb \to \mathbb{Z}_{\geq 0} $ is the tree distance between the two vertices, and $i'$ is the vertex at which the tree paths $\infty\to i$ and $\infty \to O$ stop coinciding.
\end{defn}

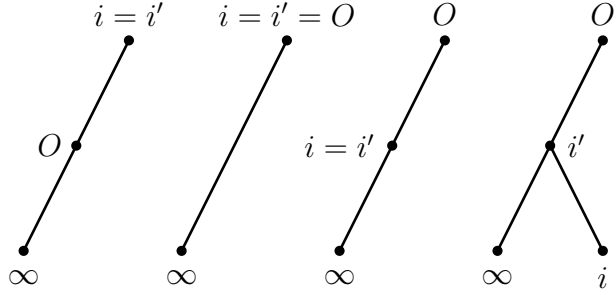
\begin{figure}[t]
\centering
\begin{tikzpicture}[scale=0.7]
		\node [style=simpletoo] (0) at (-6, -0) {};
		\node [style=simpletoo] (1) at (-5, 2) {};
		\node [style=simpletoo] (4) at (-4, 4) {};
		\node [style=simpletoo] (2) at (-3, -0) {};
		\node [style=simpletoo] (3) at (-1, 4) {};
		\node [style=simpletoo] (5) at (0, -0) {};
		\node [style=simpletoo] (6) at (2, 4) {};
		\node [style=simpletoo] (7) at (1, 2) {};
		\node [style=simpletoo] (8) at (3, -0) {};
		\node [style=simpletoo] (9) at (5, 4) {};
		\node [style=simpletoo] (10) at (4, 2) {};
		\node [style=simpletoo] (11) at (5, -0) {};
		\node [style=none] (12) at (-6, -0.5) {$\infty$};
		\node [style=none] (13) at (-3, -0.5) {$\infty$};
		\node [style=none] (14) at (0, -0.5) {$\infty$};
		\node [style=none] (15) at (3, -0.5) {$\infty$};
		\node [style=none] (16) at (5, -0.5) {$i$};
		\node [style=none] (17) at (-4, 4.5) {$i=i'$};
		\node [style=none] (18) at (-1, 4.5) {$i=i'=O$};
		\node [style=none] (19) at (2, 4.5) {$O$};
		\node [style=none] (20) at (5, 4.5) {$O$};
		\node [style=none] (21) at (-5.5, 2) {$O$};
		\node [style=none] (22) at (0, 2) {$i=i'$};
		\node [style=none] (23) at (4.5, 2) {$i'$};
		\draw [style=simple] (0.center) to (1.center);
		\draw [style=simple] (1.center) to (4.center);
		\draw [style=simple] (2.center) to (3.center);
		\draw [style=simple] (5.center) to (7.center);
		\draw [style=simple] (7.center) to (6.center);
		\draw [style=simple] (8.center) to (10.center);
		\draw [style=simple] (10.center) to (9.center);
		\draw [style=simple] (10.center) to (11.center);
\end{tikzpicture}
\caption{All possible relative positions of vertices $i,i',O,\infty$, for computing the bracket~$\langle i,\infty \rangle$.}
\label{fig1}
\end{figure}

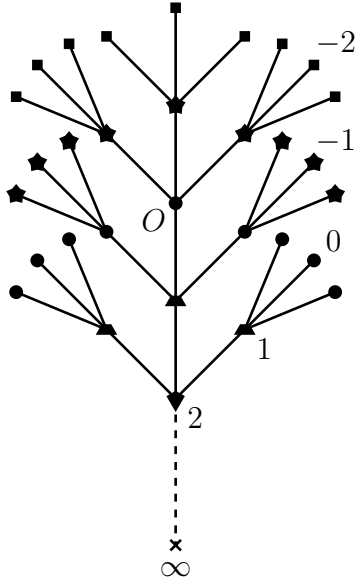
\begin{figure}[t]
\centering
\begin{tikzpicture}[scale=0.65]

\tikzmath{
\r=2;
\eps=0.5;
\cossin45=0.707107;
}

\node [style=simplesq] (lev1vert0) at (0, \r*2) {};
\node [style=simplesq] (lev1vert1) at (\r*\cossin45, \r+\r*\cossin45) {};
\node [style=simplesq] (lev1vert2) at (-\r*\cossin45, \r+\r*\cossin45) {};

\node [style=simplestar] (lev0vert0) at (0, \r) {};
\node [style=simplestar] (lev0vert1) at (\r*\cossin45, \r*\cossin45) {};
\node [style=simplestar] (lev0vert2) at (-\r*\cossin45, \r*\cossin45) {};

\node [style=simplecirc] (levm1vert0) at (0, 0) {};
\node [style=simplecirc] (levm1vert1) at (\r*\cossin45, -\r+\r*\cossin45) {};
\node [style=simplecirc] (levm1vert2) at (-\r*\cossin45, -\r+\r*\cossin45) {};

\node [style=simplesq] (levm1vert3) at (2*\r*\cossin45, 2*\r*\cossin45) {};
\node [style=simplesq] (levm1vert4) at (\r*\cossin45+\r*0.382683, \r*\cossin45+\r*0.92388) {};
\node [style=simplesq] (levm1vert5) at (\r*\cossin45+\r*0.92388, \r*\cossin45+\r*0.382683) {};
\node [style=simplesq] (levm1vert6) at (-\r*\cossin45-\r*0.382683, \r*\cossin45+\r*0.92388) {};
\node [style=simplesq] (levm1vert7) at (-\r*\cossin45-\r*0.92388, \r*\cossin45+\r*0.382683) {};
\node [style=simplesq] (levm1vert8) at (-2*\r*\cossin45, 2*\r*\cossin45) {};

\node [style=simpletrapez] (levm2vert0) at (0, -\r) {};

\node [style=simplestar] (levm2vert3) at (2*\r*\cossin45, -\r+2*\r*\cossin45) {};
\node [style=simplestar] (levm2vert4) at (\r*\cossin45+\r*0.382683, -\r+\r*\cossin45+\r*0.92388) {};
\node [style=simplestar] (levm2vert5) at (\r*\cossin45+\r*0.92388, -\r+\r*\cossin45+\r*0.382683) {};
\node [style=simplestar] (levm2vert6) at (-\r*\cossin45-\r*0.382683, -\r+\r*\cossin45+\r*0.92388) {};
\node [style=simplestar] (levm2vert7) at (-\r*\cossin45-\r*0.92388, -\r+\r*\cossin45+\r*0.382683) {};
\node [style=simplestar] (levm2vert8) at (-2*\r*\cossin45, -\r+2*\r*\cossin45) {};

\node [style=simplekite] (levm3vert0) at (0, -2*\r) {};
\node [style=simpletrapez] (levm3vert1) at (\r*\cossin45, -2*\r+\r*\cossin45) {};
\node [style=simpletrapez] (levm3vert2) at (-\r*\cossin45, -2*\r+\r*\cossin45) {};

\node [style=simplecirc] (levm3vert3) at (2*\r*\cossin45, -2*\r+2*\r*\cossin45) {};
\node [style=simplecirc] (levm3vert4) at (\r*\cossin45+\r*0.382683, -2*\r+\r*\cossin45+\r*0.92388) {};
\node [style=simplecirc] (levm3vert5) at (\r*\cossin45+\r*0.92388, -2*\r+\r*\cossin45+\r*0.382683) {};
\node [style=simplecirc] (levm3vert6) at (-\r*\cossin45-\r*0.382683, -2*\r+\r*\cossin45+\r*0.92388) {};
\node [style=simplecirc] (levm3vert7) at (-\r*\cossin45-\r*0.92388, -2*\r+\r*\cossin45+\r*0.382683) {};
\node [style=simplecirc] (levm3vert8) at (-2*\r*\cossin45, -2*\r+2*\r*\cossin45) {};

\node [style=simplecrossout] (levminfvert0) at (0, -3.5*\r) {};

\draw [style=simple] (lev0vert0.center) to (lev1vert0.center);
\draw [style=simple] (lev0vert0.center) to (lev1vert1.center);
\draw [style=simple] (lev0vert0.center) to (lev1vert2.center);

\draw [style=simple] (lev0vert0.center) to (levm1vert0.center);
\draw [style=simple] (lev0vert1.center) to (levm1vert0.center);
\draw [style=simple] (lev0vert2.center) to (levm1vert0.center);

\draw [style=simple] (levm1vert3.center) to (lev0vert1.center);
\draw [style=simple] (levm1vert4.center) to (lev0vert1.center);
\draw [style=simple] (levm1vert5.center) to (lev0vert1.center);
\draw [style=simple] (levm1vert6.center) to (lev0vert2.center);
\draw [style=simple] (levm1vert7.center) to (lev0vert2.center);
\draw [style=simple] (levm1vert8.center) to (lev0vert2.center);

\draw [style=simple] (levm2vert3.center) to (levm1vert1.center);
\draw [style=simple] (levm2vert4.center) to (levm1vert1.center);
\draw [style=simple] (levm2vert5.center) to (levm1vert1.center);
\draw [style=simple] (levm2vert6.center) to (levm1vert2.center);
\draw [style=simple] (levm2vert7.center) to (levm1vert2.center);
\draw [style=simple] (levm2vert8.center) to (levm1vert2.center);

\draw [style=simple] (levm3vert3.center) to (levm3vert1.center);
\draw [style=simple] (levm3vert4.center) to (levm3vert1.center);
\draw [style=simple] (levm3vert5.center) to (levm3vert1.center);
\draw [style=simple] (levm3vert6.center) to (levm3vert2.center);
\draw [style=simple] (levm3vert7.center) to (levm3vert2.center);
\draw [style=simple] (levm3vert8.center) to (levm3vert2.center);

\draw [style=simple] (levm1vert0.center) to (levm2vert0.center);
\draw [style=simple] (levm1vert1.center) to (levm2vert0.center);
\draw [style=simple] (levm1vert2.center) to (levm2vert0.center);

\draw [style=simple] (levm2vert0.center) to (levm3vert0.center);

\draw [style=simple] (levm3vert0.center) to (levm3vert1.center);
\draw [style=simple] (levm3vert0.center) to (levm3vert2.center);

\draw [style=simpled] (levm3vert0.center) to (levminfvert0.center);

\node [style=none] (labelinf) at (0,-3.5*\r-\eps) {$\infty$};

\node [style=none] (labelO) at (-\eps*0.9,-\eps*0.7) {$O$};

\node [style=none] (labelm2) at (2*\r*\cossin45+0.9*\eps, 2*\r*\cossin45+0.9*\eps) {$-2$};
\node [style=none] (labelm1) at (2*\r*\cossin45+0.9*\eps, -\r+2*\r*\cossin45+0.9*\eps) {$-1$};
\node [style=none] (label0) at (2*\r*\cossin45+0.8*\eps, -2*\r+2*\r*\cossin45+0.8*\eps) {$0$};
\node [style=none] (label1) at (\r*\cossin45+0.8*\eps, -2*\r+\r*\cossin45-0.8*\eps) {$1$};
\node [style=none] (label2) at (0.8*\eps, -2*\r-0.8*\eps) {$2$};

\end{tikzpicture}
\caption{Values of the bracket $\langle \cdot,\infty \rangle$ for a subset of the Bruhat-Tits tree. The origin $O$ has $\langle O,\infty\rangle=0$, and vertices denoted with the same geometric shape have the same value of the bracket, as depicted in the figure: $\langle {\scriptstyle\blacksquare},\infty\rangle=-2$, and so on. All vertices with a given value of $\langle \cdot,\infty \rangle$ have the same value of $X^a$. Note that $\langle \infty,\infty\rangle =\infty$, and the bracket values correspond to the argument of a plane wave starting at $\infty$, offset to point~$O$.}
\label{fig1p}
\end{figure}

Note that $i=i'$ when vertices $\infty,O,i$ are on the same geodesic. The factor of $1/\sqrt{p-1}$ in Eq. \eqref{eq216} is a normalization convention.

The possible relative positions of vertices $\infty,O,i$ are shown in Figure \ref{fig1}, and some values of the function $\langle\cdot,\infty\rangle$ are shown in Figure \ref{fig1p}. 

The meaning of the functions $p^{\mu\langle \cdot,\infty \rangle}$ is that they are \emph{plane waves} coming in from the point at $\infty$, analogous to functions $e^{ikx}$ which are plane waves in the Archimedean world. The bracket $\langle \cdot,\infty \rangle$ counts the number of steps starting from $\infty$, just as $x$ is a measure of Archimedean distance; because the distance from $\infty$ to any vertex in $T_p$ is infinite, $\langle\cdot,\infty\rangle$ is offset by an infinite amount to the center $O$, so that $\langle O,\infty\rangle =0$.

The following two results (Remark \ref{rmk5} and Theorem \ref{thm1}) justify interpreting the functions $p^{\mu\langle \cdot ,\infty \rangle}$ as plane waves. It is also possible to define a translation operator such that $p^{\langle \cdot,\infty\rangle} $ are momentum eigenfunctions, see \cite{Huang:2019nog}.

\begin{remark}
\label{rmk5}
The functions $p^{\mu\langle \cdot,\infty \rangle}$ are tree Laplacian eigenvalues, that is
\be
\Delta p^{\mu\langle i,\infty \rangle} = \lambda_\mu p^{\mu \langle i,\infty \rangle},
\ee
for eigenvalue
\be
\label{eqq220}
\lambda_\mu = p^{\mu}+p^{1-\mu} - p -1,
\ee
and with the Laplacian $\Delta$ acting on an arbitrary function $\psi:V\lb T_p \rb\to \mathbb{C}$ as
\be
\Delta \psi(i) = \lb \sum_{j\sim i } \psi(j) \rb - (p+1) \psi(i).
\ee
\end{remark}

\begin{defn}
For functions $\psi_{1,2}:V\lb T_p \rb \to \mathbb{C}$, the $L^2$ inner product is
\be
\langle \psi_1 | \psi_2 \rangle \coloneqq \sum_{i \in V\lb T_p \rb} \psi_1^*(i) \psi_2(i),
\ee
with the star denoting complex conjugation.
\end{defn}

\begin{thm}
\label{thm1}
For $\mu,\nu\in\mathbb{C}$ and $\varphi_\lambda(i)=p^{\lambda\langle i,\infty \rangle}$, we have the orthonormality relations
\be
\label{eq221}
\langle \varphi_\mu | \varphi_\nu \rangle = \frac{1}{1-p}\lb \delta_{\mu^*+\nu} + \delta_{\mu^*+\nu-1} \rb,
\ee
with $\delta$ the Kronecker delta function.
\end{thm}

Note that due to normalization Eq. \eqref{eq221} implies 
\be
\label{eq223}
\langle \phi_\mu | \phi_\nu \rangle = - \delta_{\mu^*+\nu} - \delta_{\mu^*+\nu-1}.
\ee

Theorem \ref{thm1} was proven in \cite{Huang:2019nog} (see therein for the details). Schematically, the proof reduces to computing the $L^2$ sum over all vertices of the tree, which given the definition of the bracket $\langle \cdot,\infty \rangle$ is just a combination of geometric sums,
\ba
\label{eq224}
\langle \varphi_\mu | \varphi_\nu \rangle = \sum_{i=0}^\infty p^{\lb 1-\mu^*-\nu \rb i } + \sum_{i=1}^\infty p^{\lb \mu^*+\nu\rb i} \lb 1 + \frac{p-1}{p} \sum_{j=1}^\infty p^{\lb 1-\mu^*-\nu\rb j} \rb.
\ea
The first term corresponds to the branches ``above'' $O$ in the right panel of Figure \ref{fig1p} and the the second term corresponds to the branches ``below'' $O$. Eq. \eqref{eq224} is naively divergent, and must be regularized by analytically continuing each geometric sum to the entire complex plane in parameter $\mu^*+\nu$, minus the poles (see \cite{Huang:2019nog} for the details). With this regularization, sum \eqref{eq224} evaluates to Eq. \eqref{eq221} once the geometric sums are computed.

\subsection{Mode expansion on the Bruhat-Tits tree}
\label{sec3}

So far, the parameter $\mu$ entering the Fourier modes has been an arbitrary complex number. Our task in the rest of the paper will be to restrict the allowed values of $\mu$.

\begin{remark}[The pre-state plane]
The complex $\mu$ plane can be thought of as ``almost'' a plane of states, which we will call the \emph{pre-state plane}, in the following sense. Each $\mu\in \mathbb{C}$ is associated to a plane wave $\phi_\mu$; in the Fourier decomposition of the maps $X$ the Fourier modes $\phi_\mu$ will correspond to creation operators $c_\mu^{a\dagger}$ and to states $c_\mu^{a\dagger}| 0 \rangle$ once spatial index $a$ is introduced. Thus, the naive interpretation of the complex $\mu$ plane is that each pair $\lb \mu,a \rb$ corresponds to a state $c_\mu^{a\dagger} |0\rangle$. Of course, not all the states $c_\mu^{a\dagger} |0\rangle$ introduced in this way are physical: if the index $a$ is timelike the states will have negative norm, and if $\mu$ is off certain loci in the complex plane then the worldsheet time evolution will not be unitary. We will discuss restricting the values of $\mu$ and target space index $a$ later in the paper.
\end{remark}

\begin{remark}
We will show below that for $v_p$ above a certain value, the special loci in the $\mu$ plane for which the worldsheet time evolution is unitary are $\mu\in\lsb 0,1 \rsb$, $\mu\in \lcb 1/2 + it | t\in \mathbb{R} \rcb$ and $\mu\in\lsb t_1,t_2 \rsb + i\pi/\ln p$ (the values of parameters $t_{1,2}\in\mathbb{R}$ are given below), as well as their periodic repetitions with period $2\pi i/\ln p$ in the imaginary direction. These periodic repetitions correspond to the same values of $p^\mu$ and are thus not physical.
\end{remark}

\begin{remark} 
The orthogonality condition \eqref{eq221} for the inner product $\langle \phi_\mu | \phi_\nu \rangle$ is highly sensitive to the geometry of the tree. Later in the paper this condition will lead to state~orthogonality conditions.
\end{remark}

\begin{remark}[Symmetries of the maps $X$ and the $p2$-brane] 
As mentioned above, we fix the same boundary vertex $\infty$ for all the plane waves entering the Fourier decomposition of maps $X$. With $\infty$ and the center $O$ fixed, the plane waves $\phi_\mu$, and thus the maps $X$, are invariant under the action of the stabilizer $S(O,\infty)$ of vertices $\{O,\infty\}$. This stabilizer~is
\be
S(O,\infty) = B \cap \mrm{PGL}\lb 2,\mathbb{Z}_p\rb,
\ee
where $B$ is the standard Borel of upper triangular matrices (see \cite{Huang:2019nog} for the details). Thus, all vertices with a given value of the bracket $\langle i,\infty \rangle$  have \emph{the same} value of $X_i^a$ (see Figure \ref{fig1p}), and we don't have one direction's worth of degrees of freedom in the direction orthogonal to the direction of change in $\langle i,\infty \rangle$. Because of this, the theory on the $p2$-brane with vertices $V\lb T_1 \rb\times V\lb T_p \rb $ will be close in behavior to a theory with two independent directions on the worldsheet, i.e. to a string theory, as we will show later in the paper.

\end{remark}

Note that the symmetries of the maps $X$ considered in this paper are different from those of the underlying graph.

We will now expand the position $X_i^a$ in Fourier modes, as

\be
\label{exp0}
X^a_i = \mathfrak{p}^a \tau +  \sum_{\mu\in\mathbb{C}} \lsb c^a_\mu \lb r_+\rb ^{\tau/\ell} \phi_{\mu}(i) + \tilde{c}^a_\mu \lb r_-\rb ^{\tau/\ell} \phi_{\mu}(i)  \rsb.
\ee
Here $\mathfrak{p}^a \tau$ is the zero mode (just as in the Archimedean case, a linear term in worldsheet time solves the equation of motion), and the terms in the square brackets are the Fourier modes. The constant in the linear piece is included in the Fourier modes at $\mu=0$.

\begin{remark}
As a techical point, by the sum over $\mu\in\mathbb{C}$ in Eq. \eqref{exp0} we really mean a \emph{discrete} sum over a discrete subset of complex values that include certain special values (such at the zeta zeros), although we will not write this explicitly. Note that in line with this, the delta functions in Eq. \eqref{eq221} are Kronecker delta functions, and not Dirac. Later in the paper, when we will split the sum over modes into a sum over the $(0,1)$ interval and the critical line, we will likewise mean a sum over discrete subsets included in the $(0,1)$ interval and the critical line. 
\end{remark}

Coefficients $r_\pm(\mu)$ are the two solutions to the equation of motion \eqref{EOM2} applied to an eigenmode, and $c_\mu^a$, $\tilde c_\mu^a$ are Fourier coefficients that will later on become the creation and annihilation operators. Eq. \eqref{EOM2} applied to an eigenmode gives
\be
\label{eq120}
r_\pm + r_\pm^{-1} - 2 = \frac{2}{v_p} \lambda_\mu,
\ee
so that
\be
\label{eq229}
r_\pm = \frac{\lambda_\mu +v_p \pm \sqrt{\lambda_\mu (\lambda_\mu +2v_p)}}{v_p},
\ee
with tree Laplacian eigenvalue $\lambda_\mu$ given by Eq. \eqref{eqq220}.

From Vi\`ete's relations $r_+r_-=1$, thus in order for the worldsheet time evolution to be unitary we must have $|r_\pm|=1$, otherwise one mode will always blow up to $\infty$ as $\tau$ increases.

Theorem \ref{thm2}, Remark \ref{remm10} and Theorem \ref{thm3} below characterize when $\left|r_\pm\right|=1$.

\begin{thm}
\label{thm2}
We have $\left| r_\pm\right|=1$ for:
\begin{enumerate}
\item $\mu\in \lsb 0,1 \rsb$ and $v_p \geq \lb \sqrt{p}-1 \rb^2/2 $.
\item $\mu \in \lcb 1/2 + it | t\in\mathbb{R} \rcb$ and $v_p \geq \lb \sqrt{p}+1 \rb^2/2 $.
\end{enumerate}
\end{thm}

Theorem \ref{thm2} above is the result we will be interested in the rest of the paper. However, we can fully characterize the regions in the pre-state plane for which $|r_\pm|=1$, for any value of $v_p$:

\begin{remark}
\label{remm10} 
For arbitrary $v_p\in\mathbb{R}$, up to repetition by $2\pi i /\ln p$ in the imaginary direction, there are the following loci in the complex $\mu$ plane where $|r_\pm|=1$:
\begin{enumerate}
\setcounter{enumi}{-1}
\item \label{case0} If $v_p<0$, then $|r_\pm|=1$ for $\mu\in [\mu_1,0]\cup[1,\mu_2]$, where
\be
\label{twopoints}
\mu_{1,2} = \frac{1}{\ln p } \ln \lb \frac{p+1}{2}- v_p \pm \sqrt{v_p^2 - (p+1) v_p + \frac{(p-1)^2}{4}} \rb.
\ee
When $v_p=0$ we have a degeneracy, since $\mu_{1,2}=0,1$.
\item \label{case1} When $v_p$ increases past $0$ the real line locus becomes $[0,\mu_1]\cup[\mu_2,1]$, and it encompasses the entire $[0,1]$ interval at $v_p=(\sqrt{p}-1)^2/2$, but does not extend past the $[0,1]$ interval as $v_p$ is increased further. Endpoints $\mu_{1,2}$ are given by Eq. \eqref{twopoints}, so that $\mu_1(v_p=0)=0$, $\mu_2(v_p=0)=1$, and $\mu_{1,2}(v_p=(\sqrt{p}-1)^2/2)=1/2$.
\item \label{case2} The second locus, contained in the critical line, starts appearing as $v_p$ increases past $(\sqrt{p}-1)^2/2$, as $\mu\in \{ 1/2 +iT_1, 1/2 + iT_2 \}$, with
\be
T_{1,2}= -\frac{i}{\ln p} \ln \lb\frac{p+1-2 v_p}{2 \sqrt{p}} \pm \sqrt{\frac{(p-2 v_p+1)^2}{4p}-1} \rb ,
\ee
so that it encompasses the entire critical line at $v_p=(\sqrt{p}+1)^2/2$.
\item \label{case3} If $v_p > \lb \sqrt{p}+1 \rb^2/2$ then there is a third locus for which $\left| r_\pm\right|=1$, when
\be
\label{eeq230}
\mu = \frac{i\pi}{\ln p} +t,
\ee
with $t\in\mathbb{R}$ a parameter such that $t_1\leq t\leq t_2$, with the endpoints given by
\be
t_{1,2} = \frac{1}{\ln p}\ln \left( v_p-\frac{p+1}{2} \pm \sqrt{v_p^2-(p+1) v_p+\frac{(p-1)^2}{4}} \right).
\ee
\end{enumerate}
\end{remark}
Note that $t_1=1-t_2$ when $v_p\geq  \lb \sqrt{p}+1 \rb^2/2$.

Some of the cases in Remark \ref{remm10} are represented graphically in Figure \ref{figmuplane}.

\begin{figure}[t]
\centering
\begin{tikzpicture}[scale=1]
\node [style=none] (0) at (-2, 0) {};
\node [style=none] (1) at (2,0) {};
\node [style=none] (2) at (0, 2) {};
\node [style=none] (3) at (0,-2) {};
\draw [style=simple,-stealth] (0.center) to (1.center);
\draw [style=simple,stealth-] (2.center) to (3.center);
\node [style=none] (origin) at (0, 0) {};
\node [style=none] (one) at (1, 0) {};
\draw [line width=0.9mm] (origin.center) to (one.center);
\node [style=none] (originlabel) at (-0.2, -0.25) {$0$};
\node [style=none] (originlabel) at (1.05, -0.25) {$1$};

\node [style=none] (0p) at (3, 0) {};
\node [style=none] (1p) at (7,0) {};
\node [style=none] (2p) at (5, 2) {};
\node [style=none] (3p) at (5,-2) {};
\draw [style=simple,-stealth] (0p.center) to (1p.center);
\draw [style=simple,stealth-] (2p.center) to (3p.center);
\node [style=none] (originp) at (5, 0) {};
\node [style=none] (onep) at (6, 0) {};
\draw [line width=0.9mm] (originp.center) to (onep.center);
\node [style=none] (originlabelp) at (4.8, -0.25) {$0$};
\node [style=none] (originlabelp) at (6.05, -0.25) {$1$};
\node [style=none] (criticalbot) at (5.5,-2) {};
\node [style=none] (criticaltop) at (5.5,2) {};
\draw [line width=0.9mm] (criticalbot.center) to (criticaltop.center);

\node [style=none] (0s) at (8, 0) {};
\node [style=none] (1s) at (12,0) {};
\node [style=none] (2s) at (10, 2) {};
\node [style=none] (3s) at (10,-2) {};
\draw [style=simple,-stealth] (0s.center) to (1s.center);
\draw [style=simple,stealth-] (2s.center) to (3s.center);
\node [style=none] (origins) at (10, 0) {};
\node [style=none] (ones) at (11, 0) {};
\draw [line width=1mm] (origins.center) to (ones.center);
\node [style=none] (originlabels) at (9.8, -0.25) {$0$};
\node [style=none] (originlabels) at (11.05, -0.25) {$1$};
\node [style=none] (criticalbotp) at (10.5,-2) {};
\node [style=none] (criticaltopp) at (10.5,2) {};
\draw [line width=0.9mm] (criticalbotp.center) to (criticaltopp.center);
\node [style=none] (pilogp1) at (10.5-1,1) {};
\node [style=none] (pilogp2) at (10.5+1,1) {};
\draw [line width=0.9mm] (pilogp1.center) to (pilogp2.center);
\node [style=none] (mpilogp1) at (10.5-1,-1) {};
\node [style=none] (mpilogp2) at (10.5+1,-1) {};
\draw [line width=0.9mm] (mpilogp1.center) to (mpilogp2.center);
\node [style=none] (pilogpt1) at (10.5-1.3,1) {$t_1$};
\node [style=none] (pilogpt2) at (10.5+1.3,1) {$t_2$};
\end{tikzpicture}
\caption{Complex $\mu$ plane loci (in thickened lines) of unitary worldsheet time evolution for: i)~$v_p=(\sqrt{p}-1)^2/2$, ii) $v_p=(\sqrt{p}+1)^2/2$, iii) $v_p>(\sqrt{p}+1)^2/2$.}
\label{figmuplane}
\end{figure}
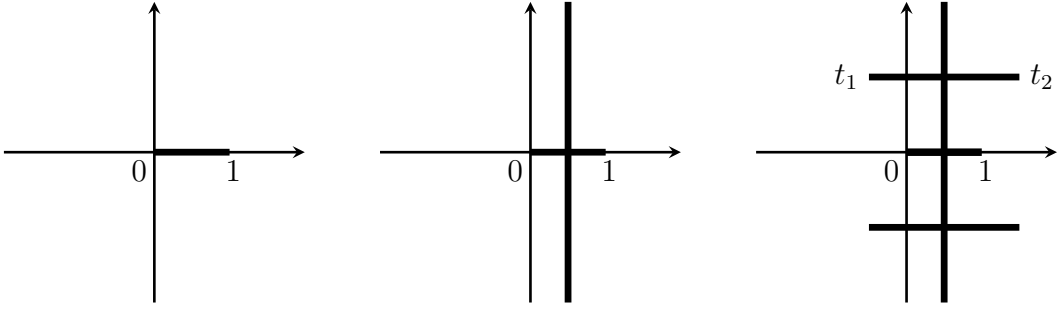

\begin{thm}
\label{thm3}
The loci in Remark \ref{remm10} are the only points in the complex $\mu$ plane for which $|r_\pm|=1$, for all values of the parameter $v_p\in\mathbb{R}$, up to periodicity by $2\pi i k/\ln p,\ k\in \mathbb{Z} $ in the imaginary direction.
\end{thm}

Note that when $v_p=\lb \sqrt{p}+1 \rb^2/2$ we have $\mu_{1,2}=1/2$, so that the two endpoints for the $\mu$ interval in Eq. \eqref{eeq230} are coincident on the critical line, and the third locus (in this case just a point) is contained in the critical line. When $v_p=p+1$ we have $[t_1,t_2]=[0,1]$, however as $v_p$ keeps increasing the length of the $[t_1,t_2]$ interval will keep increasing as well (unlike the behavior of the $\mu$ interval in the first locus, which is bounded to $[0,1]$ as $v_p$ is increased).

We should also remark that Theorems \ref{thm2}, \ref{thm3} imply that the relative sign between the kinetic and potential terms in the action \eqref{Seq} must be negative, otherwise there is no value for $\mu\in\mathbb{C}$ on the critical line for which the worldsheet time evolution is unitary.

In order to prove Theorems \ref{thm2}, \ref{thm3}, we will need the following lemma.
\begin{lemma}
\label{lemmalambdareal}
We have $\lambda_\mu\in\mathbb{R}$ iff any one of the following three conditions holds:
\begin{itemize}
\item $ \Re \lb \mu\rb =1/2$,
\item $\Im \lb \mu\rb = 0$ (up to $2\pi i /\ln p$ periodicity),
\item $\Im \lb \mu\rb = \pi/\ln p$ (up to $2\pi i /\ln p$ periodicity).
\end{itemize}
\end{lemma}
\begin{proof}
We have $\lambda_\mu = p^\mu+p^{1-\mu}-p-1$. Write $\mu=a+ib$, $a,b\in\mathbb{R}$, then
\be
\Im \lb \lambda_\mu \rb = p^a \sin\lb b \ln p \rb - p^{1-a} \sin\lb b \ln p \rb.
\ee
Therefore $\Im \lb \lambda_\mu \rb=0$ if $a=1/2$, or if $b=0,\pi/\ln p $ (up to $2\pi/\ln p$ repetition).
\end{proof}

We now prove Theorems \ref{thm2}, \ref{thm3}, and Remark \ref{remm10}.

\begin{proof}

We have that $\lambda_\mu \in \mathbb{R}$, $\lambda_\mu + 2v_p \geq 0$ when  $\mu \in \mathbb{R}$ and $v_p\geq \lb\sqrt{p}-1\rb^2/2$, or when $\mu\in\{1/2+it|t\in\mathbb{R}\}$, $v_p\geq \lb\sqrt{p}+1\rb^2/2$. Furthermore $\lambda_\mu\leq 0$ for $\mu \in [0,1]$ or $\mu\in\{1/2+it|t\in\mathbb{R}\}$.
Therefore the square root in Eq. \eqref{eq229} is purely imaginary (or zero) on
\be
\mathcal{D} \coloneqq \lcb\mu\ |\ \mu\in [0,1] \rcb \cup \lcb \mu\ |\ \Re\lb\mu\rb = \frac{1}{2} \rcb.
\ee
Then the norm of $r_\pm$ on $\mathcal{D}$ is
\be
\label{eeeq235}
|r_\pm|^2 = \left| \frac{\lambda_\mu +v_p \pm \sqrt{\lambda_\mu (\lambda_\mu +2v_p)}}{v_p} \right|^2 = \frac{\lb \lambda_\mu + v_p \rb^2 - \lambda_\mu\lb \lambda_\mu+2v_p \rb}{v_p^2} =1.
\ee
This proves Theorem \ref{thm2}.

To prove Remark \ref{remm10} and Theorem \ref{thm3}, note that from Vi\`ete's relations we have
\be
\label{eq235}
r_+ + r_- = 2 \lb 1 + \frac{\lambda_\mu}{v_p} \rb,
\ee
and for parameter $\mu$ such that $|r_\pm|=1$ we have $r^*_+=r_-$, so that Eq. \eqref{eq235} becomes
\be
r_+ + r_+^* = 2 \lb 1 + \frac{\lambda_\mu}{v_p} \rb.
\ee
Therefore, for the values of the parameter $\mu\in\mathbb{C}$ such that $|r_\pm|=1$ we must have $\lambda_\mu\in\mathbb{R}$, so $\mu$ must be contained in the loci in Lemma \ref{lemmalambdareal}. Now, suppose that 
\be
\label{eq2337}
\lambda_\mu \lb \lambda_\mu + 2v_p \rb > 0, 
\ee
then from Eq. \eqref{eq229} we have 
\be
\label{eq239}
\left| r_\pm \right| = \pm_1 \frac{\lambda_\mu +v_p \pm \sqrt{\lambda_\mu (\lambda_\mu +2v_p)}}{v_p}, 
\ee
with the sign $\pm_1$ chosen so that the norm is positive. Since the inequality in Eq. \eqref{eq2337} is strict, there are no solutions to the equation $|r_\pm|=1$, with $|r_\pm|$ given by Eq. \eqref{eq239}. If instead 
\be
\label{eeqq240}
\lambda_\mu \lb \lambda_\mu + 2v_p \rb \leq 0,
\ee
then from Eq. \eqref{eeeq235} we have $|r_\pm|=1$ for all $\mu$ such that Eq. \eqref{eeqq240} holds. If $\lambda_\mu<0$, then direct computation gives Cases \ref{case1}--\ref{case3} in Remark \ref{remm10}, and if $\lambda_\mu>0$ (which can only happen when $v_p<0$) we obtain Case \ref{case0}. This proves Remark \ref{remm10} and Theorem \ref{thm3}. 
\end{proof}

Now, let $\UU_{v_p}\subset \mathbb{C}$ be a subset of the three loci above such that $|r_\pm|=1$ for $\mu\in\UU_{v_p}$, for a fixed choice of $v_p$. The $p$-adic time evolution of the brane is unitary for $\mu\in\UU_{v_p}$, and we restrict expansion~\eqref{exp0} to these modes, that is
\be
X^a_i(\tau) = \mathfrak{p}^a\tau + \sum_{\mu\in\UU_{v_p}} \lsb c^a_\mu e^{i\omega_\mu\tau} + \tilde{c}^a_\mu e^{-i\omega_\mu\tau} \rsb \phi_{\mu}(i) ,
\ee
where now $e^{i\omega_\mu\tau}\coloneqq \lb r_+\rb^{\tau/\ell}$ (note that then $e^{-i\omega_\mu\tau}= \lb r_-\rb^{\tau/\ell}$ ). Angular frequency $\omega_\mu\in\mathbb{R}$ is given by
\be
\label{eeqq235}
\omega_\mu = - \frac{i}{\ell} \ln \lb \frac{\lambda_\mu +v_p +\sqrt{\lambda_\mu^2+2 \lambda_\mu v_p}}{v_p} \rb.
\ee

\begin{lemma}[Properties of $\omega_\mu$]
\label{lemmaomeganuprop}
\begin{enumerate}
\item \label{lpoint1} Eq. \eqref{eeqq235} above implies that $0\leq \omega_\mu \leq \pi$, for $\mu$ anywhere in the three loci (including endpoints), and any $v_p\geq \lb \sqrt{p}-1 \rb^2/2$. Furthermore, for the three loci we have the bounds:
\begin{enumerate}
\item $0\leq \omega_\mu \leq \omega_{\mu=1/2}$ for $\mu\in [0,1]$ and  $v_p\geq \lb \sqrt{p}-1 \rb^2/2$.
\item $\omega_{\mu=1/2} \leq \omega_\mu \leq \omega_{\mu=\pi/\ln p}$ for $\mu\in \{1/2+it| t\in \mathbb{R}\}$ and  $v_p\geq \lb \sqrt{p}+1 \rb^2/2$.
\item $\omega_{\mu=\pi/\ln p} \leq \omega_\mu \leq \pi/\ell$ for $\mu = t + i\pi/\ln p$, $t \in \lsb t_1 , t_2 \rsb$ and $v_p\geq \lb \sqrt{p}+1 \rb^2/2$.
\end{enumerate}
\item \label{lpoint2} We have that
\be
\omega_{\mu=0}=\omega_{\mu=1} = 0,
\ee
so that $\mu=0,1$ in the pre-state plane correspond to the zero mode. Furthermore,~$\forall\  v_p$ we have
\be
\omega_\mu \neq 0 \mrm{\ for\ all\ } \mu\neq 0,1,
\ee 
(up to repetition by $2\pi i k/\ln p$, $k\in\mathbb{Z}$, in the imaginary direction), so no other points in the pre-state plane have zero frequency. Note that because $\omega_{\mu=0}=0$, the constant piece in the mode expansion of $X^a$ is $c_0^a+\tilde c_0^{a}$.

\item \label{lpoint3} Symmetries of $\omega_\mu$ in the pre-state plane:
\begin{enumerate}
\item $\omega_\mu=\omega_{1-\mu}$, $\forall\mu\in[0,1]$.
\item $\omega_{\frac{1}{2}+it} = \omega_{\frac{1}{2}-it}$, $\forall t\in \mathbb{R}$.
\item $\omega_{t+i\pi/\ln p}=\omega_{t-i\pi/\ln p}=-\omega_{1-t+i\pi/\ln p}=-\omega_{1-t-i\pi/\ln p}$, $\forall t\in (t_1,t_2)$, and $\omega_{t\pm i\pi/\ln p}=\omega_{1-t\pm i\pi/\ln p}$ for $t=t_{1,2}$.
\end{enumerate}
\end{enumerate}
\end{lemma}
\begin{proof}
To prove point \ref{lpoint1}, by looking at $\omega'(\mu)$ it is immediate to see that for $v_p\geq \lb \sqrt{p}-1 \rb^2/2$ and $\mu\in[0,1]$ frequency $\omega_\mu$ has a maximum at $\mu=1/2$, for $v_p\geq \lb \sqrt{p}-1 \rb^2/2$ and $\mu\in \{1/2+it|t\in\mathbb{R}\}$ $\omega(\mu)$ oscillates between the values $\omega_{\mu=1/2}$ and $\omega_{\mu=\pi/\ln p}$, and for $\mu\in \{t + i\pi/\ln p\ |\ t\in [t_1,t_2] \}$, $\omega(\mu)$ has a minimum at $t=0$, maxima at $t=t_{1,2}$ where the value $\pi/\ell$ is reached, and is monotonic between $0$ and~$t_{1,2}$.

To prove point \ref{lpoint2}, note that $\omega_\mu=0$ iff
\be
\frac{\lambda_\mu +v_p +\sqrt{\lambda_\mu^2+2 \lambda_\mu v_p}}{v_p} = 1,
\ee
which is equivalent to $\lambda_\mu=0$, i.e. $\mu=0,1$ up to repetition by $2\pi i k/\ln p$, $k\in\mathbb{Z}$.

Point \ref{lpoint3} follows by direct computation. 
\end{proof}

\begin{remark}
In the rest of the paper we will only consider the $[0,1]$ and critical line loci, that is we pick $v_p \geq (\sqrt{p}+1)^2/2$. The locus at $i\pi/ln p$ can be analyzed analogously, but we will not discuss it explicitly in the paper.  Furthermore, it is currently unclear what the representation theory content associated to the $i\pi/ln p$ locus should be, if anything.
\end{remark}

\begin{remark}[From loci to sectors: representation content of the $p2$-brane]
We will split the $[0,1]$ and $\{1/2+it|t\in\mathbb{R}\}$ loci into \emph{three} sectors, according to the unitary representation content: 1) the zero mode sector consisting of points $\{0,1\}$, 2) the complementary series sector $(0,1)$, and 3) the unitary principal series (or critical line sector) $\{1/2+it|t\in\mathbb{R}\}$.
\end{remark}

For the zero mode, $(0,1)$ interval, and critical line sector, the Fourier expansion of $X^a$ explicitly reads
\ba
X_i^a(\tau) &=& \mathfrak{p}^a\tau + \sum_{\mu\in[0,1]} \lb c^a_\mu e^{i\omega_\mu\tau} + \tilde{c}^a_\mu e^{-i\omega_\mu\tau} \rb \sqrt{p-1} p^{\mu\langle i,\infty \rangle} \\
&+& \sum_{t>0} \lb c^a_{\frac{1}{2}+it} e^{i\omega_\mu\tau} + \tilde{c}^a_{\frac{1}{2}+it} e^{-i\omega_\mu\tau} \rb \sqrt{p-1}  p^{\lb\frac{1}{2}+it\rb\langle i,\infty \rangle} \nn\\
&+& \sum_{t>0} \lb c^a_{\frac{1}{2}-it} e^{i\omega_\mu\tau} + \tilde{c}^a_{\frac{1}{2}-it} e^{-i\omega_\mu\tau} \rb \sqrt{p-1} p^{\lb\frac{1}{2}-it\rb\langle i,\infty \rangle}.\nn
\ea
Demanding $X^a$ be real determines the $\tilde c_\mu^a$ Fourier coefficients in terms of the $ c_\mu^a$ coefficients. We complex conjugate,
\ba
X_i^{a*}(\tau) &=& \mathfrak{p}^a\tau +  \sum_{\mu\in[0,1]} \lb c^{a*}_\mu e^{-i\omega_\mu\tau} + \tilde{c}^{a*}_\mu e^{i\omega_\mu\tau} \rb\sqrt{p-1} p^{\mu\langle i,\infty \rangle} \\
&+& \sum_{t>0} \lb \tilde{c}^{a*}_{\frac{1}{2}-it} e^{i\omega_\mu\tau} + c^{a*}_{\frac{1}{2}-it} e^{-i\omega_\mu\tau}  \rb \sqrt{p-1} p^{\lb\frac{1}{2}+it\rb\langle i,\infty \rangle}\nn\\
&+& \sum_{t>0} \lb \tilde{c}^{a*}_{\frac{1}{2}+it} e^{i\omega_\mu\tau} + c^{a*}_{\frac{1}{2}+it} e^{-i\omega_\mu\tau} \rb \sqrt{p-1} p^{\lb\frac{1}{2}-it\rb\langle i,\infty \rangle}.\nn
\ea
Then the reality condition $X_i^{a*}=X_i^a$ simplifies the Fourier expansion to
\ba
\label{exp249}
X_i^a(\tau) &=& \mathfrak{p}^a\tau +  \sum_{\mu\in[0,1]} \lb c^a_\mu e^{i\omega_\mu\tau} + c^{a*}_\mu e^{-i\omega_\mu\tau} \rb \phi_\mu(i) \\
&+& \sum_{\mu\in\{1/2+it|t\in\mathbb{R}\}} \lb c^a_{\mu} e^{i\omega_\mu\tau} + c^{a*}_{1-\mu} e^{-i\omega_\mu\tau} \rb \phi_\mu(i). \nn
\ea
With the reality condition, the momentum expansion is (remember that $\ell=2\ell_\tau$)
\ba
P^a_{i}(\tau) &=& \pd_\tau X^a_{i} \\
&=& \mathfrak{p}^a +  \sum_{\mu\in [0,1]} \lb e^{\frac{i\omega_\mu\ell}{2}} - e^{-\frac{i\omega_\mu\ell}{2}} \rb \lb c^a_\mu e^{i\omega_\mu\tau} - c^{a*}_\mu e^{-i\omega_\mu\tau} \rb\frac{1}{\ell} \phi_\mu(i) \nn \\
& & + \sum_{\mu\in \{1/2+it|t\in\mathbb{R}\}} \lb e^{\frac{i\omega_\mu\ell}{2}} - e^{-\frac{i\omega_\mu\ell}{2}} \rb \lb c^a_\mu e^{i\omega_\mu\tau} - c^{a*}_{1-\mu} e^{-i\omega_\mu\tau} \rb\frac{1}{\ell} \phi_\mu(i).\nn
\ea

\section{Poisson brackets and canonical commutators}
\label{sec4}

In this section we will derive the Poisson brackets for the Fourier coefficients $c_\mu^a$, $c_\mu^{a*}$, starting from the position-momentum Poisson brackets \eqref{eeqq2p14} -- \eqref{eeqq2p15}. We will then promote the Fourier coefficients to raising and lowering operators, and the Poisson brackets to commutators.

As we will see below, the position-momentum commutators will determine all the Fourier expansion commutators we are interested in, except for one of the zero mode commutators, which we will parameterize by a parameter denoted $\alpha$.

From here on throughout the rest of the paper we will set the length $\ell=1$.

\begin{remark}
The three sectors all commute with one another, because the $\omega_\mu$ frequency ranges of the three sectors do not overlap, according to Lemma \ref{lemmaomeganuprop}.
\end{remark}

Our aim will be to prove the following theorems.

\begin{thm}
\label{andthisisthm4}
For $\mu,\nu \in (0,1) $ we have
\ba
& &\lcb c_{\mu}^{a*}, c_{\nu}^b\rcb = \frac{1}{2} \frac{\eta^{ab}}{e^{-\frac{i\omega_\mu}{2}} - e^{\frac{i\omega_\mu}{2}} } \delta_{\mu + \nu-1}, \\
& &\lcb c^a_{\mu}, c^b_{\nu} \rcb = \lcb c^{a*}_{\mu}, c^{b*}_{\nu} \rcb = 0, \nn
\ea
for $\mu,\nu\in \lcb 1/2+it,t\in\mathbb{R} \rcb$ we have
\ba
& &\lcb c_{\mu}^{a*}, c_{\nu}^b\rcb = \frac{1}{2} \frac{\eta^{ab}}{e^{-\frac{i\omega_\mu}{2}} - e^{\frac{i\omega_\mu}{2}} } \delta_{\mu - \nu}, \\
& &\lcb c^a_{\mu}, c^b_{\nu} \rcb = \lcb c^{a*}_{\mu}, c^{b*}_{\nu} \rcb = 0, \nn
\ea
and for $\mu,\nu\in\{0,1\}$ we have
\ba
& &\{c_0^{a+},\mathfrak{p}^b \} = - \alpha\sqrt{p-1}\eta^{ab},\\
& &\lcb c^{a+}_{0} , c^{b+}_{0} \rcb = \lcb c^{a+}_{0} , c^{b+}_{1} \rcb = \lcb c^{a+}_{1} , c^{b+}_{1}\rcb = \lcb \mathfrak{p}^a,\mathfrak{p}^b \rcb = \lcb \mathfrak{p}^a , c_1^{b+}\rcb = 0,\nn\\
& & \Big\{c^{a+}_0 , c^{b-}_0 \Big\}  = \frac{-i\alpha\eta^{ab}}{\otoz} , \nn\\
& & \Big\{ c^{a+}_{0} , c^{b-}_{1} \Big\} = \Big\{c^{a+}_{1} , c^{b-}_{0} \Big\} = - \Big\{c_1^{a+}, c^{b-}_1 \Big\} = \frac{i\eta^{ab}}{\otoz}.\nn
\ea
\end{thm}
Here we have used the notation $x^\pm\coloneqq x \pm x^*$, and $\alpha$ is an arbitrary parameter whose value will not matter for the physical results. The parameter $\otoz$ is a zero mode (infrared) cutoff corresponding to the angular frequency $\omega_{\mu=0}=\omega_{\mu=1}$ for the zero mode, so that the commutators above are formally divergent (see Section \ref{seczeromm} for details); this will not affect the physics.

\begin{thm}
\label{andthisisthm5}
For $\mu,\nu\in \{0,1\}$ we have 
\be
\label{eqsoft34}
\lcb \mathfrak{p}^a, c_0^{b-} \rcb= \lcb \mathfrak{p}^a, c_0^{b-} + c_1^{b-} \rcb= \lcb c_\mu^{a-}, c_\nu^{b-} \rcb =  \OO\lb 1 \rb.
\ee
\end{thm}
In Theorem \ref{andthisisthm5} the $O(1)$ term on the right-hand side is in $\omega$, meaning that the brackets should have no diverging dependence of $\omega$ as $\omega\to0$.

\subsection{Inversion formulas for the $[0,1]$ locus}

The first task is to isolate the Fourier modes from $X_i^a$ and $P_i^a$. For $\nu\in[0,1]$ we have the inversion formula
\ba
\label{eq31}
\sum_{i\in V\lb T_p \rb} X_i^a(\tau) \phi_\nu(i) &=& \sum_{i\in V\lb T_p \rb} \lb \mathfrak{p}^a\tau +  \sum_{\mu\in[0,1]} \lb c^a_\mu e^{i\omega_\mu\tau} + c^{a*}_\mu e^{-i\omega_\mu\tau} \rb \phi_\mu(i) \rb \phi_\nu(i) \nn\\
&=& -\frac{1}{\sqrt{p-1}} \mathfrak{p}^a\tau\lb \delta_{\nu} + \delta_{\nu-1} \rb\\
& & -\sum_{\mu\in[0,1]} \lb c^a_\mu e^{i\omega_\mu\tau} + c^{a*}_\mu e^{-i\omega_\mu\tau} \rb \lb \delta_{\mu+\nu} + \delta_{\mu+\nu-1} \rb \nn \\
&=& -\frac{1}{\sqrt{p-1}} \mathfrak{p}^a\tau\lb \delta_{\nu} + \delta_{\nu-1} \rb \nn \\
& & - \lb c^a_{0} e^{i\omega_\nu\tau}\delta_{\nu} + c^{a*}_{0} e^{-i\omega_\nu\tau}\delta_\nu + c^a_{1-\nu} e^{i\omega_\nu\tau} + c^{a*}_{1-\nu} e^{-i\omega_\nu\tau}\rb,\nn
\ea
(where we have used that $\omega_\mu=\omega_{1-\mu}$, and orthonormality relations \eqref{eq223}), and
\ba
\label{eq32}
\sum_{i\in V\lb T_p \rb} P_i^a(\tau) \phi_\nu(i) &=& -\frac{1}{\sqrt{p-1}} \mathfrak{p^a}\lb \delta_{\nu} + \delta_{\nu-1} \rb - \lb e^{\frac{i\omega_\nu}{2}}- e^{-\frac{i\omega_\nu}{2}} \rb \times \\ 
&\times& \lb c^a_{0} e^{i\omega_\nu\tau}\delta_\nu - c^{a*}_{0} \delta_\nu e^{-i\omega_\nu\tau} + c^a_{1-\nu}  e^{i\omega_\nu\tau} - c^{a*}_{1-\nu} e^{-i\omega_\nu\tau} \rb. \nn
\ea
With inversion formulas \eqref{eq31}, \eqref{eq32} we can compute the Poisson brackets. We have 
\ba
\label{commut1}
0 &=& \sum_{i,j\in V\lb T_p \rb} \{ X_i^a, X_j^b \} \phi_\mu(i) \phi_\nu(j) \\
  &=& \Big\{ \frac{1}{\sqrt{p-1}} \mathfrak{p}^a\tau\lb \delta_{\mu} + \delta_{\mu-1} \rb + c^a_{0} e^{i\omega_\mu\tau}\delta_{\mu} + c^{a*}_{0} e^{-i\omega_\mu\tau}\delta_\mu + c^a_{1-\mu} e^{i\omega_\mu\tau} + c^{a*}_{1-\mu} e^{-i\omega_\mu\tau},\nn\\
  & & \frac{1}{\sqrt{p-1}}\mathfrak{p}^b\tau\lb \delta_{\nu} + \delta_{\nu-1} \rb + c^b_{0} e^{i\omega_\nu\tau}\delta_{\nu} + c^{b*}_{0} e^{-i\omega_\nu\tau}\delta_\nu + c^b_{1-\nu} e^{i\omega_\nu\tau} + c^{b*}_{1-\nu} e^{-i\omega_\nu\tau} \Big\}, \nn
\ea
as well as
\ba
\label{commut2}
0 &=& \sum_{i,j \in V\lb T_p \rb } \{ P_i^a, P_j^b \} \phi_\mu(i) \phi_\nu(j) \\
  &=& \Big\{ \frac{1}{\sqrt{p-1}} \mathfrak{p^a}\lb \delta_{\mu} + \delta_{\mu-1} \rb + \lb e^{\frac{i\omega_\mu}{2}} - e^{-\frac{i\omega_\mu}{2}} \rb \Big( c^a_{0} e^{i\omega_\mu\tau}\delta_{\mu} - c^{a*}_{0} e^{-i\omega_\mu\tau}\delta_\mu\nn\\
& & + c^a_{1-\mu} e^{i\omega_\mu \tau} - c^{a*}_{1-\mu} e^{-i\omega_\mu\tau} \Big) ,\nn\\
& & \frac{1}{\sqrt{p-1}} \mathfrak{p^b}\lb \delta_{\nu} + \delta_{\nu-1} \rb + \lb e^{\frac{i\omega_\nu}{2}} - e^{-\frac{i\omega_\nu}{2}} \rb \Big( c^b_{0} e^{i\omega_\nu\tau} \delta_{\nu} - c^{b*}_{0}  e^{-i\omega_\nu\tau}\delta_\nu\nn\\
& & + c^b_{1-\nu} e^{i\omega_\nu\tau} - c^{b*}_{1-\nu} e^{-i\omega_\nu\tau} \Big) \Big\}. \nn
\ea
Finally, using Eq. \eqref{eeqq2p14} evaluates the following sum as
\ba
\sum_{i,j\in V\lb T_p \rb}\{ X_i^a, P_j^b \} \phi_\mu(i) \phi_\nu(j) &=& \sum_{i,j\in V\lb T_p \rb} \eta^{ab} \delta_{ij}\phi_\mu(i) \phi_\nu(j) \\
&=& \sum_{i\in V\lb T_p \rb} \eta^{ab}\phi_\mu(i) \phi_\nu(i) \nn\\
&=& - \eta^{ab} \lb \delta_{\mu+\nu} + \delta_{\mu+\nu-1} \rb,\nn
\ea
so that
\ba
\label{commut3}
& &-\eta^{ab} \lb \delta_{\mu+\nu} + \delta_{\mu+\nu-1} \rb = \Big\{ \frac{1}{\sqrt{p-1}} \mathfrak{p}^a\tau\lb \delta_{\mu} + \delta_{\mu-1} \rb + c^a_{0} e^{i\omega_\mu\tau}\delta_{\mu} \\
& & + c^{a*}_{0} e^{-i\omega_\mu\tau}\delta_\mu + c^a_{1-\mu} e^{i\omega_\mu\tau} + c^{a*}_{1-\mu} e^{-i\omega_\mu\tau} ,\nn\\ 
& &\frac{1}{\sqrt{p-1}} \mathfrak{p^b}\lb \delta_{\nu} + \delta_{\nu-1} \rb +
\lb e^{\frac{i\omega_\nu}{2}} - e^{-\frac{i\omega_\nu}{2}} \rb \Big(  c^b_{0}  e^{i\omega_\nu\tau}\delta_{\nu} - c^{b*}_{0}  e^{-i\omega_\nu\tau}\delta_\nu  \nn \\
& &+ c^b_{1-\nu} e^{i\omega_\nu\tau} - c^{b*}_{1-\nu}  e^{-i\omega_\nu\tau} \Big) \Big\}. \nn
\ea
We will use Eqs. \eqref{commut1}, \eqref{commut2}, \eqref{commut3} for the $(0,1)$ and zero sector commutation relations, and a slightly modified version for the critical line commutation relations.

\subsection{The $(0,1)$ sector}

If $\mu,\nu\neq 0,1$ then $e^{\frac{i\omega_{\mu,\nu}}{2}} - e^{-\frac{i\omega_{\mu,\nu}}{2}} \neq 0$, and Eqs. \eqref{commut1}, \eqref{commut2}, \eqref{commut3} become
\ba
\label{thisis310}
\lcb  c^a_{1-\mu} e^{i\omega_\mu\tau} + c^{a*}_{1-\mu} e^{-i\omega_\mu\tau},c^b_{1-\nu} e^{i\omega_\nu\tau} + c^{b*}_{1-\nu} e^{-i\omega_\nu\tau} \rcb &=& 0, \\
\lcb c^a_{1-\mu} e^{i\omega_\mu \tau} - c^{a*}_{1-\mu} e^{-i\omega_\mu\tau}, c^b_{1-\nu} e^{i\omega_\nu\tau} - c^{b*}_{1-\nu} e^{-i\omega_\nu\tau} \rcb &=& 0, \nn\\
\lcb  c^a_{1-\mu} e^{i\omega_\mu\tau} + c^{a*}_{1-\mu} e^{-i\omega_\mu\tau} ,
 c^b_{1-\nu} e^{i\omega_\nu\tau} - c^{b*}_{1-\nu}  e^{-i\omega_\nu\tau} \rcb &=& - \frac{\eta^{ab} \delta_{\mu+\nu-1}}{e^{\frac{i\omega_\nu}{2}} - e^{-\frac{i\omega_\nu}{2}}}. \nn
\ea

The terms in the first line above proportional to $e^{\pm i\lb \omega_\mu+\omega_\nu \rb \tau }$ give
\be
\lcb c_\mu^a, c_\nu^b \rcb = \lcb c_\mu^{a*}, c_{\nu}^{b*} \rcb =0 
\ee
for all $\mu,\nu\in (0,1)$, while the cross-terms in the first two lines of Eq. \eqref{thisis310} are equivalent~to
\ba
\label{eeqq313}
\lcb c_\mu^{a*},c_\nu^b \rcb &=& 0 \quad \mrm{\ if\ } \quad \omega_\mu\neq\omega_\nu,\\
\lcb c_\mu^{a*},c_\nu^b \rcb + \lcb c_\mu^{a},c_\nu^{b*} \rcb &=& 0 \quad \mrm{\ if\ } \quad \omega_\mu=\omega_\nu, \nn
\ea
so that for $\omega_\mu=\omega_\nu$ the third line in Eq. \eqref{thisis310} is equivalent to
\be
\label{eeqq314}
\lcb c_\mu^{a*}, c_\nu^{b} \rcb - \lcb c_\mu^a, c_\nu^{b*} \rcb = - \frac{\eta^{ab} \delta_{\mu+\nu-1}}{e^{\frac{i\omega_\nu}{2}} - e^{-\frac{i\omega_\nu}{2}}}.
\ee
Remembering that $\omega_\mu=\omega_\nu$ for $\mu=\nu$ or $\mu=1-\nu$, Eqs. \eqref{eeqq313}, \eqref{eeqq314} imply that
\ba
\lcb c_\mu^{a*},c_\mu^b \rcb &=& 0 \quad \mrm{\ for\ }\quad \mu\neq 1/2,\\
2\lcb c_\mu^{a*}, c_\nu^{b} \rcb  &=& - \frac{\eta^{ab} \delta_{\mu+\nu-1}}{e^{\frac{i\omega_\nu}{2}} - e^{-\frac{i\omega_\nu}{2}}}. \nn
\ea
Thus we conclude, for the $(0,1)$ sector,
\ba
\lcb c_{\mu}^{a*}, c_{\nu}^b\rcb &=& \frac{1}{2} \frac{\eta^{ab}}{e^{-\frac{i\omega_\mu}{2}} - e^{\frac{i\omega_\mu}{2}} } \delta_{\mu + \nu -1}, \\
\lcb c^a_{\mu}, c^b_{\nu} \rcb &=& \lcb c^{a*}_{\mu}, c^{b*}_{\nu} \rcb = 0. \nn
\ea

\subsection{Critical line sector}

Consider now $\nu\in \lcb 1/2 + it\, |\, t\in \mathbb{R} \rcb$. The mode expansions are

\ba
\sum_{i\in V\lb T_p \rb} X_i^a(\tau) \phi_\nu(i) &=& \sum_{i\in V\lb T_p \rb} \sum_{\mu\in \lcb 1/2 + it\, |\, t\in \mathbb{R} \rcb} \lb c^a_\mu e^{i\omega_\mu\tau} + c^{a*}_{1-\mu} e^{-i\omega_\mu\tau} \rb \phi_\mu(i) \phi_\nu(i) \nn\\
&=&  -\sum_{\mu\in\lcb 1/2 + it\, |\, t\in \mathbb{R} \rcb} \lb c^a_\mu e^{i\omega_\mu\tau} + c^{a*}_{1-\mu} e^{-i\omega_\mu\tau} \rb \delta_{\mu+\nu-1}  \\
&=&  - \lb  c^a_{1-\nu} e^{i\omega_\nu\tau} + c^{a*}_{\nu} e^{-i\omega_\nu\tau}\rb,\nn
\ea
and
\be
\sum_{i\in V\lb T_p \rb} P_i^a(\tau) \phi_\nu(i) = - \lb e^{\frac{i\omega_\nu}{2}} - e^{-\frac{i\omega_\nu}{2}} \rb \lb  c^a_{1-\nu} e^{i\omega_\nu\tau} - c^{a*}_{\nu}  e^{-i\omega_\nu\tau}\rb.
\ee
Then the commutation relations \eqref{commut1}, \eqref{commut2}, \eqref{commut3} modify into
\ba
0 &=& \sum_{i,j\in V\lb T_p \rb} \lcb X_i^a, X_j^b \rcb \phi_\mu(i)\phi_\nu(j)\nn\\
\label{eq414}
&=& \lcb c^a_{1-\mu} e^{i\omega_\mu\tau} + c^{a*}_{\mu} e^{-i\omega_\mu\tau} , c^b_{1-\nu} e^{i\omega_\nu\tau} + c^{b*}_{\nu} e^{-i\omega_\nu\tau} \rcb,\\
\label{eq415}
0 &=& \sum_{i,j \in V\lb T_p \rb } \{ P_i^a, P_j^b \} \phi_\mu(i) \phi_\nu(j)\nn \\
  &=& \lcb c^a_{1-\mu} e^{i\omega_\mu\tau} - c^{a*}_{\mu} e^{-i\omega_\mu\tau} , c^b_{1-\nu} e^{i\omega_\nu\tau} - c^{b*}_{\nu} e^{-i\omega_\nu\tau} \rcb, \\
\hspace{-3cm}-\eta^{ab} \delta_{\mu+\nu-1} &=& \sum_{i,j\in V\lb T_p \rb}\{ X_i^a, P_j^b \} \phi_\mu(i) \phi_\nu(j)\nn\\
\label{eq416}
&=& \lb e^{\frac{i\omega_\nu}{2}} - e^{-\frac{i\omega_\nu}{2}} \rb \lcb c^a_{1-\mu} e^{i\omega_\mu\tau} + c^{a*}_{\mu} e^{-i\omega_\mu\tau}, c^b_{1-\nu} e^{i\omega_\nu\tau} - c^{b*}_{\nu} e^{-i\omega_\nu\tau} \rcb.
\ea
The terms proportional to $e^{i\lb\omega_\mu+\omega_\nu\rb\tau}$ and $e^{-i\lb\omega_\mu+\omega_\nu\rb\tau}$ in Eq. \eqref{eq414} give (after relabeling)
\be
\lcb c^a_{\mu}, c^b_{\nu} \rcb = \lcb c^{a*}_{\mu}, c^{b*}_{\nu} \rcb = 0,
\ee
while the cross-terms give
\be
\lcb c^a_{1-\mu}, c^{b*}_{\nu} \rcb e^{i\lb \omega_\mu - \omega_\nu \rb \tau} + \lcb c^{a*}_{\mu}, c^{b}_{1-\nu} \rcb e^{i\lb -\omega_\mu + \omega_\nu \rb \tau} = 0,
\ee
which implies (here $\mu,\nu=1/2+it_{\mu,\nu}$)
\ba
\lcb c^{a*}_\mu, c^b_\nu \rcb &=& 0 \quad \mrm{if} \quad t_\mu \neq \pm t_\nu, \\
\lcb c^{a*}_{1-\mu}, c^b_\nu \rcb + \lcb c^a_\mu, c^{b*}_{1-\nu} \rcb  &=& 0 \quad \mrm{if} \quad t_\mu = \pm t_\nu. \nn
\ea
Eq. \eqref{eq415} is satisfied trivially by the above. Eq. \eqref{eq416} gives
\ba
-\eta^{ab} \delta_{\mu+\nu-1} = \lb e^{\frac{i\omega_\nu}{2}} - e^{-\frac{i\omega_\nu}{2}} \rb \lb -\lcb c^a_{1-\mu},  c^{b*}_{\nu} \rcb + \lcb c^{a*}_{\mu}, c^b_{1-\nu} \rcb \rb.
\ea
i.e.
\be
\eta^{ab} \delta_{\mu-\nu} = -2\lb e^{\frac{i\omega_\nu}{2}} - e^{-\frac{i\omega_\nu}{2}} \rb \lcb  c^{a*}_{\mu} , c^b_{\nu} \rcb,
\ee
so that we can conclude
\be
\lcb c_\mu^{a*}, c^b_{1-\mu} \rcb = 0, \quad \mu\neq \frac{1}{2}.
\ee
Therefore, for the entire critical line sector, the Poisson brackets are
\ba
\label{eq154}
& &\lcb c_{\mu}^{a*}, c_{\nu}^b\rcb = \frac{1}{2} \frac{\eta^{ab}}{e^{-\frac{i\omega_\mu}{2}} - e^{\frac{i\omega_\mu}{2}} } \delta_{\mu - \nu}, \\
& &\lcb c^a_{\mu}, c^b_{\nu} \rcb = \lcb c^{a*}_{\mu}, c^{b*}_{\nu} \rcb = 0.\nn
\ea

\subsection{Zero mode sector}
\label{seczeromm}
We have that $\omega_{\mu=0}=\omega_{\mu=1}=0$, however note that the denominator in Eq. \eqref{eq154} diverges when $\omega\to 0$, therefore we must be careful when computing the zero mode commutators. In particular, due to the appearance of an extra contribution at $\mu=0$ in the orthonormality relations \eqref{eq221}, one cannot take the $\omega\to0$ limit of the $(0,1)$ sector results.\\

\begin{remark}
The extra contribution $\delta_{\mu+\nu}$ in the inner product \eqref{eq221} is necessary in order for the zero mode contributions to momentum and angular momentum to obey the Poincar\'e algebra, as we will show below.
\end{remark}

We will take $\omega_{\mu=0}=\omega_{\mu=1}$ to equal a parameter $\otoz$, and compute the commutators to leading orders in $\otoz$. The structure of commutators will be a power series in $\otoz$, starting with a potentially diverging term. We will keep track of the diverging and constant terms, and ignore the rest of the terms, which vanish as $\otoz\to 0$.

For the zero mode sector only it will be useful to introduce the notation
\be
c_\mu^{a\pm} \coloneqq c_\mu^{a} \pm c_\mu^{a*},
\ee
where $a$ is the usual spacetime index and $\mu\in\{0,1\}$ (note that the $\pm$ superscript does not refer to the lightcone spacetime directions). Once we promote the Fourier coefficients to operators we will instead denote 
\be
c_\mu^{a\pm} \coloneqq c_\mu^{a} \pm c_\mu^{a\dagger}.
\ee
We now proceed with the computation. Eq. \eqref{commut1} is
\be
0 = \Big\{ \frac{\mathfrak{p}^a\tau}{\sqrt{p-1}} \lb \delta_{\mu} + \delta_{\mu-1} \rb + c^{a+}_{0} \delta_\mu + c^{a+}_{1-\mu}, \frac{\mathfrak{p^b\tau}}{\sqrt{p-1}} \lb \delta_{\nu} + \delta_{\nu-1} \rb + c^{b+}_{0} \delta_\nu + c^{b+}_{1-\nu} \Big\}, 
\ee
and for the parameter values $\mu,\nu\in\{0,1\}$ it splits into
\ba
0 &=& \lcb \frac{\mathfrak{p}^a\tau}{\sqrt{p-1}}  +  c^{a+}_{0} , \frac{\mathfrak{p}^b\tau}{\sqrt{p-1}}  + c^{b+}_{0} \rcb, \\
0 &=& \lcb \frac{\mathfrak{p}^a\tau}{\sqrt{p-1}}  +  c^{a+}_{0} , \frac{\mathfrak{p}^b\tau}{\sqrt{p-1}}  + c^{b+}_{0} + c^{b+}_{1} \rcb, \nn \\
%0 &=& \lcb c^a_{0} + c^{a*}_{0} + c^a_{1} + c^{a*}_{1} , c^b_{0} + c^{b*}_{0} \rcb, \\
0 &=& \lcb \frac{\mathfrak{p}^a\tau}{\sqrt{p-1}}  + c^{a+}_{0}  + c^{a+}_{1} , \frac{\mathfrak{p}^b\tau}{\sqrt{p-1}}  + c^{b+}_{0} + c^{b+}_{1} \rcb, \nn
\ea
so that
\ba
\label{eq161}
& &0=\lcb  c^{a+}_{0} , c^{b+}_{0} \rcb = \lcb  c^{a+}_{0} , c^{b+}_{1} \rcb = \lcb  c^{a+}_{1} , c^{b+}_{1}\rcb,\\
& & 0 = \lcb \mathfrak{p}^a,\mathfrak{p}^b \rcb = \lcb \mathfrak{p}^a , c_1^{b+} \rcb.\nn
\ea
And Eq. \eqref{commut3} is
\ba
\hspace{-1.9cm}&-&\eta^{ab} \lb \delta_{\mu+\nu} + \delta_{\mu+\nu-1} \rb = \Big\{ \frac{\mathfrak{p}^a\tau}{\sqrt{p-1}} \lb \delta_{\mu} + \delta_{\mu-1} \rb + c^{a+}_{0} \delta_\mu + c^{a+}_{1-\mu} ,\\
\hspace{-1.9cm}& & \frac{\mathfrak{p}^b}{\sqrt{p-1}} \lb \delta_{\nu} + \delta_{\nu-1} \rb + \lb e^{\frac{i \otoz}{2}} - e^{-\frac{i \otoz}{2}} \rb \lb c^{b-}_{0} \delta_{\nu} + c^{b-}_{1-\nu} \rb \Big\}. \nn
\ea
The $\tau$ term implies
\be
\label{eq338O1}
\lcb \mathfrak{p}^a, c_0^{b-} \rcb= \lcb \mathfrak{p}^a, c_0^{b-} + c_1^{b-} \rcb= \OO\lb 1 \rb,
\ee
where $\OO\lb 1 \rb$ refers to the fact that the bracket should have no diverging contribution in $\otoz$.\footnote{Strictly speaking, any divergence milder than $1/\otoz$ would also be allowed.} The equation then splits into
\ba
\label{eq336split}
0 &=& \frac{1}{\sqrt{p-1}} \lcb c_0^{a+},\mathfrak{p}^b \rcb + \lb e^{\frac{i \otoz}{2}} - e^{-\frac{i \otoz}{2}} \rb \lcb c_0^{a+} , c_0^{b-}  \rcb, \\
-\eta^{ab} &=& \frac{1}{\sqrt{p-1}} \lcb c_0^{a+},\mathfrak{p}^b \rcb + \lb e^{\frac{i \otoz}{2}} - e^{-\frac{i \otoz}{2}} \rb \lcb c_0^{a+} ,c_0^{b-} + c_1^{b-} \rcb, \nn \\
-\eta^{ab} &=& \frac{1}{\sqrt{p-1}} \lcb c_0^{a+},\mathfrak{p}^b \rcb + \lb e^{\frac{i \otoz}{2}} - e^{-\frac{i \otoz}{2}} \rb  \lcb  c_0^{a+} + c_1^{a+},c_0^{b-} \rcb, \nn\\
-\eta^{ab} &=& \frac{1}{\sqrt{p-1}} \lcb c_0^{a+},\mathfrak{p}^b \rcb + \lb e^{\frac{i \otoz}{2}} - e^{-\frac{i \otoz}{2}} \rb \lcb c_0^{a+} +c_1^{a+}, c_0^{b-} + c_1^{b-}  \rcb. \nn
\ea
These relations simplify to
\ba
\label{eq164}
0 &=& \frac{1}{\sqrt{p-1}} \{c_0^{a+},\mathfrak{p}^b \} + \lb e^{\frac{i \otoz}{2}} - e^{-\frac{i \otoz}{2}} \rb \lcb c_0^{a+} , c_0^{b-} \rcb  , \\
-\eta^{ab} &=& \lb e^{\frac{i \otoz}{2}} - e^{-\frac{i \otoz}{2}} \rb \lcb c_0^{a+} , c_1^{b-} \rcb , \nn \\
-\eta^{ab} &=& \lb e^{\frac{i \otoz}{2}} - e^{-\frac{i \otoz}{2}} \rb \lcb c_1^{a+}, c_0^{b-} \rcb  , \nn\\
0 &=& \lb e^{\frac{i \otoz}{2}} - e^{-\frac{i \otoz}{2}} \rb \lcb c_1^{a+}, c^{b-}_{0}  + c^{b-}_{1} \rcb, \nn
\ea
and the last equation can be further reduced to 
\be
\label{eeq341}
\eta^{ab} = \lb e^{\frac{i \otoz}{2}} - e^{-\frac{i \otoz}{2}} \rb \lcb c_1^{a+}, c_1^{b-} \rcb.
\ee
Now, let's parameterize the $\{c_0^{a+},\mathfrak{p}^b \}$ Poisson bracket by a parameter $\alpha$, i.e.
\be
\label{defalpha342}
\{c_0^{a+},\mathfrak{p}^b \} \eqqcolon - \sqrt{p-1}\alpha \eta^{ab}.
\ee
Equations \eqref{eq161} and \eqref{eq164} -- \eqref{eeq341}, together with the definition \eqref{defalpha342} and Eq. \eqref{eq338O1}, are necessary conditions on the zero sector brackets. We can rewrite them as
\ba
\label{bigPoisson}
& &\lcb c^{a+}_{0} , c^{b+}_{0} \rcb = \lcb c^{a+}_{0} , c^{b+}_{1} \rcb = \lcb c^{a+}_{1} , c^{b+}_{1}\rcb = \lcb \mathfrak{p}^a,\mathfrak{p}^b \rcb = \lcb \mathfrak{p}^a , c_1^{b+}\rcb = 0,\\
& &\{c_0^{a+},\mathfrak{p}^b \} = - \sqrt{p-1}\eta^{ab} \alpha, \nn\\
& & \Big\{c^{a+}_0 , c^{b-}_0 \Big\}  = -\frac{i\alpha\eta^{ab}}{\otoz} , \nn\\
& & \Big\{ c^{a+}_{0} , c^{b-}_{1} \Big\} = \Big\{c^{a+}_{1} , c^{b-}_{0} \Big\} = - \Big\{c_1^{a+}, c^{b-}_1 \Big\} = \frac{i\eta^{ab}}{\otoz},\nn\\
& & \lcb \mathfrak{p}^a, c_0^{b-} \rcb= \lcb \mathfrak{p}^a, c_0^{b-} + c_1^{b-} \rcb= \OO\lb 1 \rb. \nn
\ea

The momentum brackets in Eq. \eqref{commut2} are
\ba
0&=&\Big\{ \frac{\mathfrak{p}^a}{\sqrt{p-1}}  + \lb e^{\frac{i \otoz}{2}} - e^{-\frac{i \otoz}{2}} \rb c^{a-}_{0} , \frac{\mathfrak{p}^b}{\sqrt{p-1}} + \lb e^{\frac{i \otoz}{2}} - e^{-\frac{i \otoz}{2}} \rb c^{b-}_{0} \Big\},\nn\\
0&=& \Big\{ \frac{\mathfrak{p}^a}{\sqrt{p-1}} + \lb e^{\frac{i \otoz}{2}} - e^{-\frac{i \otoz}{2}} \rb c^{a-}_{0} , \frac{\mathfrak{p}^b}{\sqrt{p-1}} + \lb e^{\frac{i \otoz}{2}} - e^{-\frac{i \otoz}{2}} \rb \times\nn \\
& & \times \lb c^{b-}_{0} + c^{b-}_{1} \rb \Big\},\nn \\ 
0&=& \Big\{ \frac{\mathfrak{p}^a}{\sqrt{p-1}} + \lb e^{\frac{i \otoz}{2}} - e^{-\frac{i \otoz}{2}} \rb  \lb c^{a-}_{0} + c^{a-}_{1} \rb , \frac{\mathfrak{p}^b}{\sqrt{p-1}} + \nn\\
& &\lb e^{\frac{i \otoz}{2}} - e^{-\frac{i \otoz}{2}} \rb  \lb c^{b-}_{0} + c^{b-}_{1} \rb\Big\},
\ea
and are automatically satisfied in the limit ~$\omega\to 0$, given the last line of Eq. \eqref{bigPoisson}.

We have finished proving Theorem \ref{andthisisthm4}.

Note that the last line in Eq. \eqref{bigPoisson} only constrains the leading power of $\omega$ in the bracket. In fact, Eqs. \eqref{bigPoisson} put similar constraints on the $\{c^{a-}_\mu, c^{b-}_\nu\}$ commutators for $\mu,\nu\in\{0,1\}$, in the following way. Let's define (here $\mu\in\{0,1\}$)
\ba
\label{eq342cmu01}
c^{a\pm}_\mu(\tau) &\coloneqq& c_\mu^a e^{i\omega \tau} \pm c_\mu^{a*} e^{-i\omega\tau} \\
&=& c^{a\pm}_\mu + ic_\mu^{a\mp}\omega\,\tau + \OO\lb \omega^2 \rb, \nn
\ea
so that the time-independent $c^{a\pm}_\mu$'s defined above are just
\be
c^{a\pm}_\mu(\tau=0) = c^{a\pm}_\mu.
\ee
Expanding the commutators to linear order in $\omega$, we have
\ba
\lcb c^{a+}_{\mu}(\tau) , c^{b+}_{\nu}(\tau) \rcb &=& \lcb c^{a+}_{\mu} + i \omega \tau c_\mu^{a-}, c^{b+}_{\nu} + i \omega \tau c_\nu^{b-} \rcb + \dots \\
&=& \lcb c^{a+}_{\mu} , c^{b+}_{\nu} \rcb +i\omega \tau \lb \lcb c^{a+}_\mu , c^{b-}_\nu \rcb + \lcb c^{a-}_\mu , c^{b+}_\nu \rcb \rb + \dots \nn\\
\lcb c^{a+}_{\mu}(\tau) , c^{b-}_{\nu}(\tau) \rcb &=& \lcb c^{a+}_{\mu} + i \omega \tau c_\mu^{a-}, c^{b-}_{\nu} + i \omega \tau c_\nu^{b+} \rcb + \dots \nn\\
&=& \lcb c^{a+}_{\mu} , c^{b-}_{\nu} \rcb +i\omega \tau \lb \lcb c^{a+}_\mu , c^{b+}_\nu \rcb + \lcb c^{a-}_\mu , c^{b-}_\nu \rcb \rb + \dots, \nn
\ea
and a similar expression holds for $\lcb c^{a-}_{\mu}(\tau) , c^{b-}_{\nu}(\tau) \rcb$, with the $\dots$ denoting higher terms. In order for the computation above yielding \eqref{bigPoisson} to remain valid, we must have the linear term not contaminate the leading piece. This is a nontrivial condition only for the $\lcb c^{a+}_{\mu}(\tau) , c^{b-}_{\nu}(\tau) \rcb$ commutator, yielding
\be
\lcb c^{a-}_\mu, c^{b-}_\nu\rcb = \OO(1),
\ee
for $\mu,\nu\in\{0,1\}$. We have finished proving Theorem \ref{andthisisthm5}.

\subsection{Promoting to canonical commutators}
\label{subsecpromoting}

We now promote the Poisson brackets from Theorem \ref{andthisisthm4} to commutators, remembering that when doing so an extra $i$ appears, $\{\cdot,\cdot\}\to -i \lsb \cdot,\cdot \rsb$.

\begin{remark}
Away from the zero mode the commutators are real and equal
\ba
\label{eq176}
\lsb c_{\mu}^{a}, c_{\nu}^{b\dagger}\rsb &=& \frac{1}{2} \frac{i\eta^{ab}}{ e^{\frac{i\omega_\mu}{2}} - e^{-\frac{i\omega_\mu}{2}}} \Delta_{\mu,\nu} \\
&=& \frac{\eta^{ab}}{4\sin \frac{\omega_\mu}{2}} \Delta_{\mu,\nu}, \nn\\
\lsb c^a_{\mu}, c^b_{\nu} \rsb &=& \lsb c^{a\dagger}_{\mu}, c^{b\dagger}_{\nu} \rsb = 0, \nn
\ea
where
\be
\label{eq349pairing}
\Delta_{\mu,\nu} \coloneqq\begin{cases}
\delta_{\mu+\nu-1}, \qquad \mu,\nu \in (0,1) \\
\delta_{\mu-\nu}, \qquad \mu,\nu \in \lcb \frac{1}{2} + i t, t\in\mathbb{R} \rcb
\end{cases}.
\ee
\end{remark}

\begin{remark}
The zero mode commutators are 
\ba
\label{zerozmcommutators}
& &\lsb c_0^{a+},\mathfrak{p}^b \rsb = - i \alpha\sqrt{p-1}\eta^{ab},\\
& &\lsb c^{a+}_{0} , c^{b+}_{0} \rsb = \lsb c^{a+}_{0} , c^{b+}_{1} \rsb = \lsb c^{a+}_{1} , c^{b+}_{1}\rsb = \lsb \mathfrak{p}^a,\mathfrak{p}^b \rsb = \lsb \mathfrak{p}^a , c_1^{b+}\rsb = 0, \nn\\
& & \lsb c^{a+}_0 , c^{b-}_0 \rsb   = \frac{\alpha\eta^{ab}}{\otoz} , \nn\\
& & \lsb  c^{a+}_{0} , c^{b-}_{1} \rsb  = \lsb c^{a+}_{1} , c^{b-}_{0} \rsb  = - \lsb c_1^{a+}, c^{b-}_1\rsb  =  \frac{-\eta^{ab}}{\otoz},\nn\\
& & \lsb \mathfrak{p}^a, c_0^{b-} \rsb= \lsb \mathfrak{p}^a, c_0^{b-} + c_1^{b-} \rsb = \lsb c^{a-}_\mu, c^{b-}_\nu\rsb = \OO(1). \nn
\ea
\end{remark}

We will leave parameter $\alpha$ undetermined, and as we will see, none of the physical results (Poincar\'e invariance, etc.) will depend on it.\footnote{Note that from the 4th line in Eq. \eqref{zerozmcommutators} the modes at $\mu=0$ and $\mu=1$ are coupled, however one can un-couple them by considering instead the linear combinations (say for $\alpha=1$) $c_0^{a\pm} + c_1^{a\pm}$, $c_0^{a\pm} - c_1^{a\pm}$.}

\begin{remark} 
We will interpret the operators $c_\mu^{a}$, $c_\mu^{a\dagger}$ on the $(0,1)$ interval and critical line as annihilation and creation operators, respectively. In order for this interpretation to make sense, the right-hand side of Eq. \eqref{eq176} needs to be positive, to avoid negative norm states in the spatial directions oscillations. This holds, because from Lemma \ref{lemmaomeganuprop} we have $0<\omega_\mu<\pi$ for all unitary loci in the complex plane. Note also that from Eq. \eqref{eq349pairing} the creation/annihilation pairing is $\mu\leftrightarrow\mu$ on the critical line, but $\mu\leftrightarrow 1-\mu$ on the $(0,1)$ interval (so that the two pairings agree at $\mu=1/2$).
\end{remark}

\begin{remark}[Rescaling the $c$'s]
We note that
\be
\label{eqomegamurel349}
\lb e^{\frac{i\omega_\mu}{2}} - e^{-\frac{i\omega_\mu}{2}} \rb^2 = \frac{2}{v_p} \lb p^\mu + p^{1-\mu} - p - 1 \rb
\ee 
so we can rewrite the nonvanishing commutation relations~as
\ba
\lsb c_{\mu}^{a}, c_{\nu}^{b\dagger}\rsb &=& \frac{1}{2} \frac{\eta^{ab}}{ \sqrt{\frac{2}{v_p}\lb 1 + p - p^\mu - p^{1-\mu} \rb} } \Delta_{\mu,\nu}.
\ea
We can thus normalize the states by introducing operators $b_\mu^a$, $b_\mu^{a\dagger}$ such that
\ba
c^a_\mu &\eqqcolon& \frac{1}{\sqrt 2} \frac{b_\mu^a}{\lsb \frac{2}{v_p} \lb 1 + p - p^\mu - p^{1-\mu} \rb \rsb^\frac{1}{4}}, \quad \mu\neq 0,1, \\  
c^{a\dagger}_\mu &\eqqcolon& \frac{1}{\sqrt 2} \frac{b_\mu^{a\dagger}}{\lsb \frac{2}{v_p} \lb 1 + p - p^\mu - p^{1-\mu} \rb \rsb^\frac{1}{4}}\nn,
\ea
so that we have
\be
\lsb b_\mu^a, b_\nu^{b\dagger} \rsb = \eta^{ab}\Delta_{\mu,\nu}, \quad \mu,\nu\neq 0,1.
\ee
\end{remark}

\begin{remark}[$p2$-brane sectors] 
Writing $\mu$ in the complementary and critical sectors as $\mu=1/2+t$ and $\mu=1/2+it$ respectively, each sector will have two sub-sectors. For the $(0,1)$ interval:
\begin{itemize}
\item Raising operator $b^\dagger_{\mu}$, lowering operator $b_{1-\mu}$, for $t>0$.
\item Raising operator $b^\dagger_{\mu}$, lowering operator $b_{1-\mu}$, for $t<0$.
\end{itemize}
And for the critical line:
\begin{itemize}
\item Raising operator $b^\dagger_{\mu}$, lowering operator $b_{\mu}$, for $t>0$.
\item Raising operator $b^\dagger_{\mu}$, lowering operator $b_{\mu}$, for $t<0$.
\end{itemize}
\end{remark}

\subsection{Interpretation of the commutation relations}
\label{sec36intep}
The commutation relations derived above, as well as the position-momentum commutator relations and so on, are equations in the operator sense, by acting on appropriate test functions (that is on functions which admit an expansion in Fourier modes on the tree). We now illustate this. Consider for example the position-momentum commutator, we have the following remark.

\begin{remark}
\label{remarkopaction19}
The commutator $\lsb X_i^a,P^b \rsb$ acts as $i\eta^{ab}\delta_{i,\cdot}$ on appropriate test functions.
\end{remark}

Let's prove Remark \ref{remarkopaction19} by direct computation. The commutator separates into the three sectors as
\be
\lsb X_i^a,P_j^b \rsb = \lsb X_i^a,P_j^b \rsb_\mrm{z.m.} + \lsb X_i^a,P_j^b \rsb_\mrm{(0,1)} + \lsb X_i^a,P_j^b \rsb_\mrm{crit}.
\ee
We have
\ba
\lsb X_i^a,P_j^b \rsb_\mrm{(0,1)} &=& \sum_{\mu,\nu\in (0,1)}  \Big[ \lb c^a_\mu e^{i\omega_\mu \tau} + c_\mu^{a\dagger} e^{-i\omega_\mu \tau} \rb \phi_\mu(i) , \\
& & \lb e^\frac{i\omega_\nu}{2} - e^\frac{-i\omega_\nu}{2} \rb \lb c^b_\nu e^{i\omega_\nu \tau} - c_\nu^{b\dagger} e^{-i\omega_\nu \tau} \rb \phi_\nu(j) \Big] \nn\\
&=& \sum_{\mu,\nu\in \lb 0,1 \rb} \lb e^\frac{i\omega_\nu}{2} - e^\frac{-i\omega_\nu}{2} \rb \lb - \lsb c^a_\mu, c^{b\dagger}_\nu \rsb + \lsb c^{a\dagger}_\mu, c^{b}_\nu \rsb \rb \phi_\mu(i) \phi_\nu(j) \nn\\
&=& - i\eta^{ab} \sum_{\mu\in \lb 0,1 \rb} \phi_\mu(i)\phi_{1-\mu}(j),\nn
\ea
and similarly
\ba
\lsb X_i^a,P_j^b \rsb_\mrm{crit} &=&  \sum_{\mu,\nu\in \{\frac{1}{2}+it|t\in\mathbb{R}\}}  \Big[ \lb c^a_\mu e^{i\omega_\mu \tau} + c_{1-\mu}^{a\dagger} e^{-i\omega_\mu \tau} \rb \phi_\mu(i) , \\
& & \lb e^\frac{i\omega_\nu}{2} - e^\frac{-i\omega_\nu}{2} \rb \lb c^b_\nu e^{i\omega_\nu \tau} - c_{1-\nu}^{b\dagger} e^{-i\omega_\nu \tau} \rb \phi_\nu(j) \Big] \nn\\
&=& - i\eta^{ab} \sum_{\mu\in \{\frac{1}{2}+it|t\in\mathbb{R}\}} \phi_\mu(i)\phi_{1-\mu}(j),\nn
\ea
and for the zero mode
\ba
\lsb X_i^a,P_j^b \rsb_\mrm{z.m.} &=&  \Big[ c^{a+}_0 \phi_0(i) + c^{a+}_1 \phi_1(i), \mathfrak{p}^b  + i\otoz \lb c^{b-}_0 \phi_0(j) +  c^{b-}_1 \phi_1(j) \rb \Big] \nn \\
&=& - i \eta^{ab} \phi_0(i)\phi_1(j) - i \eta^{ab}\phi_1(i)\phi_0(j) + i\eta^{ab}\phi_1(i)\phi_1(j) 
\ea

Now consider a function $F:V\lb T_p \rb\to \mathbb{R}$ on the tree. We can expand this function in Fourier modes as
\be
F(j) = \sum_\nu f_\nu \phi_\nu(j).
\ee
The position-momentum commutator $\lsb X_i^a,P^b \rsb$ acts on $F$ as
\ba
\lsb X_i^a,P^b \rsb \cdot F &=& \sum_{j\in V(T_p)} \lsb X_i^a,P_j^b \rsb F(j)\\
&=& \sum_{j\in V(T_p)} \lb  \lsb X_i^a,P_j^b \rsb_\mrm{z.m.} + \lsb X_i^a,P_j^b \rsb_\mrm{(0,1)} + \lsb X_i^a,P_j^b \rsb_\mrm{crit} \rb F(j), \nn
\ea
and for the three sectors we have
\ba
\sum_{j\in V(T_p)} \lsb X_i^a,P_j^b \rsb_\mrm{(0,1)} F(j) &=& - i \eta^{ab} \sum_{\mu\in(0,1)} \sum_\nu \sum_{j\in V(T_p)}  \phi_\mu(i)\phi_{1-\mu}(j) f_\nu \phi_\nu(j) \\
&=& i\eta^{ab} \sum_{\mu\in(0,1)} f_\mu \phi_\mu(i),\nn\\
\sum_{j\in V(T_p)} \lsb X_i^a,P_j^b \rsb_\mrm{crit} F(j) &=& - i \eta^{ab} \sum_{\mu\in \{\frac{1}{2}+it|t\in\mathbb{R}\} } \sum_\nu \sum_{j\in V(T_p)}  \phi_\mu(i)\phi_{1-\mu}(j) f_\nu \phi_\nu(j)\nn \\
&=& i\eta^{ab} \sum_{\mu\in \{\frac{1}{2}+it|t\in\mathbb{R}\} } f_\mu \phi_\mu(i),\nn
\ea
and finally
\ba
\sum_{j\in V(T_p)} \lsb X_i^a,P_j^b \rsb_\mrm{z.m.} F(j) &=& -i\eta^{ab}\sum_\nu \sum_{j\in V(T_p)} \lb \phi_0(i)\phi_1(j) + \phi_1(i)\phi_0(j) - \phi_1(i)\phi_1(j) \rb \times\nn\\
& &\times f_\nu \phi_\nu(j) \nn \\
&=& i \eta^{ab}\lsb  f_0 \phi_0(i) + \phi_1(i) \lb f_0+f_1 \rb - \phi_1(i) f_0 \rsb \nn\\
&=& i\eta^{ab}\lb f_0\phi_0(i) +  f_1\phi_1(i) \rb.
\ea
Thus we have obtained
\ba
\lsb X_i^a,P^b \rsb \cdot F &=& i\eta^{ab} \lsb f_0\phi_0(i) +  f_1\phi_1(i) + \lb \sum_{\mu\in (0,1) } + \sum_{\mu\in \{\frac{1}{2}+it|t\in\mathbb{R}\} } \rb f_\mu\phi_\mu(i)  \rsb\nn\\
&=& i\eta^{ab}\delta_{i,\cdot} \cdot F,
\ea
as advertised.

\section{Conserved quantities mode expansion}
\label{sec5}

In this section we will compute the total momentum $\PP^a$, mass squared $\PP^2$, Hamiltonian $H$ and angular momentum operators $J^{ab}$, in terms of the creation and annihilation operators.

\begin{remark}[The role of the zero mode cutoff $\otoz$ in the conserved quantities]
Because of the zero mode contributions, momentum $\PP^a$, Hamiltonian $H$ and angular momentum $J^{ab}$ will all have dependence on the zero mode cutoff $\otoz$. As might be expected, $\PP^a$, $H$ and $J^{ab}$ will all have finite limits when the cutoff $\otoz$ is taken to zero, and in fact these expressions will closely resemble the bosonic string ones. However, there is a further subtlety involving the zero mode cutoff $\otoz$ that needs to be taken into account: because some of the \eqref{zerozmcommutators} commutators diverge in $\otoz$, in checking the Poincar\'e algebra in Section \ref{sec8} one needs to take the subleading $\otoz$ terms in $\PP^a$, $H$ and $J^{ab}$ into account, otherwise cross-terms could end up contaminating the leading order result. For this reason, we will keep track of the first subleading (linear) pieces in $\otoz$ throughout this section.
\end{remark}

\begin{remark}
Because there is no worldsheet Noether theorem, it is not evident that $\PP^2$, $H$ and $J^{ab}$ are worldsheet time $\tau$ independent. This will indeed turn out to be the case, reliant on the Fourier modes.
\end{remark}

\subsection{The total momentum operator}

\begin{defn}
The total momentum operator is
\be
\mathcal{P}^a \coloneqq \sum_{i\in V\lb T_p \rb} P_i^a.
\ee
\end{defn}

\begin{lemma} 
The total momentum operator equals
\be
\label{eq72}
\mathcal{P}^a = - \frac{\mathfrak{p}^a}{p-1}.
\ee
\end{lemma}
\begin{proof}
We expand in Fourier modes, noting that only the $\{0,1\}$ sector contributes to the sum over the tree vertices,
\ba
\mathcal{P}^a &=& \sum_{i\in V\lb T_p \rb} \lsb \mathfrak{p}^a + \sum_{\mu\in[0,1]} \lb e^\frac{i\omega_\mu}{2} - e^{-\frac{i\omega_\mu}{2}} \rb\lb c_\mu^a e^{i\omega_\mu\tau} - c_\mu^{a\dagger} e^{-i\omega\tau} \rb \phi_\mu(i) \rsb \\
&=& - \frac{\mathfrak{p}^a}{p-1} - \sum_{\mu\in[0,1]} \lb e^\frac{i\omega_\mu}{2} - e^{-\frac{i\omega_\mu}{2}} \rb\lb c_\mu^a e^{i\omega_\mu\tau} - c_\mu^{a\dagger} e^{-i\omega\tau} \rb \lb \delta_\mu + \delta_{\mu-1} \rb \frac{1}{\sqrt{p-1}}.\nn
\ea
To leading order in $\otoz$ this is just 
\be
\label{eq4p4Pab}
\mathcal{P}^a = - \frac{\mathfrak{p}^a}{p-1} - \frac{i\otoz\lb c_0^{a-} + c_1^{a-} \rb}{\sqrt{p-1}}.
\ee
\end{proof}

\subsection{The Hamiltonian}
\label{secHamicons}

We now expand the Hamiltonian in the creation and annihilation operators. Our aim will be to prove the following result.

\begin{thm}
The Fourier mode expansion of the Hamiltonian is
\ba
\label{eq44Hami}
H &=& \frac{1}{1-p} \frac{\mathfrak{p}^2}{2}  -\sqrt{\frac{2}{v_p}} \sum_{\mu\in [0,1]-\{1/2\}} \sqrt{ 1+ p - p^\mu - p^{1-\mu} }  b^{a\dagger}_{\mu} b^a_{1-\mu}\\
& &-\sqrt{\frac{2}{v_p}} \sum_{\mu\in \{1/2+it,\, t\in\mathbb{R}\}} \sqrt{ 1+ p - p^\mu - p^{1-\mu} }  b^{a\dagger}_{\mu} b^a_{\mu}. \nn
\ea
\end{thm}

Note that Eq. \eqref{eq44Hami} is a semi-classical result, in that it does not include information on the operator ordering. We will discuss discuss operator ordering below in Section~\ref{sec7}. Furthermore, Eq. \eqref{eq44Hami} does not include the contribution of the zero mode which vanishes in the limit $\otoz\to 0$; for this contribution see Eqs. \eqref{eqqq4182}, \eqref{eqqq418} below.

\begin{proof}
From Eq.~\eqref{eqhami} with $\ell=1$ we have
\be
H = \sum_{i\in V\lb T_p \rb} \frac{\lb P_i^a \rb^2 }{2} + \sum_{\ipj \in E\lb T_p \rb} \frac{\lb X^a_i-X^a_j\rb^2}{v_p}.
\ee
We will expand the two terms above separately in the Fourier modes. Let
\be
\DD \coloneqq [0,1] \cup \{1/2+it|t\in\mathbb{R}\},
\ee
and denote the coefficients of $\phi_\mu(i)$ entering the $X^a_i$ and $P^a_i$ expansions by $C^a_\mu$ and $D_\mu^a$ respectively, that is
\ba
\label{eqs196}
C^a_\mu(\tau) \coloneqq \begin{cases}
 c^a_\mu e^{i\omega_\mu\tau} + c^{a\dagger}_\mu e^{-i\omega_\mu\tau} \qquad\ \ \mu\in [0,1] \\
 c^a_\mu e^{i\omega_\mu\tau} + c^{a\dagger}_{1-\mu} e^{-i\omega_\mu\tau} \qquad \mu\in \{\frac{1}{2} + it|t\in\mathbb{R}\}
\end{cases},
\ea
and
\ba
\label{eqs198}
D^a_\mu(\tau) \coloneqq \lb e^{\frac{i\omega_\mu}{2}} - e^{-\frac{i\omega_\mu}{2}} \rb \begin{cases}
 c^a_\mu e^{i\omega_\mu\tau} - c^{a\dagger}_\mu e^{-i\omega_\mu\tau} \qquad\ \ \mu\in [0,1] \\
 c^a_\mu e^{i\omega_\mu\tau} - c^{a\dagger}_{1-\mu} e^{-i\omega_\mu\tau} \qquad \mu\in \{\frac{1}{2} + it|t\in\mathbb{R}\}
\end{cases},
\ea
so that
\ba
X_i^a(\tau) &=& \mathfrak{p}^a\tau+ \sum_{\mu\in \DD} C^a_\mu(\tau) \phi_\mu(i), \\
P_i^a(\tau) &=& \mathfrak{p}^a + \sum_{\mu\in \DD} D^a_\mu(\tau) \phi_\mu(i). \nn
\ea
We now compute directly. Using $\sum_{\ipj \in E\lb T_p \rb} = 1/2\sum_{i\in V(T_p)}\sum_{j\sim i}$, and that there are $p$ neighbors of vertex $i$ where the eigenvalue changes by $p^{-\mu}$ and one where it changes by $p^\mu$, we have 
\ba
\sum_{\ipj \in E\lb T_p \rb} \lb X_i^a - X_j^a \rb^2 &=& \sum_{\ipj \in E\lb T_p \rb} \sum_{\mu,\nu\in \DD} C_\mu^a C_\nu^a \lb \phi_\mu(i) - \phi_\mu(j) \rb \lb \phi_\nu(i) - \phi_\nu(j) \rb \nn \\
&=& \frac{1}{2} \sum_{i\in V\lb T_p \rb} \sum_{\mu,\nu\in \DD} C^a_\mu C^a_\nu \Big[  p \lb 1 - p^{-\mu} \rb 
\lb 1 - p^{-\nu}\rb \nn\\
& &+ \lb 1 - p^\mu \rb\lb 1 - p^\nu \rb \Big] \phi_\mu(i) \phi_\nu(i).
\ea
Before proceeding further we need to understand the $\delta_{\mu+\nu}$ term in the Fourier mode orthogonality relations \eqref{eq221}. Because this term only matters for the zero mode, and in particular it only applies for $\mu=\nu=0$ (and the time $\tau$ evolution is not unitary for $\mu<0$), it should be understood as $\delta_{\mu,0}\delta_{\nu,0}$. Then we have
\ba
\sum_{\ipj \in E\lb T_p \rb} \lb X_i^a - X_j^a \rb^2 &=& - \frac{1}{2} \sum_{\mu,\nu\in \DD} C^a_\mu C^a_\nu \lsb  p \lb p^{-\mu} -1 \rb 
\lb p^{-\nu} -1 \rb + \lb p^\mu-1 \rb\lb p^\nu-1 \rb \rsb \times\nn\\
& &\times \lb \delta_{\mu}\delta_{\nu} +\delta_{\mu+\nu-1} \rb 
\ea
The square bracket contribution at $\mu=\nu=0$ is order $\OO(\mu\nu)$, which from Eq. \eqref{eqomegamurel349} translates to order $\OO\lb\otoz^4\rb$, so we will ignore it. We have thus obtained
\be
\sum_{\ipj \in E\lb T_p \rb} \lb X_i^a - X_j^a \rb^2 = - \sum_{\mu \in \DD} C^a_\mu C^a_{1-\mu} \left(1-p^{\mu}\right) \left(1-p^{1-\mu}\right),
\ee
where we should note that $\DD$ also includes a zero mode contribution at $\mu=0,1$.

Similarly, to second order in $\otoz$ we have
\ba
\sum_{i \in V\lb T_p \rb} \lb P_i^a \rb^2 &=& \lb \sum_{i\in V\lb T_p \rb} 1 \rb\mathfrak{p}^2 + 2 \mathfrak{p}^a \sum_{i \in V\lb T_p \rb} \sum_{\mu\in \DD} D_\mu^a \phi_\mu(i) \\
& & + \sum_{i \in V\lb T_p \rb} \sum_{\mu,\nu\in \DD} D_\mu^a D_\nu^a \phi_\mu(i) \phi_\nu(i)\nn \\
&=& \frac{\mathfrak{p}^2}{1-p}  - \frac{2}{\sqrt{p-1}} \mathfrak{p}^a \lb D_0^a + D_1^a \rb - \sum_{\mu,\nu\in \DD} D_\mu^a D_\nu^a \lb \delta_{\mu}\delta_{\nu} + \delta_{\mu+\nu-1} \rb \nn \\
&=& \frac{\mathfrak{p}^2}{1-p}  - \frac{2i\otoz}{\sqrt{p-1}} \mathfrak{p}^a \lb c_0^{a-} + c_1^{a-} \rb + \otoz^2 c_0^{a-}c_0^{a-} -  \sum_{\mu\in \DD} D_\mu^a D_{1-\mu}^a, \nn
\ea
where in the zero mode contributions we have kept all terms that may potentially contribute once the divergences in the zero mode commutators are taken into account. 

Then the Hamiltonian is
\be
\label{Hexp}
H = \frac{1}{1-p} \frac{\mathfrak{p}^2}{2} + H_\mrm{zr}  - \frac{1}{2} \sum_{\mu\in \DD} D_\mu^a D_{1-\mu}^a - \frac{1}{v_p}\sum_{\mu\in \DD} \left(1-p^{\mu}\right) \left(1-p^{1-\mu}\right)  C^a_\mu C^a_{1-\mu},
\ee
where (but note that the sum over $\DD$ also contains part of the zero sector contribution, when~$\mu=0,1$)
\be
\label{eqqq4182}
H_\mrm{zr} \coloneqq - \frac{i\omega_{\mu=0}}{\sqrt{p-1}} \mathfrak{p}^a \lb c_0^{a-} + c_1^{a-} \rb + \frac{1}{2}\omega_{\mu=0}^2 c_0^{a-}c_0^{a-}.
\ee

We now employ Eqs. \eqref{eqs196}, \eqref{eqs198}, remembering that $e^{i\omega_\mu}$ satisfies (from Eq. \eqref{eqomegamurel349})
\be
\label{eq616}
e^{i\omega_\mu} + e^{-i\omega_\mu} - 2 = \frac{2}{v_p} \lb p^\mu + p^{1-\mu} - p -1 \rb,
\ee
and we obtain
\ba
\label{eqqq418}
H&=& \frac{1}{1-p} \frac{\mathfrak{p}^2}{2} + H_\mrm{zr} + \frac{2}{v_p}\sum_{\mu\in [0,1]-\{1/2\}}\lb p^\mu + p^{1-\mu} - p - 1 \rb \lb c^a_\mu c^{a\dagger}_{1-\mu}+c^a_{1-\mu}c^{a\dagger}_\mu \rb\nn\\
& & +  \frac{2}{v_p}\sum_{\mu\in \{1/2+it,\, t\in\mathbb{R}\}}\lb p^\mu + p^{1-\mu} - p - 1 \rb \lb c^a_\mu c^{a\dagger}_{\mu}+c^a_{1-\mu}c^{a\dagger}_{1-\mu} \rb\nn \\
&=& \frac{1}{1-p} \frac{\mathfrak{p}^2}{2} + H_\mrm{zr} -\frac{1}{2} \sqrt{\frac{2}{v_p}} \sum_{\mu\in [0,1]-\{1/2\}} \sqrt{ 1+ p - p^\mu - p^{1-\mu} } \lb b^a_\mu b^{a\dagger}_{1-\mu}+b^a_{1-\mu}b^{a\dagger}_\mu \rb\nn\\
& &-\frac{1}{2} \sqrt{\frac{2}{v_p}} \sum_{\mu\in \{1/2+it,\, t\in\mathbb{R}\}} \sqrt{ 1+ p - p^\mu - p^{1-\mu} } \lb b^a_\mu b^{a\dagger}_{\mu}+b^a_{1-\mu}b^{a\dagger}_{1-\mu} \rb\nn\\
\label{eq1109}
&=& \frac{1}{1-p} \frac{\mathfrak{p}^2}{2} + H_\mrm{zr} -\sqrt{\frac{2}{v_p}} \sum_{\mu\in [0,1]-\{1/2\}} \sqrt{ 1+ p - p^\mu - p^{1-\mu} }  b^{a\dagger}_{\mu} b^a_{1-\mu}\nn\\
& &-\sqrt{\frac{2}{v_p}} \sum_{\mu\in \{1/2+it,\, t\in\mathbb{R}\}} \sqrt{ 1+ p - p^\mu - p^{1-\mu} }  b^{a\dagger}_{\mu} b^a_{\mu},
\ea
where we have used the change of variables $\mu\to1-\mu$ in some of the sums.
\end{proof}
Note that Eq. \eqref{eqqq418} is a semi-classical result, i.e. we have discared the constant from ordering $b_\mu$, $b_\nu^\dagger$ past each other, and thus it does not contain information on the operator ordering. We will address this in Section~\ref{sec7} below.

\begin{remark} Hamiltonian $H$ is worldsheet time $\tau$ independent. The derivation of this depended on the momentum operator eigenvalue squared equaling the Laplacian eigenvalue.
\end{remark}

\subsection{Angular momentum}

We now decompose the angular momentum operator in creation and annihilation operators.

\begin{defn}
Let the angular momentum operator be
\be
\label{eq435Jab}
J^{ab} \coloneqq \sum_{i\in V\lb T_p \rb}  \lb X_i^a P_i^b - X_i^b P_i^a \rb.
\ee
\end{defn}
Note that there is no ordering ambiguity in Eq. \eqref{eq435Jab}.

\begin{thm}
In terms of the Fourier modes, the angular momentum $J^{ab}$ equals
\ba
-J^{ab} &=& \frac{1}{\sqrt{p-1}}\lsb \lb c_0^{a+} + c_1^{a+} \rb \mathfrak{p}^b - \lb c_0^{b+} + c_1^{b+} \rb \mathfrak{p}^a \rsb\nn\\
& &+ 2\sum_{\mu\in (0,1)-\{1/2\}} \lb e^\frac{i\omega_\mu}{2} - e^{-\frac{i\omega_\mu}{2}} \rb \lb c_\mu^{a\dagger} c_{1-\mu}^b - c_{1-\mu}^{b\dagger} c_{\mu}^a \rb\nn\\
& & +2\sum_{\mu\in \{1/2+it,\, t\in\mathbb{R}\}} \lb e^\frac{i\omega_\mu}{2} - e^{-\frac{i\omega_\mu}{2}} \rb \lb c_\mu^{a\dagger} c_{\mu}^b - c_{\mu}^{b\dagger} c_{\mu}^a \rb,\nn
\ea 
and is worldsheet time $\tau$ independent.
\end{thm}

\begin{proof}
We proceed by direct computation, by expanding in the Fourier sectors. Noting that the cross-terms across sectors cancel, we have
\ba 
-J^{ab} &=& - \sum_{i\in V\lb T_p \rb} \lsb \lb \mathfrak{p}^a\tau + \sum_{\mu\in \DD} C^a_\mu \phi_\mu(i) \rb \lb \mathfrak{p}^b + \sum_{\nu\in \DD} D^b_\nu \phi_\nu(i) \rb - \lb a \leftrightarrow b \rb \rsb \nn\\
&\eqqcolon&-J_\mrm{zm} + \sum_{\mu\in (0,1)-\{1/2\}} \lb C_\mu^a D_{1-\mu}^b - C_\mu^b D_{1-\mu}^a \rb \\
& &+ \sum_{\mu\in \{1/2+it,t\in\mathbb{R}\}} \lb C_\mu^a D_{1-\mu}^b - C_\mu^b D_{1-\mu}^a\rb\nn\\
&\eqqcolon& -J_\mrm{zm}^{ab} - J_\mrm{comp}^{ab} - J_\mrm{crit}^{ab}, \nn
\ea
where (remembering Eqs. \eqref{eq342cmu01} for the definitions of $c_\mu^{a\pm}(\tau)$)
\ba
-J_\mrm{zm}^{ab} &=& - \sum_{i\in V\lb T_p \rb}  \lsb \mathfrak{p}^a\tau + c_0^{a+}(\tau) \phi_0(i) + c_1^{a+}(\tau) \phi_1(i) \rsb \big[ \mathfrak{p}^b + i\otoz c_0^{b-}(\tau)\phi_0(i)\nn\\
& & + i\otoz c_1^{b-}(\tau) \phi_1(i) \big] - \lb a \leftrightarrow b \rb \\
&=& \frac{1}{\sqrt{p-1}} \Big[ \mathfrak{p}^a\tau i\otoz \lb c_0^{b-} + c_1^{b-} \rb + c_0^{a+}(\tau) \Big( \mathfrak{p}^b + i\otoz \sqrt{p-1} c_0^{b-} \nn\\
& &+i\otoz \sqrt{p-1} c_1^{b-} \Big) +c_1^{a+}(\tau) \lb \mathfrak{p}^b + i\otoz \sqrt{p-1} c_0^{b-} \rb \Big] - \lb a \leftrightarrow b \rb \nn \\
&=& \frac{1}{\sqrt{p-1}}\lsb \lb c_0^{a+} + c^{a+}_1 \rb \mathfrak{p}^b - \lb c_0^{b+} +c_1^{b+} \rb \mathfrak{p}^a \rsb+ \frac{i\otoz}{\sqrt{p-1}} \big[ \mathfrak{p}^a\tau(c_0^{b-}+c_1^{b-})\nn\\
& &+ (c_0^{a-}+c_1^{a-})\mathfrak{p}^b\tau + \sqrt{p-1} \lb c_0^{a+} c_0^{b-} + c_0^{a+} c_1^{b-} + c_1^{a+} c_0^{b-}\rb - \lb a \leftrightarrow b \rb \big], \nn
\ea
and we have kept terms to first order in $\otoz$. 

We now compute the contributions of the complementary and critical sectors. We~have
\ba
-J_\mrm{comp}^{ab}&=& \sum_{\mu\in (0,1)-\{1/2\}} \lb e^\frac{i\omega_\mu}{2} - e^{-\frac{i\omega_\mu}{2}} \rb \Big[ \lb c_\mu^a e^{i\omega_\mu \tau } + c_\mu^{a\dagger} e^{-i\omega_\mu \tau } \rb \lb c_{1-\mu}^b  e^{i\omega_\mu \tau } - c_{1-\mu}^{b\dagger} e^{-i\omega_\mu \tau } \rb \nn\\
& & - \lb c_\mu^b e^{i\omega_\mu \tau } + c_\mu^{b\dagger} e^{-i\omega_\mu \tau } \rb \lb c_{1-\mu}^a e^{i\omega_\mu \tau } - c_{1-\mu}^{a\dagger} e^{-i\omega_\mu \tau } \rb \Big] \nn\\
&=& \sum_{\mu\in (0,1)-\{1/2\}} \lb e^\frac{i\omega_\mu}{2} - e^{-\frac{i\omega_\mu}{2}} \rb \lb -c_\mu^a c_{1-\mu}^{b\dagger} + c_\mu^{a\dagger} c_{1-\mu}^b + c_\mu^b c_{1-\mu}^{a\dagger} - c_\mu^{b\dagger} c_{1-\mu}^a \rb \nn\\
&=& 2\sum_{\mu\in (0,1)-\{1/2\}} \lb e^\frac{i\omega_\mu}{2} - e^{-\frac{i\omega_\mu}{2}} \rb \lb c_\mu^{a\dagger} c_{1-\mu}^b - c_{1-\mu}^{b\dagger} c_{\mu}^a \rb,
\ea
where in the last line we have used the commutation relations and performed the variable change $\mu\to1-\mu$.

Remembering that the critical line sector is the same as the unit interval sector, except for the index change $c_\mu^{a\dagger}\to c_{1-\mu}^{a\dagger}$, we have
\be
-J_\mrm{crit}^{ab} =  2\sum_{\mu\in \{1/2+it,\, t\in\mathbb{R}\}} \lb e^\frac{i\omega_\mu}{2} - e^{-\frac{i\omega_\mu}{2}} \rb \lb c_\mu^{a\dagger} c_{\mu}^b - c_{\mu}^{b\dagger} c_{\mu}^a \rb.
\ee
Putting everything together, we have obtained
\ba
\label{eq441Jab}
-J^{ab}&=& -J_\mrm{zm}^{ab} - J_\mrm{comp}^{ab} - J_\mrm{crit}^{ab}\\
&=& \frac{1}{\sqrt{p-1}}\lsb \lb c_0^{a+} + c_1^{a+} \rb \mathfrak{p}^b - \lb c_0^{b+} + c_1^{b+} \rb \mathfrak{p}^a \rsb\nn\\
& &+ \frac{i\otoz}{\sqrt{p-1}} \big[ \mathfrak{p}^a\tau(c_0^{b-}+c_1^{b-}) + (c_0^{a-}+c_1^{a-})\mathfrak{p}^b\tau + \sqrt{p-1} ( c_0^{a+} c_0^{b-} + c_0^{a+} c_1^{b-} \nn\\
& & + c_1^{a+} c_0^{b-}) - \lb a \leftrightarrow b \rb \big] + 2\sum_{\mu\in (0,1)-\{1/2\}} \lb e^\frac{i\omega_\mu}{2} - e^{-\frac{i\omega_\mu}{2}} \rb \lb c_\mu^{a\dagger} c_{1-\mu}^b - c_{1-\mu}^{b\dagger} c_{\mu}^a \rb\nn\\
& & +2\sum_{\mu\in \{1/2+it,\, t\in\mathbb{R}\}} \lb e^\frac{i\omega_\mu}{2} - e^{-\frac{i\omega_\mu}{2}} \rb \lb c_\mu^{a\dagger} c_{\mu}^b - c_{\mu}^{b\dagger} c_{\mu}^a \rb.\nn
\ea
\end{proof}

Eq. \eqref{eq441Jab} is the angular momentum in terms of the Fourier modes, to first order in $\otoz$. Note that the time dependent piece at order $O(\otoz)$ would cancel if we strengthened the $[\mathfrak{p}^a,c_0^{b-}+c_1^{b-}]$ commutator to $[\mathfrak{p}^a,c_0^{b-}+c_1^{b-}]=0$, however we will not need to make this choice in the rest of the paper.

\section{Poincar\'e algebra}
\label{sec8}

In this section we verify the Poincar\'e algebra by direct computation.

As explained above, because some of the creation-annihilation commutators diverge in the zero mode cutoff $\otoz$, when computing the zero mode contributions to the Poincar\'e algebra we must keep the subleading terms $\otoz$ in $\PP^a$ and $J^{ab}$, and send $\otoz\to 0$ at the very end.

Because the worldsheet is discrete, it is not given that our target space theory has Poincar\'e invariance. The Poincar\'e invariance of the target is nontrivial and must be checked carefully, and in fact it will not hold unless the $\delta_{\mu+\nu}$ term is present in the orthogonality relations (see Remark \ref{Remark26Poin} below).

\begin{remark}
Once the zero mode cutoff $\otoz$ is removed (that is in the limit $\otoz\to0$), the 
\be
[\PP^a,\PP^b] =0 
\ee
part of the Poincar\'e algebra follows trivially from Eq. \eqref{eq4p4Pab} and the the second line of Eq. \eqref{zerozmcommutators}.
\end{remark}

\begin{lemma}
The Poincar\'e algebra commutator 
\be
\lsb J^{ab}, \PP^c \rsb =  i \lb \eta^{ac} \PP^b - \eta^{bc} \PP^a \rb
\ee
is obeyed, in the $\otoz\to 0$ limit.
\end{lemma}

\begin{proof}
We proceed by direct computation. We have (from Eqs. \eqref{eq4p4Pab} and \eqref{eq441Jab})
\ba
\lsb J_\mrm{zm}^{ab}, \PP^c \rsb &=& \Bigg[ \frac{1}{\sqrt{p-1}}\lsb \lb c_0^{a+} + c^{a+}_1 \rb \mathfrak{p}^b - \lb c_0^{b+} +c_1^{b+} \rb \mathfrak{p}^a \rsb + \frac{i\otoz}{\sqrt{p-1}} \times\\
& &\times \big[ \mathfrak{p}^a\tau(c_0^{b-}+ c_1^{b-}) + (c_0^{a-}+c_1^{a-})\mathfrak{p}^b\tau + \sqrt{p-1} ( c_0^{a+} c_0^{b-} + c_0^{a+} c_1^{b-} \nn\\
& & + c_1^{a+} c_0^{b-} ) - \lb a \leftrightarrow b \rb \big] , \frac{\mathfrak{p}^c}{p-1} + \frac{i\otoz\lb c_0^{c-} + c_1^{c-} \rb}{\sqrt{p-1}} \Bigg] \nn\\
&=& \frac{1}{\sqrt{p-1}} \lsb c_0^{a+} \mathfrak{p}^b - c_0^{b+}  \mathfrak{p}^a, \frac{\mathfrak{p}^c}{p-1}  \rsb \nn \\
& &+ \frac{i\otoz}{p-1} \lsb \lb c_0^{a+} + c^{a+}_1 \rb \mathfrak{p}^b - \lb c_0^{b+} +c_1^{b+} \rb \mathfrak{p}^a,  c_0^{c-} + c_1^{c-}  \rsb, \nn
\ea
where in the second equality we have kept only the terms that could potentially contribute to the leading piece. Now note that
\ba
\label{thisiseq54}
\lsb c_0^{a+} + c_1^{a+}, c_0^{c-}+c_1^{c-} \rsb = \frac{\eta^{ac}}{\otoz} \lb \alpha -1  \rb
\ea
from expanding into the four commutators and Eq. \eqref{zerozmcommutators}. Then we have
\be
\label{JabPccommutans}
\lsb J_\mrm{zm}^{ab}, \PP^c \rsb = i \lb \eta^{ac} \PP^b - \eta^{bc} \PP^a \rb,
\ee
where parameter $\alpha$ has dropped out.

\end{proof}

\begin{remark}
\label{Remark26Poin}
The Poincar\'e algebra does not depend only on the structure of the flat target space, but also on the structure of the Bruhat-Tits tree, and so it needs to be checked carefully, as it is not automatic that it will hold. In particular, commutator \eqref{JabPccommutans} above delicately depends on the presence of the $\delta_{\mu+\nu}$ contribution in the tree Fourier mode orthogonality conditions \eqref{eq223}. To see this, consider for instance a version of the orthogonality relations where the $\delta_{\mu+\nu}$ term has been replaced by $q\delta_{\mu+\nu}$, for a parameter $q$. Then the last line of Eq. \eqref{eq336split} changes to $-q\eta^{ab}$ on the left-hand side, which propagates into the $[c^{a+}_1,c^{b-}_1]$ commutator as (no other commutators are affected)
\be
[c^{a+}_1,c^{b-}_1] = \frac{(2-q)\eta^{ab}}{\otoz}.
\ee
With this modification Eq. \eqref{thisiseq54} changes into
\be
\lsb c_0^{a+} + c_1^{a+}, c_0^{c-}+c_1^{c-} \rsb = \frac{\eta^{ac}}{\otoz} \lb \alpha - q \rb,
\ee
so that the Poincar\'e algebra commutator \eqref{JabPccommutans} modifies as 
\be
\lsb J_\mrm{zm}^{ab}, \PP^c \rsb = i q \lb \eta^{ac} \PP^b - \eta^{bc} \PP^a \rb.
\ee
Thus the Poincar\'e algebra holds only for $q=1$. It can also be checked that $q\neq 1$ breaks the zero mode contribution in the $[J^{ab},J^{cd}]$ Poincar\'e algebra commutator (see Remark \ref{appremark} in Appendix \ref{zerosectorappendix} below).
\end{remark}

In the remainder of this section we will prove the remaining Poincar\'e algebra commutator, which will be more involved.

\begin{lemma}
\label{lemmaintermstep}
For the open unit interval we have
\be
\lsb c_\mu^{a\dagger} c_{1-\mu}^b , c_\nu^{c\dagger} c_{1-\nu}^d  \rsb = \frac{1}{2} \frac{\eta^{bc}c_\mu^{a\dagger}c_{1-\nu}^d-\eta^{ad}c_\nu^{c\dagger}c_{1-\mu}^b}{\sqrt{\frac{2}{v_p}\lb 1+p-p^\mu - p^{1-\mu} \rb}}\delta_{\mu-\nu}.
\ee
\end{lemma}
\begin{proof}
By direct computation,
\ba
\lsb c_\mu^{a\dagger} c_{1-\mu}^b , c_\nu^{c\dagger} c_{1-\nu}^d  \rsb &=& c_\mu^{a\dagger} \lsb c_{1-\mu}^b , c_\nu^{c\dagger} c_{1-\nu}^d  \rsb + \lsb c_\mu^{a\dagger} , c_\nu^{c\dagger} c_{1-\nu}^d  \rsb c_{1-\mu}^b \\
&=& c_\mu^{a\dagger} \lsb c_{1-\mu}^b , c_\nu^{c\dagger} \rsb c_{1-\nu}^d + c_\nu^{c\dagger} \lsb c_\mu^{a\dagger} , c_{1-\nu}^d  \rsb c_{1-\mu}^b \nn \nn\\
&=& c_\mu^{a\dagger} \frac{1}{2} \frac{\eta^{bc}}{\sqrt{\frac{2}{v_p}\lb 1+p-p^\mu - p^{1-\mu} \rb}}\delta_{\mu-\nu} c_{1-\nu}^d \nn \nn\\
& & - c_\nu^{c\dagger} \frac{1}{2} \frac{\eta^{ad}}{\sqrt{\frac{2}{v_p}\lb 1+p-p^\mu - p^{1-\mu} \rb}}\delta_{\mu-\nu} c_{1-\mu}^b. \nn 
\ea
\end{proof}

\begin{lemma}
\label{lemmaintermstep2}
For the critical line sector we have
\be
\lsb c_\mu^{a\dagger} c_{\mu}^b , c_\nu^{c\dagger} c_{\nu}^d  \rsb = \frac{1}{2} \frac{\eta^{bc}c_\mu^{a\dagger}c_{\nu}^d-\eta^{ad}c_\nu^{c\dagger}c_{\mu}^b}{\sqrt{\frac{2}{v_p}\lb 1+p-p^\mu - p^{1-\mu} \rb}}\delta_{\mu-\nu}.
\ee
\end{lemma}
\begin{proof}
By direct computation,
\ba
\lsb c_\mu^{a\dagger} c_{\mu}^b , c_\nu^{c\dagger} c_{\nu}^d  \rsb &=& c_\mu^{a\dagger} \lsb c_{\mu}^b , c_\nu^{c\dagger} c_{\nu}^d  \rsb + \lsb c_\mu^{a\dagger} , c_\nu^{c\dagger} c_{\nu}^d  \rsb c_{\mu}^b \\
&=& c_\mu^{a\dagger} \lsb c_{\mu}^b , c_\nu^{c\dagger} \rsb c_{\nu}^d + c_\nu^{c\dagger} \lsb c_\mu^{a\dagger} , c_{\nu}^d  \rsb c_{\mu}^b \nn\\
&=& c_\mu^{a\dagger} \frac{1}{2} \frac{\eta^{bc}}{\sqrt{\frac{2}{v_p}\lb 1+p-p^\mu - p^{1-\mu} \rb}}\delta_{\mu-\nu} c_{\nu}^d \nn\\
& & - c_\nu^{c\dagger} \frac{1}{2} \frac{\eta^{ad}}{\sqrt{\frac{2}{v_p}\lb 1+p-p^\mu - p^{1-\mu} \rb}}\delta_{\mu-\nu} c_{\mu}^b. \nn
\ea
\end{proof}

\begin{lemma}
\label{lemma6}
The non-zero mode part of the angular momentum obeys the Poincar\'e algebra commutators.
\end{lemma}
\begin{proof}
This is just direct computation. Remembering that the complementary and critical sectors commute with one another and using Eq. \eqref{eq441Jab} for the unit interval contributions, we have
\ba
\label{eqJ2811}
& &\lsb J_\mrm{comp}^{ab}, J_\mrm{comp}^{cd} \rsb = 4\sum_{\mu,\nu\in (0,1)-\{1/2\}} \lb e^\frac{i\omega_\mu}{2} - e^{-\frac{i\omega_\mu}{2}} \rb \lb e^\frac{i\omega_\nu}{2} - e^{-\frac{i\omega_\nu}{2}} \rb \times\\
& &\times \lsb c_\mu^{a\dagger} c_{1-\mu}^b - c_{1-\mu}^{b\dagger} c_{\mu}^a , c_\nu^{c\dagger} c_{1-\nu}^d - c_{1-\nu}^{d\dagger} c_{\nu}^c \rsb\nn\\
&=& 2 \sum_{\mu\in (0,1)-\{1/2\}} \frac{\lb e^\frac{i\omega_\mu}{2} - e^{-\frac{i\omega_\mu}{2}} \rb^2}{\sqrt{\frac{2}{v_p}\lb 1+p-p^\mu - p^{1-\mu} \rb}}\Big(\eta^{bc} c_\mu^{a\dagger} c^d_{1-\mu} - \eta^{ad} c^{c\dagger}_{\mu} c_{1-\mu}^b -\eta^{bd} c_\mu^{a\dagger} c^c_{1-\mu}\nn\\ 
& & + \eta^{ac} c^{d\dagger}_{\mu} c_{1-\mu}^b -\eta^{ac} c_\mu^{b\dagger} c^d_{1-\mu} + \eta^{bd} c^{c\dagger}_{\mu} c_{1-\mu}^a+\eta^{ad} c_\mu^{b\dagger} c^c_{1-\mu} - \eta^{bc} c^{d\dagger}_{\mu} c_{1-\mu}^a \Big),\nn
\ea
where we used Lemma \ref{lemmaintermstep} in the second equality. We now need to be careful with square root branches. Eq. \eqref{eq616} is
\be
\lb e^\frac{i\omega_\mu}{2} - e^{-i\frac{\omega_\mu}{2}} \rb^2 = \frac{2}{v_p} \lb p^\mu + p^{1-\mu} - p -1 \rb,
\ee
so that remembering that in our conventions $\omega_\mu\geq 0$, we have
\be
e^\frac{i\omega_\mu}{2} - e^{-i\frac{\omega_\mu}{2}} = 2i\sin \frac{\omega_\mu}{2} = i \sqrt{\frac{2}{v_p}\lb 1+p-p^\mu - p^{1-\mu} \rb}.
\ee
Then Eq. \eqref{eqJ2811} assembles into
\ba
& &\lsb J_\mrm{comp}^{ab}, J_\mrm{comp}^{cd} \rsb = 2i \sum_{\mu\in(0,1)-\{1/2\}}\lb e^\frac{i\omega_\mu}{2} - e^{-\frac{i\omega_\mu}{2}} \rb\Big[\eta^{bc} \lb c_\mu^{a\dagger} c^d_{1-\mu} - c^{d\dagger}_{\mu} c_{1-\mu}^a \rb \\ 
& &+ \eta^{ad} \lb c_\mu^{b\dagger} c^c_{1-\mu} - c^{c\dagger}_{\mu} c_{1-\mu}^b \rb +\eta^{bd} \lb c^{c\dagger}_{\mu} c_{1-\mu}^a - c_\mu^{a\dagger} c^c_{1-\mu} \rb + \eta^{ac} \lb c^{d\dagger}_{\mu} c_{1-\mu}^b -c_\mu^{b\dagger} c^d_{1-\mu} \rb \Big],\nn
\ea
so that
\be
\lsb J_\mrm{comp}^{ab}, J_\mrm{comp}^{cd} \rsb = i\lb  \eta^{ac} J_\mrm{comp}^{bd}  - \eta^{ad} J_\mrm{comp}^{bc} - \eta^{bc} J_\mrm{comp}^{ad} + \eta^{bd} J_\mrm{comp}^{ac}  \rb, 
\ee
as desired. Next we need to compute the critical line sector commutators,
\ba
\lsb J_\mrm{crit}^{ab}, J_\mrm{crit}^{cd} \rsb &=& 4\sum_{\mu,\nu\in\{1/2+it,\,t\in\mathbb{R}\}} \lb e^\frac{i\omega_\mu}{2} - e^{-\frac{i\omega_\mu}{2}} \rb \lb e^\frac{i\omega_\nu}{2} - e^{-\frac{i\omega_\nu}{2}} \rb \lsb c_\mu^{a\dagger} c_{\mu}^b - c_{\mu}^{b\dagger} c_{\mu}^a , c_\nu^{c\dagger} c_{\nu}^d - c_{\nu}^{d\dagger} c_{\nu}^c \rsb\nn\\
&=& 2 \sum_{\mu\in \{1/2+it,\,t\in\mathbb{R}\}} \frac{\lb e^\frac{i\omega_\mu}{2} - e^{-\frac{i\omega_\mu}{2}} \rb^2}{\sqrt{\frac{2}{v_p}\lb 1+p-p^\mu - p^{1-\mu} \rb}}\Big(\eta^{bc} c_\mu^{a\dagger} c^d_{\mu} - \eta^{ad} c^{c\dagger}_{\mu} c_{\mu}^b \\ 
& &-\eta^{bd} c_\mu^{a\dagger} c^c_{\mu} + \eta^{ac} c^{d\dagger}_{\mu} c_{\mu}^b -\eta^{ac} c_\mu^{b\dagger} c^d_{\mu} + \eta^{bd} c^{c\dagger}_{\mu} c_{\mu}^a\ +\eta^{ad} c_\mu^{b\dagger} c^c_{\mu} - \eta^{bc} c^{d\dagger}_{\mu} c_{\mu}^a \Big)\nn\\
&=& 2i \sum_{\mu\in\{1/2+it,\,t\in\mathbb{R}\}}\lb e^\frac{i\omega_\mu}{2} - e^{-\frac{i\omega_\mu}{2}} \rb\Big[\eta^{bc} \lb c_\mu^{a\dagger} c^d_{\mu} - c^{d\dagger}_{\mu} c_{\mu}^a \rb + \eta^{ad} \lb c_\mu^{b\dagger} c^c_{\mu} - c^{c\dagger}_{\mu} c_{\mu}^b \rb \nn \\
& &+\eta^{bd} \lb c^{c\dagger}_{\mu} c_{\mu}^a - c_\mu^{a\dagger} c^c_{\mu} \rb + \eta^{ac} \lb c^{d\dagger}_{\mu} c_{\mu}^b -c_\mu^{b\dagger} c^d_{\mu} \rb \Big]
\ea
where we used Lemma \ref{lemmaintermstep2} in the second equality. This is just
\be
\lsb J_\mrm{crit}^{ab}, J_\mrm{crit}^{cd} \rsb = i\lb  \eta^{ac} J_\mrm{crit}^{bd}  - \eta^{ad} J_\mrm{crit}^{bc} - \eta^{bc} J_\mrm{crit}^{ad} + \eta^{bd} J_\mrm{crit}^{ac} \rb.
\ee
\end{proof}

Finally, we consider the zero-mode contribution to the $[J^{ab}, J^{cd}]$ Poincar\'e algebra commutator.

\begin{lemma}
\label{lemma7}
The zero mode contribution to the angular momentum obeys the Poincar\'e algebra, for any value of $\alpha$.
\end{lemma}

The proof of this lemma is long (if straightforward), and is given in Appendix~\ref{zerosectorappendix}.

We have thus obtained the following theorem.

\begin{thm}
The angular momentum
\be
J^{ab} = J^{ab}_\mrm{zm} + J^{ab}_\mrm{comp} +  J^{ab}_\mrm{crit}
\ee
obeys the Poincar\'e algebra
\be
\lsb J^{ab}, J^{cd} \rsb = i \lb \eta^{ac} J^{bd} - \eta^{ad} J^{bc} - \eta^{bc} J^{ad} + \eta^{bd} J^{ac} \rb.
\ee 
\end{thm}
\begin{proof}
Immediate from Lemmas \ref{lemma6}, \ref{lemma7} and the fact that the different sector contributions commute with each other.
\end{proof}

\section{The Archimedean spectrum}
\label{sec7}

So far the treatment in Sections \ref{sec2} -- \ref{sec8} above has been semi-classical, and at one place~$p$ only. We now put all places together, into an adelic construction. In order to do this, we will impose two additional constraints, that have not been required until~now. We will not determine the entire spectrum on the Archimedean side, rather we will focus on the first excited level, by making use of the light-cone gauge. With a suitable condition on parameter $v_p$, we will show that the spectrum can be made to include a photon and graviton, similarly to the usual bosonic string. However, we should emphasize that this Archimedean theory that we recover from the $p$-adic side is not bosonic string theory around the usual unstable vacuum, although we will propose in Section~\ref{secconc} how it may be related to it.

A considerably simpler example of an Euler product of spectra has been constructed in \cite{Huang:2020vwx}, for the case of the quantum mechanical particle-in-a-box, treated at the Archimedean place and $p$-adically. In that case the Euler product can be proven to hold, by computing both sides separately.

\begin{remark}
In order to obtain an adelic spectrum interpretation we will require the following conditions:
\begin{itemize}
\item $b^{a=+}_\mu=0$ in the lightcone gauge.
\item {\rm(The Hamiltonian constraint.)} $H-\mathfrak{a}=0$, for some constant $\mathfrak{a}$. Note that the constant $\mathfrak{a}$ takes care of possible ordering ambiguities in $H$.
\end{itemize}
\end{remark}

\textbf{The lightcone gauge condition:} The first condition is imposing the lightcone gauge on the target space. Note that the lightcone $a=\pm$ coordinates defined as
\be
x^\pm = x^0\pm x^1
\ee
on the target space are not related to the zero sector $\pm$ notation defined as $b_{\mu=0,1}^{a\pm}=b^{a}_{\mu=0,1}\pm b^{a\dagger}_{\mu=0,1}$ for all values of spacetime index $a$.

In the case of the lightcone quantization of the usual bosonic string, after imposing the lightcone gauge one needs to demand that the Virasoro constraints are still satisfied. This is a nontrivial condition, eventually leading to the result that the critical dimension is $D=26$. In our case, because the worldsheet is discrete, we do not have an analogue of the Virasoro constraints.\footnote{The study of various notions of gravity on the worldsheet, in the spirit of \cite{Gubser:2016htz}, and of the constraint equations arising from them, is an important endeavour, that we will however not attempt in the present paper. For some recent developments in the case of graph gravity based on the Lin-Lu-Yau curvature see e.g. \cite{Huang:2020qik}.} Nonetheless, we will impose a certain Hamiltonian constraint, coming from a geometric condition, explained below. This condition will lead not to a determination of the spacetime dimension, but rather to the result that for the first excited level of the $p2$-brane to be massless, the parameter $\mu$ should be localized at the critical zeta zeros.

\textbf{The Hamiltonian constraint:} The second condition is a $p$-adic analogue of the mass-shell condition $L_0-\mathfrak{a}=0$ for the physical states in the usual bosonic string theory, with $\mathfrak{a}$ the normal ordering constant. It follows from imposing a geometric symmetry on the action, as we will explain in Section \ref{sec61} below. A difference between the $p2$-brane and bosonic string theory is that while in the bosonic string theory case having the first excited level be a massless vector forces $\mathfrak{a}=1$, for the $p2$-brane the same requirement will enforce $\mathfrak{a}=0$.

\subsection{A geometric symmetry of the action}
\label{sec61}

Our action from Eq. \eqref{Seq}, with the edge lengths written explicitly, is
\be
\label{eq62S}
S = \sum_{T_1} \lb \frac{1}{2} \sum_{i\in V\lb T_p \rb} \frac{\lb \pd_{\tau'} X_i^a \rb^2}{4\ell^2_\tau}  - \sum_{\ipj \in E\lb T_p \rb} \frac{\lb X^a_i-X^a_j\rb^2}{v_p \ell^2}\rb,
\ee
where $\pd_\tau'$ is now the finite difference between the values of $X_i^a$ at different time steps, not divided by $4\ell_\tau$ (this factor has been pulled in front of the derivative in Eq. \eqref{eq62S}). The variation of this action w.r.t. the edge lengths is
\be
-\frac{1}{2}\delta S = \sum_{T_1} \lb \frac{1}{2} \sum_{i\in V\lb T_p \rb} \frac{\lb \pd_{\tau'} X_i^a \rb^2}{4\ell^3_\tau} \delta \ell_\tau  - \sum_{\ipj \in E\lb T_p \rb} \frac{\lb X^a_i-X^a_j\rb^2}{v_p \ell^3} \delta \ell \rb.
\ee
Then for a variation such that 
\be
\label{vardeltaell}
\delta \ell_\tau = - \frac{1}{2}\delta \ell
\ee
we will have that $\delta S = 0$ iff $H=0$, around the edge length configuration $2\ell_\tau = \ell$. Thus the Hamiltonian constraint is equivalent to the action being stationary around $\ell_\tau = \ell$ for the variation \eqref{vardeltaell}; note that $\ell$, $\ell_\tau$ in Eq. \eqref{vardeltaell} are still functions of time $\tau$, so that we have different amounts of variation at different time steps.

Note that $H=0$ is the semi-classical constraint. Once ordering ambiguities are taken into account, the constraint becomes 
\be
\label{Hconst62}
H - \mathfrak{a} = 0,
\ee
just as in the usual bosonic string theory.

\subsection{The first excited level and spectrum Euler product}

The Hamiltonian constraint \eqref{Hconst62} in the lightcone gauge is
\ba
\label{Hconst66}
\frac{1}{1-p} \frac{\mathfrak{p}^2}{2} &=& \sqrt{\frac{2}{v_p}} \sum_{\mu\in [0,1]-\{1/2\}} \sqrt{ 1+ p - p^\mu - p^{1-\mu} }  b^{i\dagger}_{\mu} b^i_{1-\mu}\\
&+&\sqrt{\frac{2}{v_p}} \sum_{\mu\in \{1/2+it,\, t\in\mathbb{R}\}} \sqrt{ 1+ p - p^\mu - p^{1-\mu} }  b^{i\dagger}_{\mu} b^i_{\mu} +\mathfrak{a}. \nn
\ea
where we have used Eq. \eqref{eq44Hami} and dropped the $H_\mrm{zr}$ contribution since in this section we will only be concerned with the critical and complementary sectors. Using that $m^2 = \PP^2$, from Eq. \eqref{eq4p4Pab} we have (where again we drop the zero mode terms)
\be
\label{Eq67}
m^2 = \frac{\mathfrak{p}^2}{\lb p-1\rb^2}.
\ee

By analogy with the bosonic string (or with any quantum harmonic oscillator), we introduce the ground state $| 0\rangle$ such that
\be
b_\mu^{i} |0\rangle = 0, \qquad \mu \in (0,1) \cup \lcb 1/2 + i t | t\in \mathbb{R} \rcb,
\ee
and $b^{i\dagger}_\mu |0\rangle $ represents the creation of a nontrivial excitation.

\begin{remark}
From the commutation relations \eqref{eq176} and Eq. \eqref{eq349pairing} we take the conjugate of $b_\mu^{i\dagger}|0\rangle$ to be $\langle 0|b_{\mu}^i$ on the critical line, but $\langle 0|b_{1-\mu}^i$ on the $(0,1)$ interval.
\end{remark}

\begin{lemma}
\label{lemma9hereitis}
The states
\be
| \psi^j_\nu\rangle \coloneqq b^{j\dagger}_\nu |0\rangle, \qquad \nu \in (0,1) \cup \lcb 1/2 + i t | t\in \mathbb{R} \rcb
\ee
have mass squared
\be
m_{(p)}^2 = \frac{2\mathfrak{a}^2}{1-p} + \frac{2}{1-p} \sqrt{\frac{2}{v_p}} \sqrt{1+p-p^\nu-p^{1-\nu}}.
\ee
\end{lemma}

\begin{proof}
This is immediate computation. Remembering that
\be
\lsb b_\mu^{i}, b_\nu^{j\dagger} \rsb = \eta^{ij}\Delta_{\mu,\nu},
\ee
for $\nu$ on the critical line we have
\be
\sum_{\mu\in \{1/2+it| t\in\mathbb{R}\}}  \sqrt{1+p-p^\mu-p^{1-\mu}} b_\mu^{i\dagger}b_\mu^{i} b_\nu^{j\dagger} | 0 \rangle =  \sqrt{1+p-p^\nu-p^{1-\nu}} b_\nu^{j\dagger} | 0 \rangle
\ee
so that from Eq. \eqref{Hconst66} we obtain
\be
\label{eqcrit01613}
\langle \psi_\nu^j | \frac{\mathfrak{p}^2}{2} | \psi^j_\nu\rangle = \mathfrak{a} + \sqrt{\frac{2}{v_p}} \sqrt{1+p-p^\nu-p^{1-\nu}}.
\ee
A parallel computation shows that Eq. \eqref{eqcrit01613} also holds on the $(0,1)$ interval. Then from Eq. \eqref{Eq67} the mass squared of state $|\psi_\nu^j\rangle$ at place $p$ is
\be
\label{mp614}
m_{(p)}^2 = \frac{2\mathfrak{a}^2}{1-p} + \frac{2}{1-p} \sqrt{\frac{2}{v_p}} \sqrt{1+p-p^\nu-p^{1-\nu}},
\ee
as advertised.
\end{proof}

\begin{remark}
Up to now parameter $v_p$ has been any real number, with the property $v_p\geq\lb \sqrt p + 1 \rb^2/2$. In the remainder of the paper we will pick the potential term parameter to be $v_p=8p$, noting that indeed $8p \geq \lb \sqrt p + 1 \rb^2/2$ for $p\geq 1/9$. The parameter $v_p$ controls the overall prefactor in the Euler product, and $v_p=8p$ corresponds to a certain value of this prefactor, which will be explained below.
\end{remark}

\begin{remark}
The value $v_p=8p$ may seem arbitrary, but is in fact natural and can be understood in the following way. First note that in order for the potential term in Eq. \eqref{Seq} to not dominate the kinetic term at large $p$, $v_p$ should be order $\OO(p)$, because for a fixed vertex $i$ the $j$ sum over its neighbors is of order $\OO(p)$. Thus, a natural choice would be $v_p=2p$, for two reasons: i) in the $p\to 1$ limit corresponding to an Archimedean theory which will not be considered in this paper the kinetic and potential term prefactors both become $1/2$, and ii) for $v_p=2p$ the inequality $v_p\geq \lb \sqrt p + 1 \rb^2/2$ holds precisely for $p\geq 1$. Now, the choice $v_p=8p$ is related to $v_p=2p$ by a factor of $2$ outside the square root in Eq. \eqref{mp614}. This factor of $2$ corresponds to parameter $\mu$ double-counting the $(0,1)$ interval and critical line modes in Eq. \eqref{eqqq418} and above, by running over the entire $(0,1)$ interval and critical line ``cross,'' as opposed to the ``inverted L'' $(0,1/2)\cup \{1/2+it|t\geq 0\}$. Restricting $\mu$ to not double count (as in \cite{Huang:2019nog}) would fix $v_p=2p$ in line with the discussion above, however we will choose conventions so that we keep the double counting and use $v=8p$ instead. Note that irrespective of the double-counting, all points above and below $1/2$ on the critical line (and to the left and right of $1/2$ on the $(0,1)$ interval) get included in the modes, both as creation and as annihilation operators.
\end{remark}

For the masses $m_{(p)}$ defined by Eq. \eqref{mp614} at the finite places, we define an Archimedean mass squared by Euler product~as 
\be
m_{(\infty)}^2 \coloneqq \prod_p m_{(p)}^{-2}.
\ee
The following lemma holds.

\begin{lemma} 
Pick $v_p=8p$. For $\mathfrak{a}=0$ and $\nu$ at a critical line zeta zero we have~$m_{(\infty)}=~0$.
\end{lemma}
\begin{proof}
If $\mathfrak{a}=0$ and $v_p=8p$ then from Eq. \eqref{mp614} we have
\ba
\label{Eulerproduct616}
\prod_p m_{(p)}^2 &=& \prod_p \frac{1}{1-p} \sqrt{\left(1-p^{\nu -1}\right) \left(1-p^{-\nu }\right)} \\
&=& \frac{\zeta(-1)}{\sqrt{\zeta(\nu)\zeta(1-\nu)}},\nn
\ea
so that
\ba
m_{(\infty)}^{2} = \frac{\sqrt{\zeta(\nu)\zeta(1-\nu)}}{\zeta(-1)}.
\ea
Then $m_{(\infty)}=0$ if $\nu$ is a critical zeta zero.
\end{proof}

The Euler product \eqref{Eulerproduct616} holds in the sense of analytic continuation.  Note that in the usual bosonic string theory the zeta value $\zeta(-1)=-1/12$ gives rise to the critical dimension $D=26$.

Condition $m_{(\infty)}=0$ is a consistency condition for the first excited level to be in the irreducible vector representation of $SO(D-2)$, on the Archimedean side. With this condition the $(0,1)$ sector cannot contribute to the massless Archimedean spectrum, because there are no zeta zeros on the $(0,1)$ interval.

With the choices above, the $p2$-brane spectrum contains the following:

\begin{itemize}
\item The ground state $|0\rangle$. With $\mathfrak{a}=0$ in Eq. \eqref{Hconst66} the mass of the ground state will be
$m_{(p)}^2 = 0$, i.e. the $p$-adic ground state is massless. We will comment on the adelic interpretation of this in Section \ref{secconc} below.
\item States $b_\nu^{i\dagger}|0\rangle$ at place $p$, with $\nu$ a critical zeta zero, corresponding to a \emph{photon} transforming as a vector irreducible representation of $SO(D-2)$ at the Archimedean place.
\end{itemize}
In addition to the above, one can also consider the following states:
\begin{itemize}
\item States $b_\nu^{i\dagger}b_{1-\nu}^{j\dagger}|0\rangle$ at place $p$, with $\nu$ a critical zeta zero. In the usual bosonic string such states are in the closed sector. Remembering that in the usual bosonic string case the closed sector mass squared operator has an extra factor of $2$ relative to the open sector mass squared operator, we normalize the $v_p$ coefficients as $v^\mrm{closed}_p=4v^\mrm{open}_p$ (because the $p$-adic mass squared is proportional to $\sqrt{1/v_p}$). Then a computation analogous to Lemma \ref{lemma9hereitis} shows that the $p$-adic mass of these excitations is
\be
\label{padicmass618}
\frac{1}{1-p} \sqrt{\left(1-p^{\nu -1}\right) \left(1-p^{-\nu }\right)}.
\ee 
These excitations will thus have Archimedean mass proportional to $\zeta(\nu)\zeta(1- \nu)$, so they will again be massless at the Archimedean place. Thus, these modes will decompose into a graviton, antisymmetric tensor and dilaton, similar to the $N=2$ level in the usual closed bosonic string theory.
\end{itemize}

\begin{remark} 
The derivation of the photon and graviton from the critical zeta zeros presented here is perfectly analogous to the derivation of the photon and graviton from the critical $D=26$ dimension condition in the usual bosonic string theory. That is, in the usual bosonic string theory the mass of the photon is proportional to $26-D$, so that \emph{choosing} the photon to be a massless vector mode is equivalent to setting $D=26$, otherwise the theory is a non-critical string theory. Similarly, in the case of the p2-brane \emph{choosing} the photon to be a massless vector mode is equivalent to fixing parameter~$\nu$ at the critical zeta zeros.
\end{remark}

\section{A conjecture: Berkovich flow and tachyon decay}
\label{secconc}

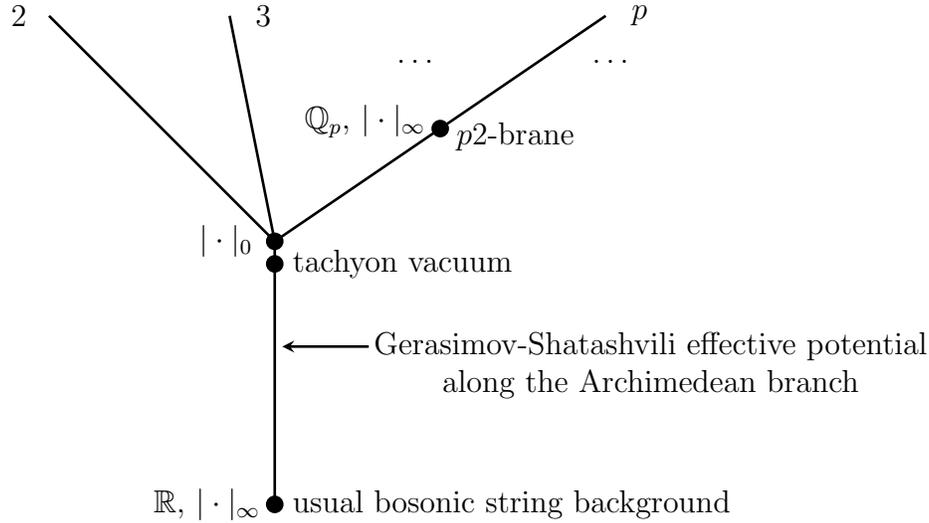
\begin{figure}[h]
\centering
\begin{tikzpicture}
\node [style=simple] (u0) at (0,0) {};
\node [style=simple] (u1)  at (0,-3.5) {};
\node [style=simple] (u1epsilon)  at (0,-0.3) {};
\node [style=simple] (u8)  at (2.2,1.5) {};
\node [style=none] (u3) at (-3,3) {};
\node [style=none] (u4) at (4.4,3) {};
\node [style=none] (u5) at (-0.6,3) {};
\node [style=none] (l2) at (1.2,1.6) {$\mathbb{Q}_p,\, |\cdot|_\infty$};
\node [style=none] (l2prime) at (3.2,1.4) {$p2$-brane};
\node [style=none] (l3) at (-0.15,3) {$3$};
\node [style=none] (l4) at (-3.4,3) {$2$};
\node [style=none] (l5) at (4.85,3) {$p$};
\node [style=none] (l6) at (1.9,2.4) {$\dots$};
\node [style=none] (l6d) at (4.5,2.4) {$\dots$};
\node [style=none] (l7) at (-0.65,0) {$|\cdot|_0$};
\node [style=none] (l9) at (-0.9,-3.5) {$\mathbb{R},\, |\cdot|_\infty$};
\node [style=none] (l9prime) at (3.15,-3.5) {usual bosonic string background};
\node [style=none] (l9epsilon) at (1.7,-0.3) {tachyon vacuum};
\node [style=none] (lGS) at (5.0,-1.4) {Gerasimov-Shatashvili effective potential};
\node [style=none] (lGS) at (5.0,-1.9) {along the Archimedean branch};

\draw [style=simple,stealth-](0.1,-1.4) -- (1.25,-1.4);

\draw [style=simple] (u0.center) to (u1.center);
\draw [style=simple] (u0.center) to (u3.center);
\draw [style=simple] (u0.center) to (u4.center);
\draw [style=simple] (u0.center) to (u5.center);
\end{tikzpicture}
\caption{Proposed relation between the $p2$-brane and the usual unstable tachyon vacuum in bosonic string theory, via Berkovich RG flow. The spectrum Euler product is a local condition around $|\cdot|_0$.}
\label{figberk}
\end{figure}

We will now conclude with a conjectural proposal on how the $p2$-brane fits in the broader landscape of string theory.

Let's first comment on the Archimedean mass of the $p$-adic ground state $|0\rangle$. As explained above, for normal ordering constant $\mathfrak{a}=0$ the $p$-adic ground state is massless. Then the Euler product argument suggests that the Archimedean mass is infinite, and the Archimedean state corresponding to $|0\rangle$ is not a ground state at all.

This observation motivates the following proposal, building on the ideas in \cite{Huang:2020vwx}. Paper \cite{Huang:2020vwx} showed that an Euler product relation is obeyed between the spectra of the $p$-adic and Archimedean particle-in-a-box theories. Furthemore, this Euler product was interpreted as a local condition around the Gauss point $|\cdot|_0$ in the Berkovich space $\MM(\mathbb{Z})$ (for review on Berkovich spaces see e.g. \cite{PoineauTurchettiBerkSchot,baker,jonsson}). This Berkovich space $\MM(\mathbb{Z})$ is a graph with an infinite number of branches, corresponding to places $(p)$ and the Archimedean place $(\infty)$. The point in \cite{Huang:2020vwx} was that along the branches the flow exhibits trivial scaling (because the theory is free), and the Euler product constrains the spectrum around the point $|\cdot|_0$.

We can apply the same ideas here, in the case of the $p2$-brane/bosonic string. Our conjecture will be the following (see Figure \ref{figberk}). The $p2$-brane theory lives at the $|\cdot|_p$ points on the $p$-adic branches, and the usual unstable vacuum of bosonic string theory is at the $|\cdot|_\infty$ point on the Archimedean branch. Along the $p$-adic branches the theory scales trivially, because $p$-adic theories are rigid (a similar point of view was presented in \cite{Ghoshal:2006zh}). The spectrum Euler product of Section \ref{sec7} is a local condition around the Gauss point $|\cdot|_0$, so that the Archimedean spectrum obtained from the product corresponds to a theory right under the Gauss point $|\cdot|_0$ on the Archimedean branch, however in this case the flow along the Archimedean branch will be nontrivial. We furthermore conjecture that this flow corresponds to tachyon decay. Gerasimov-Shatashvili \cite{gershata} (see also \cite{Witten:1992qy,Ghoshal:2000dd,Ghoshal:2000gt}) considered the effective potential for tachyon decay in open-string field theory. They found that at the end of the potential roll the tachyon acquires infinite mass, corresponding to the decay of the space-filling brane. We propose to identify this infinite mass with the infinite Archimedean mass of the $p$-adic ground state $|0\rangle$. Thus, the background where the space-filling brane has decayed would in our picture correspond to the point on the Archimedean branch right under (i.e. infinitesimally away from) the Gauss point $|\cdot|_0$. The Euler product provides information about the theory at this point by flowing down from the $p$-adic branches and going around the $|\cdot|_0$ point, whereas Gerasimov-Shatashvili \cite{gershata} have integrated up along the Archimedean branch. To recover the usual bosonic string unstable tachyon vacuum from our $p2$-brane construction, one would then have to further flow down the Archimedean branch to~$|\cdot|_\infty$.

This interpretation provides a conjecture for the open string tachyon decay. It is known that since the space-filling brane decays, the usual fields arising from the string excitations can no longer be supported on it. However, we have found that such fields can still be supported from the $p$-adic 2-branes, at the point right under $|\cdot|_0$ on the Archimedean branch. Thus, one could conjecture that the space-filling brane decays into $p2$-branes (or rather their scaled analogues on the $p$-adic branches, infinitesimally away from the Gauss point $|\cdot|_0$). 

We should emphasize that in order to try to make the above precise, one will need a larger Berkovich space than the space $\MM\lb \mathbb{Z}\rb$ considered in \cite{Huang:2020vwx}, such as the Berkovich affine line.

The interpretation above also explains why our approach in this paper has functioned. A priori, it is not obvious that for a $p$-adic theory one should have the same Poisson brackets and commutators as in the Archimedean case, and use those to determine the Fourier mode expansion and the operators in the theory. However, if the $p$-adic theories are indeed connected to the Archimedean theories in this manner, it is plausible that having the correct Poisson brackets and commutators at the $\mathbb{Q}_p$ points is a necessary condition in order for them to hold at the Archimean theories also. Similarly, obtaining the Poincar\'e algebra correctly (including the subtle zero mode contributions) at the $\mathbb{Q}_p$ points is likely necessary in order for it to hold at the point on the Archimedean branch right under $|\cdot|_0$, and then further down at the point $|\cdot|_\infty$. 

\section*{Acknowledgments}

We would like to thank John Joseph Carrasco, Jonathan Heckman, Matilde Marcolli, Maria Nastasescu, and Sarthak Parikh for useful discussions. B.S. is indebted to Djordje Radicevic, Matt Headrick, and the other participants at the virtual Brandeis Meetings, where an in-progress version of this work was presented. The work of A.H. was supported by the Simons Collaboration for Mathematicians. The work of B.S. was supported by the Department of Energy under Award Number DE-SC0021485. The work of X.Z. was supported in part by a Summer Undergraduate Research Fellowship at Northwestern University.

\appendix

\section{Zero sector contributions to the Poincar\'e algebra angular momentum commutators}
\label{zerosectorappendix}

We now prove Lemma \ref{lemma7}. It needs to be approached in several steps.

Remember we have (from Eq. \eqref{eq441Jab})
\ba
-J_\mrm{zm}^{ab} &=& \frac{1}{\sqrt{p-1}}\lsb \lb c_0^{a+} + c_1^{a+} \rb \mathfrak{p}^b - \lb c_0^{b+} + c_1^{b+} \rb \mathfrak{p}^a \rsb + \frac{i\otoz}{\sqrt{p-1}} \Big[ \mathfrak{p}^a\tau(c_0^{b-}+c_1^{b-})\nn\\
&+& (c_0^{a-}+c_1^{a-})\mathfrak{p}^b\tau + \sqrt{p-1} (c_0^{a+} c_0^{b-} +  c_0^{a+} c_1^{b-} + c_1^{a+} c_0^{b-}) - \lb a \leftrightarrow b \rb \Big].
\ea

There will be several types of contributions to the $[J^{ab}, J^{cd}]$ commutator. Throughout the computation, we will only keep terms in $\otoz$ that can contribute to the leading order piece.

\textbf{The leading piece contribution:} 

\begin{remark}
We have
\ba
\label{eq832}
& &\lsb \lb c_0^{a+} + c_1^{a+} \rb \mathfrak{p}^b , \lb c_0^{c+} + c_1^{c+} \rb \mathfrak{p}^d \rsb = \\
&=&\lb c_0^{a+} + c_1^{a+} \rb \lsb \mathfrak{p}^b , \lb c_0^{c+} + c_1^{c+} \rb \mathfrak{p}^d \rsb + \lsb c_0^{a+} + c_1^{a+} , \lb c_0^{c+} + c_1^{c+} \rb \mathfrak{p}^d \rsb \mathfrak{p}^b\nn\\
&=& \lb c_0^{a+} + c_1^{a+} \rb \lsb \mathfrak{p}^b , c_0^{c+}\rsb \mathfrak{p}^d + \lb c_0^{c+} + c_1^{c+} \rb \lsb  c_0^{a+} , \mathfrak{p}^d \rsb \mathfrak{p}^b \nn\\
&=& i \alpha \sqrt{p-1} \lsb  \eta^{bc} \lb c_0^{a+} + c_1^{a+} \rb \mathfrak{p}^d - \eta^{ad}\lb c_0^{c+} + c_1^{c+} \rb \mathfrak{p}^b \rsb. \nn
\ea
\end{remark}

With this result we have
\ba
\lsb J_\mrm{zm}^{ab}, J_\mrm{zm}^{cd} \rsb_\mrm{ldn} &\coloneqq& \frac{1}{p-1} \lsb \lb c_0^{a+} + c_1^{a+} \rb \mathfrak{p}^b - \lb c_0^{b+} + c_1^{b+} \rb \mathfrak{p}^a, \lb c_0^{c+} + c_1^{c+} \rb \mathfrak{p}^d - \lb c_0^{d+} + c_1^{d+} \rb \mathfrak{p}^c \rsb \nn \\
&=& \frac{i\alpha}{\sqrt{p-1}}\Big[ - \eta^{ad}\lb c_0^{c+} + c_1^{c+} \rb \mathfrak{p}^b  + \eta^{bc} \lb c_0^{a+} + c_1^{a+} \rb \mathfrak{p}^d \\
& & + \eta^{ac}\lb c_0^{d+} + c_1^{d+} \rb \mathfrak{p}^b  - \eta^{bd} \lb c_0^{a+} + c_1^{a+} \rb \mathfrak{p}^c\nn \\
& & +\eta^{bd}\lb c_0^{c+} + c_1^{c+} \rb \mathfrak{p}^a  - \eta^{ac} \lb c_0^{b+} + c_1^{b+} \rb \mathfrak{p}^d\nn \\
& & - \eta^{bc}\lb c_0^{d+} + c_1^{d+} \rb \mathfrak{p}^a  + \eta^{ad} \lb c_0^{b+} + c_1^{b+} \rb \mathfrak{p}^c \Big]\nn\\
&=& - \frac{i\alpha}{\sqrt{p-1}} \Big\{ \eta^{ac} \lsb \lb c_0^{b+} + c_1^{b+} \rb \mathfrak{p}^d - \lb c_0^{d+} + c_1^{d+} \rb \mathfrak{p}^b \rsb \nn\\
& &+  \eta^{ad} \lsb \lb c_0^{c+} + c_1^{c+} \rb \mathfrak{p}^b - \lb c_0^{b+} + c_1^{b+} \rb \mathfrak{p}^c\rsb\nn\\
& &+  \eta^{bc} \lsb \lb c_0^{d+} + c_1^{d+} \rb \mathfrak{p}^a - \lb c_0^{a+} + c_1^{a+} \rb \mathfrak{p}^d\rsb\nn\\ 
& &+  \eta^{bd} \lsb \lb c_0^{a+} + c_1^{a+} \rb \mathfrak{p}^c - \lb c_0^{+c} + c_1^{+c} \rb \mathfrak{p}^a\rsb\Big\}, \nn
\ea
which is just
\be
\label{appJ1}
\lsb J_\mrm{zm}^{ab}, J_\mrm{zm}^{cd} \rsb_\mrm{ldn} = i\alpha \lb \eta^{ac} J_\mrm{zm}^{bd} - \eta^{ad} J_\mrm{zm}^{bc} - \eta^{bc} J_\mrm{zm}^{ad} + \eta^{bd} J_\mrm{zm}^{ac} \rb.
\ee

\textbf{The time-dependent piece:}

As might be expected, this piece will cancel in the Poincar\'e algebra commutator.

\begin{remark}
Keeping track only of terms that will contribute at order $1/\otoz$ and using Eq. \eqref{thisiseq54}, we have
\ba
& &\lsb \lb c_0^{a+} + c_1^{a+} \rb \mathfrak{p}^b, \mathfrak{p}^c(c_0^{d-}+c_1^{d-}) \rsb =  \lb c_0^{a+} + c_1^{a+} \rb \lsb \mathfrak{p}^b, \mathfrak{p}^c(c_0^{d-}+c_1^{d-}) \rsb \\
& &+ \lsb c_0^{a+} + c_1^{a+} , \mathfrak{p}^c(c_0^{d-}+c_1^{d-}) \rsb \mathfrak{p}^b = \lb \alpha -1  \rb \frac{\eta^{ad}}{\otoz} \mathfrak{p}^c \mathfrak{p}^b. \nn
\ea
\end{remark}

With this result the contribution to the $\lsb J_\mrm{zm}^{ab}, J_\mrm{zm}^{cd} \rsb$ commutator can be calculated. There will be two contributions, one of which is
\ba
\label{appJ2}
\lsb J_\mrm{zm}^{ab}, J_\mrm{zm}^{cd} \rsb_\mrm{td,1} &\coloneqq& \frac{i\otoz \tau}{p-1} \big[ \lb c_0^{a+} + c_1^{a+} \rb \mathfrak{p}^b - \lb c_0^{b+} + c_1^{b+} \rb \mathfrak{p}^a , \mathfrak{p}^c(c_0^{d-}+c_1^{d-})\\
& & + (c_0^{c-}+c_1^{c-})\mathfrak{p}^d - \lb c \leftrightarrow d \rb \big] \nn \\
&=& \frac{i(\alpha-1) \tau}{p-1} \lsb \eta^{ad}\mathfrak{p}^b\mathfrak{p}^c + \eta^{ac}\mathfrak{p}^b \mathfrak{p}^d - \eta^{bd}\mathfrak{p}^a\mathfrak{p}^c - \eta^{bc}\mathfrak{p}^a \mathfrak{p}^d - \lb c \leftrightarrow d \rb \rsb \nn \\
&=&0, \nn
\ea
so it vanishes. The second contribution is obtained by swapping indices $\{ab\}$ and $\{cd\}$ and introducing a minus sign, so it vanishes as well.

\textbf{The cross-term contribution:}

\begin{remark}
We have the following commutators (only to leading order in $1/\otoz$):
\ba
\label{eqA7}
& &\lsb \lb c_0^{a+} + c_1^{a+} \rb \mathfrak{p}^b , c_0^{c+} c_0^{d-} \rsb = c_0^{c+} \lsb c_0^{a+} + c_1^{a+}, c_0^{d-} \rsb \mathfrak{p}^b = \frac{\alpha-1}{\otoz} \eta^{ad} c_0^{c+} \mathfrak{p}^b\\
& &\lsb \lb c_0^{a+} + c_1^{a+} \rb \mathfrak{p}^b , c_0^{c+} c_1^{d-} \rsb = c_0^{c+} \lsb c_0^{a+} + c_1^{a+}, c_1^{d-} \rsb \mathfrak{p}^b = 0\nn\\
& &\lsb \lb c_0^{a+} + c_1^{a+} \rb \mathfrak{p}^b , c_1^{c+} c_0^{d-} \rsb = c_1^{c+} \lsb c_0^{a+} + c_1^{a+}, c_0^{d-} \rsb \mathfrak{p}^b = \frac{\alpha-1}{\otoz} \eta^{ad} c_1^{c+} \mathfrak{p}^b\nn
\ea
\end{remark}

Then there will be two contributions to the Poincar\'e algebra commutator, one of which is
\ba
\label{ctm1}
[J_\mrm{zm}^{ab},J_\mrm{zm}^{cd}]_\mrm{ctm,1} &\coloneqq& \frac{i\otoz}{\sqrt{p-1}}\big[ \lb c_0^{a+} + c_1^{a+} \rb \mathfrak{p}^b - \lb c_0^{b+} + c_1^{b+} \rb \mathfrak{p}^a, c_0^{c+} c_0^{d-} + c_0^{c+} c_1^{d-} \nn\\
& & + c_1^{c+} c_0^{d-} - \lb c \leftrightarrow d \rb \big] \\
&=& \frac{i(\alpha-1)}{\sqrt{p-1}} \big[ \eta^{ad} \lb c_0^{c+} +c_1^{c+} \rb \mathfrak{p}^b - \eta^{bd} \lb c_0^{c+} +c_1^{c+} \rb \mathfrak{p}^a \nn\\
& & - \eta^{ac} \lb c_0^{d+} +c_1^{d+} \rb \mathfrak{p}^b + \eta^{bc} \lb c_0^{d+} +c_1^{d+} \rb \mathfrak{p}^a  \big]\nn 
\ea
The second contribution can be obtained from the first by swapping indices $\{ab\}$ and $\{cd\}$ and introducing a minus sign:
\ba
\label{ctm2}
[J_\mrm{zm}^{ab},J_\mrm{zm}^{cd}]_\mrm{ctm,2} &=& \frac{i(\alpha-1)}{\sqrt{p-1}} \big[ - \eta^{bc} \lb c_0^{a+} +c_1^{a+} \rb \mathfrak{p}^d + \eta^{bd} \lb c_0^{a+} +c_1^{a+} \rb \mathfrak{p}^c \\
& & + \eta^{ac} \lb c_0^{b+} +c_1^{b+} \rb \mathfrak{p}^d - \eta^{ad} \lb c_0^{b+} +c_1^{b+} \rb \mathfrak{p}^c  \big] \nn
\ea
Putting Eqs. \eqref{ctm1}, \eqref{ctm2} together, we have
\ba
\label{appJ3}
[J_\mrm{zm}^{ab},J_\mrm{zm}^{cd}]_\mrm{ctm} &=& -i(\alpha-1) \lb \eta^{ac} J_\mrm{zm}^{bd} - \eta^{ad} J_\mrm{zm}^{bc} - \eta^{bc} J_\mrm{zm}^{ad} + \eta^{bd} J_\mrm{zm}^{ac} \rb.
\ea

\textbf{The second order contribution:}

Direct computation shows that the second order contribution 
\ba
\label{appJ4}
[J_\mrm{zm}^{ab}, J_\mrm{zm}^{cd} ]_\mrm{2nd} &=&  -\otoz^2 \Big[ (c_0^{a+} c_0^{b-} +  c_0^{a+} c_1^{b-} + c_1^{a+} c_0^{b-}) - \lb a \leftrightarrow b \rb, \\
& &(c_0^{c+} c_0^{d-} +  c_0^{c+} c_1^{d-} + c_1^{c+} c_0^{d-}) - \lb c \leftrightarrow d \rb \Big] \nn
\ea
cannot contaminate the leading-order piece.

\textbf{The result:}

Putting the diferent contributions \eqref{appJ1}, \eqref{appJ2}, \eqref{appJ3}, \eqref{appJ4} together, only Eqs. \eqref{appJ1}, \eqref{appJ3} contribute, the piece proportional to $\alpha$ cancels and we have obtained
\ba
\lsb J_\mrm{zm}^{ab}, J_\mrm{zm}^{cd} \rsb &=& \lsb J_\mrm{zm}^{ab}, J_\mrm{zm}^{cd} \rsb_\mrm{ldn} + [J_\mrm{zm}^{ab},J_\mrm{zm}^{cd}]_\mrm{ctm} \\
&=& i \lb \eta^{ac} J_\mrm{zm}^{bd} - \eta^{ad} J_\mrm{zm}^{bc} - \eta^{bc} J_\mrm{zm}^{ad} + \eta^{bd} J_\mrm{zm}^{ac} \rb, \nn
\ea
as advertised.

\begin{remark}
\label{appremark}
If $q\neq 1$, that is we alter the $\delta_{\mu+\nu}$ term in orthonormality conditions \eqref{eq221}, then the second commutator in Eq. \eqref{eqA7} no longer vanishes, which contaminates Eq. \eqref{ctm1} and the $[J_\mrm{zm}^{ab},J_\mrm{zm}^{cd}]$ commutator no longer closes.
\end{remark}

\section{Critical line algebras at place $p$}

As explained above, since the worldsheet is discrete one should not expect to recover the Virasoro algebra in the case of the $p2$-brane, from a single place $p$. Nonetheless, the $p2$-brane creation and annihilation operators can be combined to obtain an algebra, in a fashion analogous to how the Virasoro operators are obtained from the bosonic string creation and annihilation operators. We will explain this below.

Note that we will ignore ordering issues, so this construction can be better thought of as a $p2$-brane analogue of the Witt algebra, rather than the Virasoro algebra.

Because for our construction the set of values $t$ at which the operators are defined must be closed under addition, our algebras will be most naturally defined on the critical line, rather than on $[0,1]$ or the set of critical zeros.

\begin{defn}
Let $S$ be a discrete subset of the critical line $\{1/2+it|t\in\mathbb{R}\}$ that is closed under addition of the imaginary parts. We define
\ba
L_{t} &\coloneqq& \sum_{t'\in S} b_{\frac{1}{2}+i(t+t')}^{a\dagger} b_{\frac{1}{2}-it'}^a,\\
G_{t} &\coloneqq& \sum_{t'\in S} b_{\frac{1}{2}+i(t+t')}^{a\dagger} b_{\frac{1}{2}+it'}^a,\nn
\ea
and similarly
\ba
\LL_{t} \coloneqq \sum_{t'\in S} \sqrt{\frac{2}{v_p}\lb 1 + p - p^{\frac{1}{2}-i(t+t')} - p^{\frac{1}{2}+i(t+t')} \rb} c_{\frac{1}{2}+i(t+t')}^{a\dagger} c^a_{\frac{1}{2}-it'} \\
\GG_t \coloneqq \sum_{t'\in S} \sqrt{\frac{2}{v_p}\lb 1 + p - p^{\frac{1}{2}-i(t+t')} - p^{\frac{1}{2}+i(t+t')} \rb} c_{\frac{1}{2}+i(t+t')}^{a\dagger} c^a_{\frac{1}{2}+it'} \nn
\ea
\end{defn}

\begin{lemma}
We have the commutators
\ba
\lsb L_{t_1},L_{t_2}\rsb &=& G_{t_1-t_2}-G_{t_2-t_1},\\
\lsb G_{t_1},G_{t_2}\rsb &=&0, \nn \\
\lsb L_{t_1},G_{t_2}\rsb &=& L_{t_1-t_2}-L_{t_1+t_2},\nn
\ea
and
\ba
[ \LL_{t_1}, \LL_{t_2} ] &=& \frac{1}{2}\lb \GG_{t_1-t_2}-\GG_{t_2-t_1} \rb,\\
\lsb \GG_{t_1}, \GG_{t_2} \rsb &=& 0,\nn\\
\lsb \LL_{t_1}, \GG_{t_2} \rsb &=& \frac{1}{2} \lb \LL_{t_1-t_2}-\LL_{t_1+t_2} \rb. \nn
\ea
\end{lemma}
\begin{proof}
We show this by direct computation. All sums will be over $S$.
\ba
\lsb L_{t_1},L_{t_2}\rsb &=& \sum_{t',t''} \lsb b_{\frac{1}{2}+i(t_1+t')}^{a\dagger} b^a_{\frac{1}{2}-it'} , b_{\frac{1}{2}+i(t_2+t'')}^{b\dagger} b^b_{\frac{1}{2}-it''} \rsb\\
&=&\sum_{t',t''} \lb b_{\frac{1}{2}+i(t_1+t')}^{a\dagger} \lsb b^a_{\frac{1}{2}-it'} , b_{\frac{1}{2}+i(t_2+t'')}^{b\dagger} b^b_{\frac{1}{2}-it''} \rsb + \lsb b_{\frac{1}{2}+i(t_1+t')}^{a\dagger}  , b_{\frac{1}{2}+i(t_2+t'')}^{b\dagger} b_{\frac{1}{2}-it''}^b \rsb b_{\frac{1}{2}-it'}^a \rb \nn\\
&=&\sum_{t',t''} \lb b_{\frac{1}{2}+i(t_1+t')}^{a\dagger} \lsb b^a_{\frac{1}{2}-it'} , b_{\frac{1}{2}+i(t_2+t'')}^{b\dagger} \rsb b^b_{\frac{1}{2}-it''}  + b_{\frac{1}{2}+i(t_2+t'')}^{b\dagger} \lsb b_{\frac{1}{2}+i(t_1+t')}^{a\dagger}  , b^b_{\frac{1}{2}-it''} \rsb b^a_{\frac{1}{2}-it'} \rb \nn \\
&=&\sum_{t',t''} \lb b_{\frac{1}{2}+i(t_1+t')}^{a\dagger} \eta^{ab}\delta_{t_2+t''+t'} b^b_{\frac{1}{2}-it''}  - b_{\frac{1}{2}+i(t_2+t'')}^{b\dagger} \eta^{ab} \delta_{t_1+t'+t''} b^a_{\frac{1}{2}-it'} \rb \nn\\
&=& \sum_{t''} b_{\frac{1}{2}+i(t_1-t_2-t'')}^{a\dagger} b^a_{\frac{1}{2}-it''}  - \sum_{t'} b_{\frac{1}{2}+i(-t_1+t_2-t')}^{a\dagger} b^a_{\frac{1}{2}-it'} \nn\\
&=& G_{t_1-t_2}-G_{t_2-t_1}. \nn
\ea

\ba
\lsb G_{t_1},G_{t_2}\rsb &=& \sum_{t',t''} \lsb b_{\frac{1}{2}+i(t_1+t')}^{a\dagger} b^a_{\frac{1}{2}+it'} , b_{\frac{1}{2}+i(t_2+t'')}^{b\dagger} b^b_{\frac{1}{2}+it''} \rsb\\
&=&\sum_{t',t''} \lb b_{\frac{1}{2}+i(t_1+t')}^{a\dagger} \lsb b^a_{\frac{1}{2}+it'} , b_{\frac{1}{2}+i(t_2+t'')}^{b\dagger} b^b_{\frac{1}{2}+it''} \rsb + \lsb b_{\frac{1}{2}+i(t_1+t')}^{a\dagger}  , b_{\frac{1}{2}+i(t_2+t'')}^{b\dagger} b_{\frac{1}{2}+it''}^b \rsb b_{\frac{1}{2}+it'}^a \rb \nn\\
&=&\sum_{t',t''} \lb b_{\frac{1}{2}+i(t_1+t')}^{a\dagger} \lsb b^a_{\frac{1}{2}+it'} , b_{\frac{1}{2}+i(t_2+t'')}^{b\dagger} \rsb b^b_{\frac{1}{2}+it''}  + b_{\frac{1}{2}+i(t_2+t'')}^{b\dagger} \lsb b_{\frac{1}{2}+i(t_1+t')}^{a\dagger}  , b^b_{\frac{1}{2}+it''} \rsb b^a_{\frac{1}{2}+it'} \rb \nn \\
&=&\sum_{t',t''} \lb b_{\frac{1}{2}+i(t_1+t')}^{a\dagger} \eta^{ab}\delta_{t_2+t''-t'} b^b_{\frac{1}{2}+it''}  - b_{\frac{1}{2}+i(t_2+t'')}^{b\dagger} \eta^{ab} \delta_{t_1+t'-t''} b^a_{\frac{1}{2}+it'} \rb \nn\\
&=& 0. \nn
\ea

\ba
\lsb L_{t_1},G_{t_2}\rsb &=& \sum_{t',t''} \lsb b_{\frac{1}{2}+i(t_1+t')}^{a\dagger} b^a_{\frac{1}{2}-it'} , b_{\frac{1}{2}+i(t_2+t'')}^{b\dagger} b^b_{\frac{1}{2}+it''} \rsb\\
&=&\sum_{t',t''} \lb b_{\frac{1}{2}+i(t_1+t')}^{a\dagger} \lsb b^a_{\frac{1}{2}-it'} , b_{\frac{1}{2}+i(t_2+t'')}^{b\dagger} b^b_{\frac{1}{2}+it''} \rsb + \lsb b_{\frac{1}{2}+i(t_1+t')}^{a\dagger}  , b_{\frac{1}{2}+i(t_2+t'')}^{b\dagger} b_{\frac{1}{2}+it''}^b \rsb b_{\frac{1}{2}-it'}^a \rb \nn\\
&=&\sum_{t',t''} \lb b_{\frac{1}{2}+i(t_1+t')}^{a\dagger} \lsb b^a_{\frac{1}{2}-it'} , b_{\frac{1}{2}+i(t_2+t'')}^{b\dagger} \rsb b^b_{\frac{1}{2}+it''}  + b_{\frac{1}{2}+i(t_2+t'')}^{b\dagger} \lsb b_{\frac{1}{2}+i(t_1+t')}^{a\dagger}  , b^b_{\frac{1}{2}+it''} \rsb b^a_{\frac{1}{2}-it'} \rb \nn \\
&=&\sum_{t',t''} \lb b_{\frac{1}{2}+i(t_1+t')}^{a\dagger} \eta^{ab}\delta_{t_2+t''+t'} b^b_{\frac{1}{2}+it''}  - b_{\frac{1}{2}+i(t_2+t'')}^{b\dagger} \eta^{ab} \delta_{t_1+t'-t''} b^a_{\frac{1}{2}-it'} \rb \nn\\
&=& \sum_{t''} b_{\frac{1}{2}+i(t_1-t_2-t'')}^{a\dagger} b^a_{\frac{1}{2}+it''}  - \sum_{t'} b_{\frac{1}{2}+i(t_1+t_2+t')}^{a\dagger} b^a_{\frac{1}{2}-it'} \nn\\
&=& L_{t_1-t_2}-L_{t_1+t_2}. \nn
\ea
The computations for the $\LL_t$ and $\GG_t$ commutators are analogous.
\end{proof}

\end{spacing}

\end{document}